\documentclass[a4paper,10pt]{article}

\usepackage{amsmath,amsthm,amssymb,epsfig,graphicx}
\usepackage[latin1]{inputenc}

\usepackage{xspace}
\usepackage{latexsym}
\usepackage{amstext}
\usepackage{amsxtra}

\usepackage{latexsym}

\usepackage{amscd}

\usepackage{amsgen,amstext,amsbsy,amsopn}
\usepackage{math rsfs}

\numberwithin{equation}{section}

\newcommand{\Z}{\mathbb{Z}}
\newcommand{\R}{\mathbb{R}}
\newcommand{\N}{\mathbb{N}}
\newcommand{\NN}{\mathcal{N}}
\newcommand{\C}{\mathbb{C}}

\newcommand{\A}{\mathcal{A}}
\newcommand{\HH}{\mathcal{H}}
\newcommand{\DD}{\mathcal{D}}
\newcommand{\FF}{\mathcal{F}}

\newcommand{\Cc}{\mathcal{C}}

\newcommand{\ep}{\varepsilon}

\newcommand{\Om}{\Omega}
\newcommand{\om}{\omega}

\newtheorem{remark}{Remark}

\newtheorem{theorem}{Theorem}[section]
\newtheorem{proposition}{Proposition}[section]
\newtheorem{corollaire}{Corollary}[section]
\newtheorem{lemma}{Lemma}[section]

\def\XXint#1#2#3{{\setbox0=\hbox{$#1{#2#3}{\int}$}
     \vcenter{\hbox{$#2#3$}}\kern-.5\wd0}}

\DeclareMathAlphabet{\mathonebb}{U}{bbold}{m}{n}
\newcommand{\one}{\ensuremath{\mathonebb{1}}}

\title{The giant vortex state for a Bose-Einstein condensate in a rotating anharmonic trap : extreme rotation regimes}
\author{Nicolas Rougerie \footnote{Universit{\'e} Paris 6, Laboratoire Jacques-Louis Lions, 175 rue du Chevaleret, 75013 Paris, France. email: rougerie@ann.jussieu.fr}}

\begin{document}

\maketitle

\begin{abstract}

We study the fast rotation limit for a Bose-Einstein condensate in a quadratic plus quartic confining potential within the framework of the two dimensional Gross-Pitaevskii energy functional. As the rotation speed tends to infinity with a proper scaling of the other parameters in the model, a linear limit problem appears for which we are able to derive precise energy estimates. We prove that the energy and density asymptotics of the problem can be obtained by minimizing a simplified one dimensional energy functional. In the case of a fixed coupling constant we also prove that a giant vortex state appears. It is an annulus with pure irrotational flow encircling a central low-density hole around which there is a macroscopic phase circulation.

\end{abstract}

\section{Introduction}

\subsection{Physical background}

Since the first experimental achievement of a Bose-Einstein condensate by the Jila and MIT groups in 1995 (2001 Nobel prize in physics attributed to Cornell, Wieman and Ketterle), these systems have been the subject of many studies from the condensed matter community. A reason (among others) for this is the fact that a Bose-Einstein condensate is a good system to study superfluidity issues, such as the existence of quantized vortices. These vortices can be generated by rotation of the container (generally a magnetic trap) enclosing the condensate. They are the subject of an ever increasing number of experimental and theoretical papers, see e.g. the review \cite{Fetter} or the monograph \cite{Aft} for extensive lists of references.\\
Recently, mathematical contributions have studied certain issues arising from the physics of rotating Bose-Einstein condensates. Let us cite some of these issues :
\begin{itemize}
\item the appearance of the first vortices when increasing the rotational speed of the trap in the strong coupling regime \cite{Ign-Mil1,Ign-Mil2};
\item the formation of a vortex lattice characteristic of regimes when the rotational speed nearly deconfines the atoms \cite{AB,AB-2D,ABN};
\item the interaction of many condensates in an optical lattice \cite{AH};
\item the symmetry breaking of the ground state of a condensate in parameter regimes \cite{S};
\item the energy and density asymptotics of strongly interacting condensates in anharmonic traps \cite{CRY1,CRY2,CY} .
\end{itemize}

Most of the available mathematical studies on Bose-Einstein condensates are made in the framework of the Gross-Pitaevskii energy. Here we will use a two-dimensional energy although the actual energy should be three-dimensional (see the discussion in Section 1.4 below). The consensate is described by a complex macroscopic wave-function $\psi$ minimizing the energy given (in the rotating frame) by
\begin{equation}\label{eq0:EGP}
E_{GP}(\psi) = \int _{\R ^2}  \vert \nabla \psi -i\Om  x^{\perp}\psi \vert^2 dx + \int _{\R ^2} \left( V(x) \vert \psi \vert ^2 -\Om^2 |x|^2 |\psi |^2 \right)dx + G \int _{\R ^2} \vert \psi \vert^4 dx,
\end{equation}
 under the mass constraint
\begin{equation}\label{eq0:masse}
\int_{\R^2} | \psi | ^2 dx = 1.
\end{equation}
The matter density at some point $x$ is given by $|\psi|^2(x)$. A vortex is a zero of the wave-function $\psi$ around which there is a quantized phase circulation (ie topological degree or winding number).\\
Here we note $x=(x_1,x_2)$, $x^{\perp}=(-x_2,x_1)$, $V(x)$ is the potential confining the atoms (generally representing a magnetic trap), $\Om$ is the speed at which the trap is rotated around the axis perpendicular to the $x_1-x_2$ plane and $G>0$ is the coupling constant modelling the atomic interactions which we assume to be repulsive. The different terms in the energy represent the kinetic energy (including the effect of Coriolis forces due to the transformation to the rotating frame), the potential energy (including the effect of centrifugal forces) and the energy due to the interatomic interactions.\\
The model based on the Gross-Pitaevskii energy is an approximation of the quantum mechanical many-body problem for $N$ bosons at zero temperature. We refer to \cite{LSY} and references therein for a discussion of the connection between this type of energies and the $N$-body problem. We remark that the Gross-Pitaevskii energy in the rotating case ($\Om \neq 0$) was rigorously derived from the $N$-body problem in \cite{LS}, but only for fixed rotation (ie fixed $\Om$) and fixed interactions (ie fixed $G$). In the regimes we are going to consider ($\Om$ and/or $G$ going to infinity), the derivation of the Gross-Pitaevskii energy from the quantum mechanical many-body problem is still a mostly open problem. To the best of the author's knowledge, the only rigourous results available are those of \cite{BCPY}, where it is proved that the Gross-Pitaevskii energy is a correct first-order description in regimes similar to the ones we are going to study in this paper.\\
We remark that our results provide an a posteriori criterium for the justification of the reduction to the Gross-Pitaevskii energy in some regime of parameters. However, it is far from a rigorous proof, see Section 1.4 below.\\    

In most experimental situations, the potential $V$ is harmonic, of the form
\begin{equation}\label{eq0:potentielharmonique}
V(x)= a_1 x_1 ^2 + a_2 x_2 ^2.
\end{equation}
Such a potential does not allow to take arbitrarily large rotation speeds, as the so-called effective potential  
\begin{equation}\label{eq0:potentieleffectif}
V(x)-\Om^2 |x|^2
\end{equation}
is not bounded below for $\Om^2 > \min(a_1,a_2)$ and consequently neither is the energy (\ref{eq0:EGP}). Physically, this corresponds to the fact that the centrifugal force overcomes the magnetic trapping force when $\Om^2 > \min(a_1,a_2)$, thus the condensate is no longer confined and the atoms fly apart. Fetter suggested \cite{Fetterquart} to use instead a potential with a growth steeper than harmonic, for example of the form
\begin{equation}\label{eq0:potentielquart}
 V(x)= |x|^2 + k|x|^4.
\end{equation}
The nice feature of this potential is that the centrifugal force is always compensated by the trapping force, and thus one can in theory take arbitrarily large rotation speeds. Experiments with this type of potentials have been realized by the ENS group \cite{Exp1,Exp2}, using a blue-detuned laser beam to create the quartic part of the potential. The experiments motivated numerous theoretical studies \cite{AD,Moi,D,FJS,FuZa,JaKa,JaKaLu,KTU,KB,KiFe,Lu} revealing the very rich vortex structure one can expect to be displayed by such systems.
When the rotational speed of a condensate trapped by a potential of this kind is increased from zero, many phase transitions are expected to happen. Firstly, vortices are expected to appear one by one as is the case for a harmonically trapped condensate, but eventually with multiply quantized vortices becoming stable \cite{JaKa,JaKaLu,Lu}, which is never the case with a purely harmonic trap.\\
When the centrifugal force begins to compensate the trapping force corresponding to the quadratic part in the potential, a triangular lattice of singly quantized vortices (Abrikosov lattice) similar to that which is observed in purely quadratic traps appears \cite{Moi}, but a new feature of the quadratic plus quartic trap is the existence of a critical speed for the centrifugal force to create a central hole in the condensate. The resulting state is an annular condensate with a vortex lattice encircling a central giant hole carrying a macroscopic phase circulation \cite{AD,Moi,D,FJS,KB}. At even larger rotation speeds, a new transition is expected to happen, all the individual vortices present in the annulus retreating in the central hole and gathering in a single multiply quantized vortex at the center of the trap. This is what we will refer to as the \emph{giant vortex state}: an annular condensate with pure irrotational flow encircling a central multi-quantized vortex \cite{FJS,FuZa,KTU,KB,KiFe}.\\

In this paper we aim at justifying rigorously the appearance of the giant vortex state in the limit $\Om \rightarrow +\infty$. The existing mathematical results on giant vortices in Bose-Einstein condensates (mainly those of \cite{AAB} to our knowledge) focused on the case when the trapping potential is taken so that the condensate has an annular shape even at slow rotation speeds and studied the effect of the central hole on the vortex structure in the annulus. Here we study the case where, as the rotational speed increases, a central hole is created in the condensate by the centrifugal force. 

\subsection{Model}

Our model is the following: we consider a Gross-Pitaevskii energy of the form
\begin{equation}\label{eq0:modele}
E_{GP}(\psi) = \int _{\R ^2} \left( \vert \nabla \psi -i\Om  x^{\perp}\psi \vert^2 +\right(1-\Om^2 \left) |x|^2 |\psi |^2+ k |x|^4|\psi|^2 + G \vert \psi \vert^4 \right)dx 
\end{equation}
to be minimized under the mass-constraint (\ref{eq0:masse}). In order to study the asymptotics of the problem when $\Om \rightarrow +\infty$, it is more convenient to change scales, setting for $\Om > 1$ 
\begin{equation}\label{eq0:echelle}
u(x)= R\psi \left(R x\right), \quad R=\sqrt{\frac{\Om^2-1}{2k}}. 
\end{equation}
 We then have
\begin{equation}\label{E/F}
 E_{GP}(\psi) = \FF_{\Om}(u)
\end{equation}
with
\begin{equation}\label{eq0:modele2}
\FF_{\Om}(u)= \frac{2k}{\Om^2-1} \int _{\R ^2}  \left| \nabla u -i\Om \frac{\Om^2 -1}{2k}  x^{\perp}u \right|^2 + G \vert u \vert^4 + \frac{\left(\Om^2 -1\right)^3}{8k^2} \left(|x|^4 -2|x|^2 \right)|u|^2     .  
\end{equation}
Let us emphasize that formal calculations \cite{FJS} suggest that the ground state of (\ref{eq0:modele}) should be confined on an annulus of radius $\sim R$ so that our change of scales is quite natural in this setting. Because of the mass constraint 
\begin{equation}\label{eq0:masseu}
\int_{\R ^2} |u|^2dx=1,
\end{equation}
we do not change the physics of the problem by adding to $\FF_{\Om}$ any multiple of $\int_{\R ^2} |u|^2dx$. Hence it is equivalent to minimize
\begin{equation}\label{eq0:modele3}
\tilde{\FF}_{\Om}(u)= \frac{2k}{\Om^2-1} \int _{\R ^2}  \left| \nabla u -i\Om \frac{\Om^2 -1}{2k}  x^{\perp}u \right|^2 + G \vert u \vert^4 + \frac{\left(\Om^2 -1\right)^3}{8k^2} \left(|x|^2-1 \right)^2 |u|^2.
\end{equation}
Defining the parameters
\begin{eqnarray}\label{eq0:parametres}
\om &=& \Om \frac{\Om^2 -1}{2k} \nonumber \\
D_{\Om} &=& \frac{\Om^2-1}{\Om^2} 
\end{eqnarray}
we have
\begin{equation}\label{eq0:lienenergies}
\tilde{\FF}_{\Om}(u)  = \frac{2k}{\Om^2-1} F_{\om}(u)
\end{equation}
with 
\begin{equation}\label{eq0:modele4}
F_{\om}(u)= \int _{\R ^2} \left( \vert \nabla u -i\omega  x^{\perp}u \vert^2 +  D_{\Om}\frac{\om ^ 2}{2}\left(|x|^2-1\right)^2|u|^2 +  G \vert u \vert^4 \right)dx,
\end{equation}
and we are going to study the asymptotics of $F_{\om}$. Results for $\tilde{\FF}_{\Om}(u)$ are straigthforward modifications of those we will present.\\ 
The potential $D_{\Om}\om^2 \left(|x|^2-1 \right)^2$ is positive and has a degenerate minimum for $|x|=1$, so that we expect the condensate to be tightly confined on an annulus centered on the circle $|x|=1$ in the limit $D_{\Om}\om^2 \rightarrow + \infty$. This is reminiscent of semi-classical studies of Hamiltonians with potential wells (see \cite{HeSjo} and the references therein), although our analysis will be quite different.\\

We denote
\begin{equation}\label{eq0:Energiemin}
 I_{\om} = \inf \left\lbrace F_{\om}(u),\: u\in H^1(\R^2)\cap L^2(\R^2,x^4dx) ,\:\int_{\R^2} |u|^2 = 1 \right\rbrace = F_{\om}(u_{\om}).
\end{equation}
It is classical that $I_{\om}$ is achieved by some (a priori non unique) $u_{\om}$ satisfying the Euler-Lagrange equation
\begin{equation}\label{eq0:EELu}
-\left(\nabla -i\om x^{\perp} \right)^2 u_{\om} +D_{\Om}\frac{\om^2}{2} \left(\vert x\vert^2 - 1 \right)^2 u_{\om} + 2 G \vert u_{\om}\vert ^2 u_{\om}=\mu _{\om} u_{\om}, 
\end{equation} 
where $\mu_{\om}$ is the Lagrange multiplier associated with the mass constraint. We aim at studying the asymptotics of $I_{\om}$ and $u_{\om}$ when $\om \rightarrow +\infty$. In particular we want to confirm rigorously the predicted transition of $u_{\om}$ to a state with pure irrotational flow of the form
\begin{equation}\label{eq0:cible}
(r,\theta) \mapsto f(r)e^{in\theta},
\end{equation}
where $r,\theta$ are the polar coordinates and $n$ an integer.\\
More precisely, it is to be expected that $u_{\om}$ will be close to a function of the form (\ref{eq0:cible}) with $n\sim \om $ and $f$ a function tightly confined on a shrinking annulus centered on the potential well. 

\subsection{Main results}

In order to state our main results, let us introduce some notation. For every $n\in \Z$ we define the one dimensional energy $F_n$ 
\begin{equation}\label{eq0:EnergieN1}
 F_n(f)=2\pi \int _ {\R ^+} \left( |f' (r)|^2 +V_n(r)|f(r)|^2 \right)rdr
\end{equation}
defined for every $f\in H^1(\R^+,rdr)\cap L^2(\R^+,rV_n(r)dr)$. The potential $V_n$ is given by:
\begin{equation}\label{eq:Vn1}
 V_n(r)=\frac{n^2}{r^2}-2n\om + \left(1-D_{\Om}\right)\om^2 r^2 + D_{\Om}\frac{\om ^2}{2}+D_{\Om}\frac{\om^2}{2} r^4.
\end{equation} 
The energies $F_n$ appear when decomposing the function $u_{\om}$ as a Fourier series in the angular variable (see (\ref{eq0:decouple}) below). 
We will denote $g_{1,n}$ the ground-state (it is unique up to a multiplicative constant of modulus $1$, that we fix so as to have $g_{1,n}\geq 0$) of the energy (\ref{eq0:EnergieN1}) satisfying the mass constraint
\begin{equation}\label{eq0:masseg}
2\pi \int_{\R ^+} g_{1,n}^2(r)rdr =1 
\end{equation}
and 
\begin{equation}\label{eq0:lambda1n}
 \lambda_{1,n}= F_n(g_{1,n})= \inf \left\lbrace F_n(f), \: f\in H^1(\R^+,rdr)\cap L^2(\R^+,rV_n(r)dr), 2\pi \int_{\R ^+} f^2(r)rdr =1  \right\rbrace. 
\end{equation}
Similarly we introduce
\begin{equation}\label{eq0:1Dnonlineaire}
 E_n(f)=2\pi \int _ {\R ^+} \left( |f' (r)|^2 +V_n(r)|f(r)|^2 +G |f(r)|^4\right)rdr
\end{equation}
with
\begin{equation}\label{eq0:gamman}
 \gamma_{n}= E_n(\Psi_{n})= \inf \left\lbrace E_n(f), \: f\in H^1(\R^+,rdr)\cap L^2(\R^+,rV_n(r)dr), 2\pi \int_{\R ^+} f^2(r)rdr =1  \right\rbrace. 
\end{equation}
Similarly to $g_{1,n}$ we can ensure the unicity of $\Psi_n$ by requiring that $\Psi_n \geq 0$.\\
We will denote $\xi_1$ and $\xi_2$ the first two eigenfunctions normalized in $L^2(\R)$ of the harmonic oscillator
\begin{equation}\label{eq0:OH}
-\frac{d^2}{dx^2} + x^2,  
\end{equation}
associated to the eigenvalues $1$ and $3$. \\

Our results concern the energy (\ref{eq0:modele4}) in the limit $\om \rightarrow +\infty$. They hold when the interaction energy is small compared to the potential and kinetic terms so that the first order of the energy is a quadratic term, leading to a linear equation. We show that this condition is fulfilled if
\begin{equation}\label{eq0:scaling}
\om \rightarrow +\infty, \quad \om \gg G^2 \geq g^2, \quad D_{\Om}\in (0,1), \quad D_{\Om} \rightarrow D \in (0,1] 
\end{equation}
where $D$ and $g$ are fixed constants. The assumption that $G^2$ is larger than some constant is only a matter of convenience for writing the results. The opposite case where $G\ll 1$ can be dealt with using the tools that we develop in this paper, and is actually easier than the case where $G$ stays bounded below. We will refer to the limit (\ref{eq0:scaling}) as the \emph{extreme rotation regime}.\\


We have the following result (the notation $a \propto b$ has the usual meaning that $a/b$ converges to some constant.):
\begin{theorem}[Energy and density asymptotics in the extreme rotation regime]\label{theo:densite/energie} \mbox{}\\
Let $u_{\om}$ be any solution of (\ref{eq0:Energiemin}) and suppose that (\ref{eq0:scaling}) holds. There is an integer $n^* = \om + O(1)$ so that 
\begin{equation}\label{eq0:resultenergie}
I_{\om}=F_{\om}(u_{\om})=\lambda_{1,n^*} + 2\pi G \int_{\R^+} g_{1,n^*}(r) ^4 rdr + O(G^2)=\gamma_{n^*}+O(G^2).
\end{equation}
Moreover 
\begin{equation}\label{eq0:resultdensite}
\Vert |u_{\om}|^2 - g_{1,n^*}^2 \Vert_{L^2(\R^2)} \leq C G^{1/2} \ll \Vert g_{1,n^*} ^2 \Vert_{L^2(\R^2)}  \propto \om ^{1/4}
\end{equation}
and $|u_{\om}|^2$ converges to a Dirac delta function at $|x|=1$ in the weak sense of measures.
\end{theorem}

Note that $|g_{1,n^*}|^2$ converges to a Dirac delta function at $|x|= 1$ as a consequence of Theorem \ref{theo:1D1} below, to which we refer for more details on $\lambda_{1,n}$ and $g_{1,n}$. In particular we have 
\[
 \lambda_{1,n^*} \sim \left(2D_{\Om} + 4\right)^{1/2}\om 
\]
and 
\[
2\pi \int_{\R^+}g_{1,n^*}(r) ^4 \sim \frac{\left(2D_{\Om} + 4 \right)^{1/4}}{\sqrt{2}\pi} \om ^{1/2}\int_{\R}\xi_1 ^4 (x) dx,
\]
thus Theorem \ref{theo:densite/energie} says that the energy of the condensate is described to subleading order by a simplified one-dimensional problem, which also gives the density profile.
The second equality in (\ref{eq0:resultenergie}) holds because the energies $F_n$ and $E_n$ are actually very close in the parameter range (\ref{eq0:scaling}), see Theorems \ref{theo:1D1} and \ref{theo:1D1NL} below. Similarly (\ref{eq0:resultdensite}) could be stated with $\Psi_{n^*}$ replacing $g_{1,n^*}$.\\
The energy $E_n$ is obtained by restricting $F_{\om}$ to wave functions of the form $f(r)e^{in\theta}$, so Theorem \ref{theo:densite/energie} is a first step towards the understanding of the giant vortex. However, a result such as (\ref{eq0:resultdensite}) is not sufficient to conclude that there are no vortices in the bulk of the condensate. Actually our method does not give precise enough energy estimates to do so in the whole regime (\ref{eq0:scaling}).\\   

A special case where we can prove much more detailed results is that of a fixed coupling constant: 
\begin{theorem}[Refined asymptotics in the case of fixed $G$]\label{theo:densite/energie2}\mbox{}\\
Let $u_{\om}$ be any solution of (\ref{eq0:Energiemin}) and suppose that there holds
\begin{equation}\label{eq0:Gfixe}
\om \rightarrow +\infty, \quad G\mbox{ is a fixed constant},  \quad D_{\Om}\in (0,1), \quad D_{\Om} \rightarrow D \in (0,1].
\end{equation}
Then there is an integer $n^* = \om + O(1)$ so that for any $\ep>0$ 
\begin{equation}\label{eq0:resultenergie2}
I_{\om}=F_{\om}(u_{\om})=\gamma_{n^*} + O(\om^{-1/4+\ep}).
\end{equation}
Moreover there exists $\alpha \in \R$ such that, along some subsequence 
\begin{equation}\label{eq0:unseulmode}
 \Vert u_{\om} -  e^{i\alpha}\Psi_{n^*} e^{in^* \theta }\Vert_{L^2(\R^2)} \leq C_{\ep} \om^{-1/16+\ep}.
\end{equation}
\end{theorem}

Theorem \ref{theo:densite/energie2} is still mainly concerned with energy and density asymptotics, although (\ref{eq0:unseulmode}) allows to identify a global limiting phase. We now state our result about the appearance of the giant vortex, that is an annular condensate with no vortices in the bulk : 
 
\begin{theorem}[The giant vortex state in the case of fixed $G$]\label{theo:vortex}\mbox{}\\
Under the assumptions of Theorem \ref{theo:densite/energie2}
\begin{enumerate}
\item There is a constant $\sigma$ so that the following estimate holds true pointwise along some subsequence
\begin{equation}\label{eq0:estimunif}
 \left| |u_{\om}| -  |\Psi_{n^*}|\right| \leq C_{\ep} \om^{3/16+\ep} e^{-\sigma \om \left(|x|-1\right)^2} + C_{\ep} \om^{-1/32+\ep}.
\end{equation}
In particular
\begin{equation}\label{eq0:estimunif2}
\left\Vert |u_{\om}| -  |\Psi_{n^*}| \right\Vert_{L^{\infty}(\R^2)} \leq C_{\ep} \om^{3/16+\ep} \ll \left\Vert \Psi_{n^*} \right\Vert_{L^{\infty}(\R^2)} \propto \om^{1/4}.
\end{equation}

\item  $|u_{\om}(x)| \rightarrow 0$ uniformly on 
\begin{equation}\label{eq4:Anneau}
\HH_{\om}^{\delta}=\left\lbrace x\in \R^2,\: \left| \vert x \vert-1 \right|^2 \geq \delta \om ^{-1}\ln (\om) \right\rbrace 
\end{equation}
for every $\delta > \frac{1}{4\sigma}$.
\item Suppose that there is some $x_{\om}$ so that
\begin{equation}
u_{\om}(x_{\om}) \rightarrow 0
\end{equation}
then, if $\om$ is large enough, we have for every $\delta < \frac{1}{4\sigma}$
\begin{equation}\label{eq4:pasvortex}
 \left||x_{\om}|-1\right|^2 \geq  \delta \om^{-1} \ln (\om).
\end{equation}
\end{enumerate}
\end{theorem}

We remark that whereas there can be a great number (proportional to $G$) of integers $n^*$ satisfying (\ref{eq0:resultenergie}) and (\ref{eq0:resultdensite}), there are at most two integers satisfying (\ref{eq0:resultenergie2}) and (\ref{eq0:unseulmode}). More precisely, two cases can occur : either $\gamma_n$ is minimized at a single $\hat{n}$ and the whole sequence $u_{\om}$ satisfies (\ref{eq0:unseulmode}) and (\ref{eq0:estimunif}) with $n^* = \hat{n}$ or $\gamma_n$ is minimized at two integers $\hat{n}$ and $\hat{n}+1$. This case is very particular and happens only for special values of the parameters. For these particular values of the parameters, the sequence $u_{\om}$ could in principle oscillate indefinitely between $\Psi_{\hat{n}} e^{i\hat{n} \theta }$ and $\Psi_{\hat{n}+1} e^{i(\hat{n} +1)\theta }$, so that we have to extract a subsequence to state the result.\\

We will discuss these results in further details in Section 1.4 and now present some ideas of the proofs. It is natural (see \cite{FJS} for example) to construct a first trial function of the form
\begin{eqnarray}\label{eq0:ftest1}
 v_{test}(r,\theta)=f_{test}(r) e^{i\hat{n} \theta}\\
\nonumber f_{test}(r)=c_{test} \xi \left( \frac{r-1}{\eta} \right) \nonumber 
\end{eqnarray}
where $\xi: \R \rightarrow \R$ is chosen with $\int_{\R} \xi ^2 = 1$, $\hat{n}=[\om]$ is the closest integer to $\om $ and $\eta$ a small parameter to be chosen later on. We want the mass constraint (\ref{eq0:masseu}) to be satisfied, which implies 
 \begin{equation}\label{eq0:ctest1}
c_{test}^2\sim ( \sqrt{2} \pi \eta )^{-1}  
 \end{equation}
when $\eta \rightarrow 0$ under some conditions on $\xi$, see Remark \ref{Rq1} in Section 2 below.\\
Expanding the potential $D_{\Om}\om^2 \left(\vert x\vert^2 - 1 \right)^2$ allows to compute the energy of $v_{test}$. Minimizing it with respect to $\eta$ gives the main scales of the problem: taking
\begin{equation}
\eta = \left(2D_{\Om} +4 \right)^{-1/4}\om^{-1/2} 
\end{equation}
we get
\begin{equation}\label{eq0:energietest}
F_{\om}(u_{\om})\leq F_{\om}(v_{test}) \leq \left(2D_{\Om} +4 \right)^{1/2}\om \int_ {\R}  \left( \xi '(x) ^2 + x^2 \xi ^2(x) \right)dx +O(G\om^{1/2}) 
\end{equation}
so that when $\om \gg G^{2}$ the natural choice is to take for $\xi$ the Gaussian $\xi_1$ achieving the infimum
\begin{equation}
1 = \inf _{\int _{\R} \xi ^2 = 1} \int _{\R} \left( \xi '(x) ^2 +  x^2 \xi ^2(x)  \right)dx.
\end{equation}
The symmetry of $\xi_1$ allows to improve the remainder term in (\ref{eq0:energietest}):
\begin{equation}\label{eq0:energietest'}
F_{\om}(u_{\om})\leq F_{\om}(v_{test})  \leq \left(2D_{\Om} +4 \right)^{1/2} \om + G\om^{1/2}\frac{\left(2D_{\Om}+4\right)^{1/4}}{\sqrt{2}\pi}\int_{\R} \xi_1 ^4 (x)dx + O(1).
\end{equation}
This computation suggests that the dominant terms in the energy will be the kinetic and potential contributions, whose sum will be of the order of $\om$. The interaction energy will appear only at second order, with a contribution proportional to $G\om^{1/2}$, which is much smaller than $\om$ in the extreme rotation regime.\\  
Thus, it is natural to expand $u_{\om}$ in a Fourier series in the angular variable
\begin{equation}\label{eq0:Fourier}
 u_{\om}(r,\theta)= \sum_{n\in \Z} f_n(r) e^{in\theta}.
\end{equation}
Proving that $u_{\om}$ converges to a function of the form (\ref{eq0:cible}) then amounts to showing that $f_n$ converges to $0$ for all $n$ except one, and the crucial fact will be the following decoupling of the quadratic part of the energy (\ref{eq0:modele3}):
\begin{equation}\label{eq0:decouple}
F_{\om}(u_{\om})= \sum_{n\in \Z} F_n(f_n) + G \int_{\R^2} |u_{\om}|^4,
\end{equation}
with the energies $F_n$ being defined by (\ref{eq0:EnergieN1}). The proof of Theorem \ref{theo:densite/energie} is based on the detailed study of the ground states of these energies. In particular, for most values of the parameters,  $\lambda_{1,n}$ is minimized at a unique integer, $n^*$.\\
We prove that because of the properties of the potentials $V_n$, some modes $f_n$ carry too much energy to match our upper bound (\ref{eq0:energietest'}) and therefore can be eliminated. For the other ones, we use a lower bound on the gap between the ground-state energy of (\ref{eq0:EnergieN1}) and the first excited level. This gap is of order $\om$, which is much larger than the interaction energy if $\om \gg G^2$, as shown by the preceding computation. We use this fact to prove that $f_n\propto g_{1,n}$ when $\om$ goes to infinity.\\
This very particular form allows one to bound the quartic part of the energy from below. Combining the lower bound with the analysis of the one dimensional energies gives an estimate of the total energy contribution of each mode. Two effects come into play: concentration of the mass of $u_{\om}$ on the mode $f_{n^*}\propto g_{1,n^*}$ is favorable for the quadratic part of the energy because $\lambda_{1,n^*}$ is the minimum of $\lambda_{1,n}$ with respect to $n$, but it increases the quartic part. We then show, confirming the energy scaling of (\ref{eq0:energietest'}) that, in order to match the subleading order of the energy, it is more favorable for $u_{\om}$ to have its mass concentrated on only one mode. \\
However, the remainder terms that we obtain are not small enough to confirm this rigorously in the whole range of parameters (\ref{eq0:scaling}). The energy expansion (\ref{eq0:resultenergie}) still allows for many modes (a number proportional to $G$) not to be asymptotically $0$ in the limit. We get the density asymptotics (\ref{eq0:resultdensite}) nevertheless because the moduli of these modes are very close to one another.\\
If we fix $G$ we can prove (\ref{eq0:unseulmode}), which is much more precise than (\ref{eq0:resultdensite}) and requires a refinement of the method described above. In particular, using the Euler-Lagrange equation for $u_{\om}$ and the result (\ref{eq0:resultdensite}) we refine the estimate $f_n\propto g_{1,n}$, finding a second order correction. This in turn allows to refine the remainder terms in the energy expansions and prove (\ref{eq0:unseulmode}).\\
The dominant part of the Euler-Lagrange equation (terms corresponding to the variations of the kinetic and potential energies) is linear in the situation we consider. We are thus able, using elliptic estimates for the Ginzburg-Landau operator due to Lu and Pan \cite{LuPan}, to get estimates in stronger norms. The uniform estimate (\ref{eq0:estimunif}) is obtained by this method. It implies that $u_{\om}$ is confined to a shrinking annulus in which there are no zeroes, because so is $\Psi_{n^* }$ (see Theorem \ref{theo:1D1NL}).\\
We remark that a Fourier expansion such as (\ref{eq0:Fourier}) has been used in \cite{BPT} to compute the lowest eigenvalue of the Ginzburg-Landau operator on a disc with an applied field going to infinity, but with no potential term in the energy. In that paper the minimizer gets confined close to the boundary of the disc, whereas in our work it is the potential that forces the confinement to a shrinking annulus. In \cite{BPT}, a rescaling similar to (\ref{eq0:ftest1}) leads to an energy expansion in powers of the small characteristic length of the problem. We perform a similar analysis in Section 2, with the difference that our original problem is posed on the whole space $\R^2$, leading to rescaled eigenvalue problems on a line whereas in \cite{BPT} a compact domain is considered, so that after rescaling one gets eigenvalue problems on a half line.\\

\subsection{Discussion}

Theorems \ref{theo:densite/energie2} and \ref{theo:vortex} confirm the expected qualitative features at high rotation speeds : the condensate is tightly confined on a shrinking annulus in which there are no vortices. We identify a simplified limiting profile with a Gaussian shape (see Theorems \ref{theo:1D1} and \ref{theo:1D1NL} for details on the functions $g_{1,n^*}$ and $\Psi_{n^*}$ respectively). We are also able to identify a global limiting phase and prove that all vortices gather in the central low-density hole, confirming that a giant vortex state appears when $\om$ goes to infinity with $G$ fixed.\\

We prove that the width of the annulus where the mass is concentrated is of order $ \om^{-1/2} \sim \Om^{-3/2} k ^{1/2}$ in the extreme rotation limit. This is larger than the width predicted in former physical studies \cite{FJS,FuZa,KiFe}, and our limiting density is more regular than the Thomas-Fermi type profiles they suggest. This is due to the fact that these studies are made assuming that the interaction energy is the dominant part of the energy, leading to a profile solution of a nonlinear equation. Here we study a different regime of parameters and in particular we show that the Thomas-Fermi approximation for strong interactions breaks down in the extreme rotation limit, where $u_{\om}$ converges to the solution of some linear reduced problem.\\

We used a two-dimensional Gross-Pitaevskii energy instead of the complete 3D energy. A rather natural question would then be to find in which situations the reduction to the two-dimensional model is justified. One can always assume that the condensate is tightly confined in the direction of the rotation axis, so that the problem is essentially 2D, but in such fast rotation regimes it is to be expected that the reduction is valid independently of the strength of the vertical confinement (see \cite{AB-2D} where this is shown for the case of a fast rotating condensate in a harmonic trap).\\ 

The extreme rotation regime can be reached in two ways : either by letting $\Om$ tend to infinity with
\begin{equation}\label{eq0:scaling0}
 \Om \rightarrow +\infty, \quad \Om^3 \gg k, \quad \Om \gg G^{2/3}k.
\end{equation}
or by keeping $\Om$ of the order of a constant (strictly larger than $1$) and letting $k$ tend to zero with 
\begin{equation}\label{eq0:scaling01}
 \Om \rightarrow \Om_0 > 1, \quad k\rightarrow 0,\quad  k \ll G ^{-2}
\end{equation}
where $\Om_0$ is a constant.\\
In a recent series of papers \cite{CRY1,CRY2,CY}, the large $\Om$ limit has been studied for traps related to the one we study (homogeneous traps, $V(x)=|x|^s$ with $s>2$, and flat trap $V(x)=0$ for $|x|\leq 1$ and $V(x)=+\infty$ for $|x|>1$). Their analysis applies to our problem in regimes ``opposite'' to the two described above, namely under the condition that either $\Om \ll G^{2/3}k$ (for the regime opposite to (\ref{eq0:scaling0})) or with $k \gg G ^{-2}$ (for the regime opposite to (\ref{eq0:scaling01})). In this kind of regimes, the energy and density asymptotics are given to leading order by a limiting problem of Thomas-Fermi type. The limiting energy depends on the density only and is obtained from (\ref{eq0:modele4}) by neglecting the kinetic (first) term.\\
If $\Om_0$ is allowed to be $1$ in (\ref{eq0:scaling01}), we have $D_{\Om}\rightarrow 0$ and so the potential term in the energy becomes negligible compared to the kinetic one. A solution $u_{\om}$ should then converge to its projection on the first eigenspace of the Ginzburg-Landau operator $-\left(\nabla - i\om x^{\perp}\right)^2$. This space is called the lowest Landau level and has been widely used for the study of rapidly rotating Bose-Einstein condensates in harmonic traps (see \cite{ABN} and references therein). For the case $\Om_0 =1$ we refer to \cite{Moi} where we have studied a regime of this kind in the lowest Landau level.\\

The main drawback of the method we present here is that it confirms the appearance of the giant vortex state but does not give the critical speed at which the transition would be expected to happen. Finding such a critical speed is an important issue for the condensed matter community (see for example \cite{FuZa,KiFe}). For simplicity and comparison with other results we will consider $k$ as fixed in the following discussion.\\
At rotation speeds much smaller than those we have considered, the prefered state is known \cite{CY} to be a vortex lattice encircling a central low-density hole. The transition between this state and the giant vortex is not expected to happen in the regime $\Om \sim G^{2/3}$. Indeed, the results in \cite{CY} (see also \cite{FB}) allow to conjecture that in the case of a flat trap ($V(x)=0$ for $|x|\leq 1$ and $V(x)=+\infty$ for $|x|>1$) the transition should happen when $\Om \sim G(\log G)^{-1}$. Combined with the scalings in \cite{CRY2} this suggests that for our trap (\ref{eq0:potentielquart}) the vortices will disappear from the condensate when $\Om \sim G^{2/3}(\log G)^{-1}$, before the regime we have studied here is reached.\\  
In a forthcoming work \cite{Nous} we study, in the case of a flat trap, the transition regime $\Om \sim G(\log G)^{-1}$ (corresponding to $\Om \sim G^{2/3}(\log G)^{-1}$  for the trap (\ref{eq0:potentielquart})) using tools from the Ginzburg-Landau theory \cite{BBH,San-Ser}. This completes the analysis in \cite{CY} where it is proved that the vorticity is uniform in the bulk for $\log G \ll \Om \ll G(\log G)^{-1}$ and bridges between the situations considered in \cite{CRY1,CRY2,CY} and the present paper. We do believe that statements such as Theorem \ref{theo:densite/energie2} and \ref{theo:vortex} are valid in the whole regime (\ref{eq0:scaling}), but the proofs should rely on the tools used in \cite{Nous} that allow for a better understanding of the energetic cost of the vortices.\\ 

\bigskip

We make one last remark before turning to the proofs of our results, regarding the validity of the Gross-Pitaevskii description. It is expected to be valid when the number of particles is much larger than the number of occupied states. In that case a significant fraction of these particles must occupy the same energy state, which is the phenomenon of Bose-Einstein condensation. In Section 2 we analyze in some detail the eigenstates of the single-particle Hamiltonian (corresponding to the linear part in (\ref{eq0:EELu})). In particular, we show (Corollary \ref{cor:lambdan/n^*}) that the energy splitting between two neighboring single particle states is bounded below by a constant. On the other hand, the interaction energy per particle (remark that we have taken a unit mass constraint, so that $F_{\om}$ actually represents the energy per particle of the condensate, and $G\int_{\R^2} |u_{\om}|^4$ the interaction energy per particle) is of order $G\om^{1/2}$ as a consequence of Theorem \ref{theo:densite/energie}. The number of occupied states can be computed by dividing the interaction energy by the energy splitting between two neighboring single particle states, so that one should expect the Gross-Pitaevskii description to be valid when
\[
 G\om^{1/2}\ll N
\]
where $N$ is the number of particles.\\
This is interesting because the Gross-Pitaevskii description for rotating bosons is known to break down in some regimes of fast rotation, even at zero temperature. It is the case for a condensate trapped by an harmonic potential of the kind (\ref{eq0:potentielharmonique}) when the rotation speed is too close (in some sense related to the number of particles considered) to the limit value $\min(a_1,a_2)$. In that situation, one must go back to the Hamiltonian for $N$ bosons and the system exhibits strongly correlated states (fractional quantum Hall effect). We refer to \cite{LeSe} and references therein for details on this phenomenon.\\
The fractional quantum Hall effect happens in a regime where the minimization is restricted to the lowest Landau level. The corresponding situation in our setting would be a regime like (\ref{eq0:scaling01}) but with $\Om_0 = 1$, which we do not allow in this paper. We are thus not in a regime where one should expect to get strongly correlated states.\\

\bigskip

The article is organized as follows: in Section 2 we first gather the main notation that is to be used in the paper, then we study the one-dimensional energies (\ref{eq0:EnergieN1}), providing the lower bounds on their ground states and first excited levels that we use in our analysis. In Section 3 we show how the results of Section 2 allow to bound the quartic term in (\ref{eq0:modele3}) from below and prove Theorem \ref{theo:densite/energie}. The proofs of Theorems \ref{theo:densite/energie2} and \ref{theo:vortex} are then presented in Sections 4, with each main step occupying a subsection. 

\section{The one dimensional energies}

We recall the definition of the one dimensional energies we are interested in:
\begin{equation}\label{eq0:EnergieN}
 F_n(f)=2\pi \int _ {\R ^+} \left( |f' (r)|^2 +V_n(r)|f(r)|^2 \right)rdr
\end{equation}
for every $f\in H^1(\R^+,rdr)\cap L^2(\R^+,rV_n(r)dr)$. The potential $V_n$ is given by:
\begin{eqnarray}\label{eq:Vn}
V_n(r)&=& V(r) +\left|\frac{n}{r} -\om r\right|^2 \\
&=&\frac{n^2}{r^2}-2n\om + \left(1-D_{\Om}\right)\om^2 r^2 + D_{\Om}\frac{\om ^2}{2}+D_{\Om}\frac{\om^2}{2} r^4.\nonumber
\end{eqnarray} 
Where 
\begin{equation}\label{eq1:Voriginal}
V(r) = D_{\Om}\frac{\om ^2}{2} \left(r^2-1\right)^2
\end{equation}
is the potential in the energy (\ref{eq0:modele4}).\\
A straigthforward computation shows that $V_n$ has a unique minimum point for $r>0$, which we shall note $R_n$. It is uniquely defined by the equation
\begin{equation}\label{eq0:rn*}
R_n ^6 +\frac{1-D_{\Om}}{D_{\Om}}R_n ^4 = \frac{n^2}{D_\Om \om ^2} 
\end{equation}
and it is interesting to note that if $n=\om$ then $R_n = 1$, which is the minimum point of the original potential $V$.\\ 
It is classical to show that there is an increasing sequence $\left(\lambda_{j,n}\right)_{j=1\cdots+\infty}$ of eigenvalues for the operator associated to $F_n$. The corresponding sequence of normalized eigenfunctions $\left(g_{j,n}\right)_{j=1\cdots+\infty}$ is a Hilbert basis for the space $H^1(\R^+,rdr)\cap L^2(\R^+,rV_n(r)dr)$. Considering the one-dimensional functions $g_{j,n}$ as radial functions defined on $\R ^2$, one has
\begin{equation}\label{EElgjn}
 -\Delta g_{j,n} + V_n g_{j,n} = \lambda_{j,n} g_{j,n}
\end{equation}
and 
\begin{equation}\label{eq0:normalisation}
\int_{\R ^2} g_{j,n}^2(x) dx = 2\pi \int_{\R ^+} g_{j,n}^2(r)rdr =1.
\end{equation}
We want to study the asymptotics of the first two eigenvalues $\lambda_{1,n}$ and $\lambda_{2,n}$ of the operator associated to $F_n$. 
We remark that $g_{1,n}$ is unique up to a multiplicative constant of modulus $1$, that we fix so as to have 
\begin{equation}\label{g1npositive}
g_{1,n}\geq 0.
\end{equation}
$\lambda_{1,n}$ satisfies
\begin{eqnarray}\label{eq1:lambda1n}
\lambda_{1,n}&=& F_n(g_{1,n}) \nonumber\\
&=& \inf \left\lbrace F_n(f), 2\pi \int_{\R ^+} f^2(r)rdr =1  \right\rbrace
\end{eqnarray}
and $\lambda_{2,n}$ 
\begin{eqnarray}\label{eq1:lambda2n}
\lambda_{2,n}&=& F_n(g_{2,n}) \nonumber \\
&=& \inf \left\lbrace F_n(f), 2\pi \int_{\R ^+} f^2(r)rdr =1, \int_{\R^+} f(r)g_{1,n}(r)rdr = 0  \right\rbrace.
\end{eqnarray}
Note that $g_{2,n}$ is not uniquely defined a priori. 

In fact it is not necessary to carry out the analysis for every $n$. Indeed, we have the following result :

\begin{lemma}\label{theo:lemme1}
Let $\lambda_{1,n}$ be defined by (\ref{eq1:lambda1n}). There are constants $a,b>0$ and $c>\sqrt{6}$ so that for
\[
 \vert n-\om \vert > a\:\om^{1/2}
\]
or equivalently
\[
\vert R_n - 1\vert > b\:\om^{-1/2}
\]
we have 
\begin{equation}\label{eq0:premiersmodesborneinf}
\lambda_{1,n} > c\: \om 
\end{equation}
\end{lemma}

\begin{proof}
For $|R_n-1|>c_1\om^{-1/2}$, one has 
\begin{equation*}
\lambda_{1,n} \geq 2\pi \int_{\R^+} V_n (r)|g_{1,n}(r)|^2 rdr \geq 2\pi V_n (R_n)\int_{\R^+} |g_{1,n}(r)|^2 rdr = V_n(R_n) \geq V (R_n) > c_2\: \om  
\end{equation*}
for two constants $c_1,c_2$.\\
On the other hand one can see from formula (\ref{eq0:rn*}) that $\om \gg |n- \om|$ is equivalent to $1 \gg |R_n-1|$. For such values of $n$ and $R_n$ we have $|R_n-1| \sim C \om^{-1} |n-\om|$, for a certain constant $C$, thus $\om \gg |n- \om| > c_3\om^{1/2}$ is equivalent to $1 \gg |R_n-1|>c_2\om^{-1/2}$. We then note that the function mapping $R_n$ to $n$ is strictly increasing and has a strictly increasing inverse, which allows to deduce that $|n- \om| > c_3\om^{1/2}$ is equivalent to $|R_n-1|>c_2\om^{-1/2}$ for any value of $n$ and $R_n$. There only remains to tune the constants $c_1,c_2,c_3$ to conclude the proof.
\end{proof}

We will denote
\begin{equation}\label{eq0:N3/2}
\mathcal{N}_{1/2}= \left\lbrace n\in \Z ,\quad \left| n-\om \right| \leq a\:\om^{1/2} \right\rbrace 
\end{equation}
and study in detail the ground state and first excited level of the energy (\ref{eq0:EnergieN}) under the mass constraint (\ref{eq0:normalisation}) for $n\in \mathcal{N}_{1/2}$. For $n\notin \NN_{1/2}$, the simple lower bound (\ref{eq0:premiersmodesborneinf})
matched with the upper bound (\ref{eq0:energietest'}) will be enough for our purpose (note that $\left(2D_{\Om} +4 \right)^{1/2}\leq \sqrt{6}$).\\
We now describe the main results concerning the modes $n\in \NN_{1/2}$. As we will see $g_{1,n}$ and $g_{2,n}$ are asymptotically confined on a domain centered on $R_n$ which size is of the order of $\om^{-1/2}$. More precisely we denote
\begin{equation}\label{eq0:h}
 h_n= \left( \frac{2}{V_n '' (R_n)}\right)^{1/4}.
\end{equation}
From (\ref{eq:Vn}) and (\ref{eq0:rn*}) it is straigthforward to see that there is constant $h$ so that for any $n\in \NN_{1/2}$ 
\begin{equation}\label{eq0:h2}
 h_n = h\om^{-1/2}\left( 1 + o(1) \right).
\end{equation}
We define the rescaled functions $\xi_{i,n}$, $i\in\left\lbrace 1,2\right\rbrace $ by
\begin{equation}\label{eq0:rescaling}
 \xi _{i,n}(x)= c_{i,n}^{-1} g_{i,n}\left( R_n + h_n x\right)
\end{equation}
where $\Vert \xi_{i,n}\Vert_{L^2}=1$ and
\begin{equation}\label{eq1:cin}
c_{i,n}^2  =\frac{1}{2\pi h_n R_n + 2\pi h_n^2\int_{-\frac{R_n}{h_n}}^{+\infty} x\xi_{1,n}(x)^2dx}\sim \frac{1}{2\pi h_n R_n}.
\end{equation} 

\begin{remark}\label{Rq1}
In order for the asymptotics in equations (\ref{eq0:ctest1}) and (\ref{eq1:cin}) to be justified, we need that $\xi$ and $\xi_{1,n}$ decay faster than any polynomial when $x$ goes to infinity. A typical example is to take for $\xi$ a Gaussian, or more generally some function with an exponential decay. Very often in the rest of the paper we will have to deal with quantities of the form
\[
 \int ^{+\infty}_{-d_{\om}} F(x)dx
\]
with $F:\R \rightarrow \R$  an exponentially decreasing function ($F(x)=x^k e^{-Cx^2}, \quad k\in \N$ for example), and $d_{\om}= O(\om ^{\kappa}), \: \kappa >0$. We always replace such quantities by the corresponding 
\[
 \int ^{+\infty}_{-\infty} F(x)dx
\]
which is justified because the decay rate of $F$ guarantees that the difference between the two will be exponentially small as $\om$ goes to infinity, whereas all the other remainder terms that we will encounter will behave like powers of $\om$.\\
Note that the decay of $\xi_{1,n}$ is proved in Proposition (\ref{theo:decexpo1d}). 
\end{remark}

\textbf{Notation.} In the rest of the paper, we will always mean by $O(\om^k)$ a quantity bounded by $\om^k$ uniformly with respect to $D_{\Om}$ and $n$. $C$ will be a generic positive constant depending neither on $\om$ or $D_{\Om}$ nor on $n$. When writing that a function is a $O_{L^2}(\om^k)$ for example we mean that the function's $L^2$ norm is bounded by $\om^k$ uniformly with respect to $D_{\Om}$ and $n$. For the rescaled functions of the type (\ref{eq0:rescaling}) we will often use the notation $O_{HO}(\om^k)$, the $HO$ symbol meaning that the $O(\om^k)$ is taken in the norm associated with the harmonic oscillator (\ref{eq0:OH}), namely the norm in $H^1(\R)\cap L^2(\R,x^2dx)$.\\
We recall that the notation $a \propto b$ has the usual meaning that $a/b$ converges to some constant.\\ 

From now on and in the rest of this section, we always implicitly consider only those $n$ that belong to $\NN_{1/2}$, which have the following important properties deduced from (\ref{eq:Vn}) and (\ref{eq0:rn*}): 
\begin{itemize}
\item $|n-\om|\leq a \om^{1/2}$
\item $R_n$ and $\frac{1}{R_n}$ are bounded uniformly with respect to $\om$ and $D_{\Om}$, more precisely $|R_n-1|\leq b \om^{-1/2}$. Moreover we have 
\begin{equation}\label{Rp-Rq}
 |R_p - R_q | \propto \om^{-1} |p-q| \mbox{ for any }p,q\in \NN_{1/2}.
\end{equation}
\item $h_n\propto \om^{-1/2}$. Moreover we have 
\begin{equation}\label{hp-hq}
 |h_p - h_q | \propto \om^{-3/2} |p-q|\mbox{ for any }p,q\in \NN_{1/2}.
\end{equation}
\item $V^{(k)}_n (R_n) \propto \om^2$ for any $k\geq 2$.
\end{itemize}

\bigskip

We are able, using the localization property of $g_{1,n}$ and $g_{2,n}$, to expand the corresponding eigenvalues in powers of $h_n$. We have the followings results:\\

\begin{theorem}[Asymptotics for the ground states of the one-dimensional linear problems] \label{theo:1D1} \mbox{}\\
Let $\lambda_{1,n}$ be defined by (\ref{eq1:lambda1n}).\\ 
Suppose that $n\in \NN _{1/2}$ as defined in equation (\ref{eq0:N3/2}) and let $R_n$ be defined in equation (\ref{eq0:rn*}).\\  
We have as $\om \rightarrow \infty$ and $D_{\Om}\rightarrow D$
\begin{equation}\label{eq0:DLenergie1}
\lambda_{1,n}= V_n(R_n) + \sqrt{\frac{V_n '' (R_n)}{2}} + K_n \\
\end{equation}
where $K_n=O(1)$. Let $\xi_1$ be the normalized ground state of the harmonic oscillator (\ref{eq0:OH}) and  $\xi_{1,n}$ be defined by equation (\ref{eq0:rescaling}). We have
\begin{equation}\label{eq0:DLfonction}
\xi_{1,n}=\xi_1 + h_n P_n \xi_1 + h_n^2 Q_n \xi_1 +O_{HO}(\om^{-3/2})\\
\end{equation}
where $h_n$ is defined in equation (\ref{eq0:h}), $P_n$ and $Q_n$ are two polynomials whose coefficients are bounded independently of $\om$ and $D_{\Om}$ and depend continuously on $n$ in the following sense:
\begin{eqnarray}\label{DLxi2}
P_p \xi_1  & =& P_q \xi_1  + O_{HO}(\om^{-1}|p-q|)\\
Q_p \xi_1  & =& Q_q \xi_1  + O_{HO}(\om^{-1}|p-q|)\nonumber
\end{eqnarray}
for any $p,q\in \NN_{1/2}$.
\end{theorem} 

Note that 
\[
 V_n(R_n)+\sqrt{\frac{V_n '' (R_n)}{2}}= O(\om).
\]

From this theorem we deduce as a corollary 
\begin{corollaire} [Variations of $\lambda_{1,n}$ with respect to $n$] \label{cor:lambdan/n^*}\mbox{}\\
Under the assumptions of Theorem \ref{theo:1D1} and for $\om$ large enough, $\lambda_{1,n}$ is minimized for at most two integers $n^*= \om + O(1)$ and $n^*+1$. For any $n\neq n^*,\: n^*+1$ one has
\begin{equation}\label{eq1:dependanceN}
\lambda_{1,n}= \lambda_{1,n^*} + C  \left(n-n^*\right)^2\left(1+ O(\om^{-1/2})\right). 
\end{equation}
\end{corollaire}

Equation (\ref{eq1:dependanceN}) points out the dependence of $\lambda_{1,n}$ with respect to $n$, which is the key point to prove Theorem \ref{theo:densite/energie}. We now state the result giving the gap between $\lambda_{2,n}$ and $\lambda_{1,n}$:\\

\begin{proposition}[Lower bound on the first excited levels of the one-dimensional energies] \label{theo:1D2} \mbox{}\\
Let $\lambda_{2,n}$ be defined by (\ref{eq1:lambda2n}).\\ 
For every $n\in \NN _{1/2}$, we have as $\om \rightarrow +\infty$ and $D_{\Om}\rightarrow D$
\begin{equation}\label{eq0:DLenergie2}
\lambda_{2,n}\geq V_n(R_n) + 3 \sqrt{\frac{V_n '' (R_n)}{2}} - C\om^{1/2} \geq V_n(R_n) + 2 \sqrt{\frac{V_n '' (R_n)}{2}}. 
\end{equation}
\end{proposition}

Note that it is feasible to expand $\xi_{1,n}$ further and to make of Proposition \ref{theo:1D2} a statement as precise as Theorem \ref{theo:1D1}, two improvements that we do not need in the sequel.\\ 
For convenience we sum up some properties of $g_{1,n}$ that will be useful in the rest paper in the following Corollary.

\begin{corollaire}[$L^p$ bounds for $g_{1,n}$]\label{theo:Lpg1n}\mbox{}\\
For any $n\in \NN_{1/2}$ and any $1 \leq p \leq +\infty$ there is a constant $C_p$ depending only on $p$ such that
\begin{equation}\label{eq0:Lpg1n}
\Vert g_{1,n} \Vert_{L^p (\R^+,rdr)} \leq C_p \om^{1/4-\frac{1}{2p}}
\end{equation}
with the convention that $\frac{1}{p}= 0$ if $p=+\infty$.
\end{corollaire}

\begin{proof}
These estimates are simple consequences of (\ref{eq0:rescaling}), (\ref{eq1:cin}) and (\ref{eq0:DLfonction}). The integrals are computed using the change of variables $r=R_n+h_n x$. Remark that $c_{1,n} \leq C \om^{1/4}$ with $C$ independent of $\om$, $D_{\Om}$ and $n$.
\end{proof}

We proceed to the proof of Theorem \ref{theo:1D1}, assuming that its assumptions hold true for the rest of the section. We begin with an upper bound on  $\lambda_{1,n}$:
\begin{lemma}[Upper bound for $\lambda_{1,n}$]\label{theo:bornesup1d}\mbox{}\\
Under the assumptions of Theorem \ref{theo:1D1} we have  
\begin{equation}\label{eq1:bornesup}
\lambda_{1,n} \leq V_n(R_n) + \sqrt{\frac{V_n '' (R_n)}{2}} + \frac{V_n ^{(4)}(R_n)}{12 V_n '' (R_n)} \int_{\R}x^4 \xi_1 ^2 (x)dx + O(\om^{-1/2}).
\end{equation}
\end{lemma}

\begin{proof}
We proceed as for the proof of (\ref{eq0:energietest'}), taking as a trial function 
\begin{equation}
 g_n ^{test}(r)=c\xi _1 \left( h_n^{-1}(r-R_n) \right)
\end{equation}
where $c$ is chosen so that $2\pi\Vert g_n ^{test} \Vert _{L^2 (\R^+ ,rdr)}^2 = 1 $ with 
\begin{equation}
h_n= \left( \frac{2}{V_n '' (R_n)}\right)^{1/4}.
\end{equation}
An important fact for the calculation is that for the range of $n$ we are considering we always have $h_n=O(\om^{-1/2})$. We also use that $\xi_1$ is an even function.

\end{proof}
 
We now show the exponential decay of $g_{1,n}$ and its derivative outside of a shrinking region:

\begin{proposition}[Pointwise estimates for $g_{1,n}$]\label{theo:decexpo1d}\mbox{}\\
Let $g_{1,n}$ be defined by (\ref{eq1:lambda1n}).\\ 
There are positive constants $\sigma_1$ and $\sigma_2$ so that for every $n\in \NN _{1/2}$, every $r\in \R^+$ and for every $\ep >0$ 
\begin{equation}\label{eq1:decexpo}
 |g_{1,n}(r)|\leq C_{\ep}\om^{1/4+\ep} e^{-\sigma_1 \left( \frac{r-R_n}{h_n} \right)^2} 
\end{equation}
and
\begin{equation}\label{eq1:decexpo'}
 |g'_{1,n}(r)|\leq C_{\ep}\om^{3/4+\ep} e^{-\sigma_2 \left( \frac{r-R_n}{h_n} \right)^2}. 
\end{equation}  
Moreover, for any $\alpha < 1$ there is a constant $C_{\alpha}$ so that
\begin{equation}\label{eq1:dec0}
 |g_{1,n}(r)| \leq C e^{-C \om } r^{\frac{n}{\alpha}}
\end{equation}
for any $r \leq C_{\alpha}$. 

\end{proposition}

\begin{proof}

The proof is done in four steps. \\
\emph{Step 1.} We first need a global $L^{\infty}$ bound for $g_{1,n}$. We claim that 
\begin{equation}\label{eq1:borneLinf}
\forall \eta >0, \quad \forall \ep >0, \quad \exists C_{\ep}>0 \mbox{ so that } \Vert g_{1,n} \Vert_{L^{\infty}(\left]\eta,+\infty \right[ )} \leq C_{\ep} \om^{1/4+\ep}.
\end{equation}
We first note that (\ref{eq1:bornesup}) implies 
\begin{equation}\label{eq1:bornesupsimpl}
\lambda_{1,n} = F_n(g_{1,n})\leq V_n(R_n) + \sqrt{V_n '' (R_n)}
\end{equation}
because both $V''_n(R_n)\propto \om^2$ and $V_n^{(4)}(R_n)\propto \om^2$ when $n\in \NN_{1/2}$. Thus, using (\ref{eq1:bornesupsimpl}) and (\ref{eq0:normalisation}):
\[
2\pi \int _ {\R ^+} \left( |g_{1,n}' (r)|^2 +\left(V_n(r)-V_n(R_n)\right)|g_{1,n}(r)|^2 \right)rdr \leq \sqrt{V_n '' (R_n)}
\]
but $V_n(r)-V_n(R_n)\geq 0$ so
\begin{eqnarray*}
2\pi \eta \int_{\left]\eta,+\infty \right[} g' _{1,n}(r) ^2 dr & \leq & 2\pi \int _{\left]\eta,+\infty \right[} |g'_{1,n}(r)|^2rdr \\
&\leq& 2\pi \int _{\R ^+} |g'_{1,n}(r)|^2rdr \\
&\leq& \sqrt{V_n '' (R_n)}
\end{eqnarray*}
and recalling that $V_n '' (R_n)=O(\om ^2)$  for $n\in \NN_{1/2}$, using again (\ref{eq0:normalisation})
\begin{equation}\label{eq1:borneH1}
\Vert g_{1,n} \Vert_{H^{1}(\left]\eta,+\infty \right[ )}\leq C\om^{1/2}. 
\end{equation}
Now, an interpolation argument, using (\ref{eq0:normalisation}) and (\ref{eq1:borneH1}) yields
\begin{equation}
\forall \ep >0 \quad \Vert g_{1,n} \Vert_{H^{\frac{1+3\ep}{2}}(\left]\eta,+\infty \right[ )} \leq C_{\ep} \om^{1/4+\ep} 
\end{equation}
which implies (\ref{eq1:borneLinf}) by Sobolev imbedding.\\

\bigskip

\emph{Step 2.} 
We have
\[
 V_n(r) \geq V(r) = \frac{D_{\Om}\om^2 }{2} \left(r^2-1\right)^2
\]
and $V_n(R_n) + \sqrt{V_n '' (R_n)} = O(\om)$, thus there is a $\gamma >0$ so that 
\begin{equation}\label{eq1:domaineVn1}
 V_n(r) -V_n(R_n) - \sqrt{V_n '' (R_n)} > 0 \mbox{ for } |r-1| > \gamma \om^{-1/2}.
\end{equation}
Using that $R_n = 1+O(\om^{-1/2})$ and $h_n \propto \om^{-1/2}$ we deduce that there is a $\delta >0$ so that
\begin{equation}\label{eq1:domaineVn}
 V_n(r) -V_n(R_n) - \sqrt{V_n '' (R_n)} > 0 \mbox{ for } |r-R_n| > \delta h_n.
\end{equation}
We denote
\begin{equation}\label{eq1:segment}
 S_{n}=\left\lbrace r\in \R^+,\quad |r-R_n| \leq \delta h_n  \right\rbrace
\end{equation}
and prove an exponential decay property for $g_{1,n}$ on $S_{n}^c$. We recall the Euler-Lagrange equation for $g_{1,n}$ (considered as a radial function defined on $\R ^2$)
\begin{equation}\label{eq1:eqEL2}
 -\Delta g_{1,n} + V_n g_{1,n} = \lambda_{1,n} g_{1,n}.
\end{equation}
Using (\ref{eq1:bornesupsimpl}) we have for $\om$ large enough
\begin{equation}\label{eq1:sousol}
 -\Delta g_{1,n} + \left( V_n(r) -V_n(R_n) - \sqrt{V_n '' (R_n)}\right)g_{1,n}\leq 0.
\end{equation}
We now apply the comparison principle on $S_{n}^c$: the right-hand side of (\ref{eq1:decexpo}) is a supersolution for (\ref{eq1:sousol}) on $S_{n}^c$ with a proper choice of $\sigma_1$, the boundary condition being fulfilled thanks to (\ref{eq1:borneLinf}). Thus the result holds on $S_{n}^c$. On the other hand, on $S_n$ the inequality is true for $C_{\ep}$ large enough. Indeed, we have the global upper bound (\ref{eq1:borneLinf}) and on $S_n$ the function $e^{-\sigma_1 \left( \frac{r-R_n}{h_n} \right)^2}$ is bounded below by some constant.\\

\emph{Step 3.} We prove (\ref{eq1:dec0}). It is straigthforward to see that the function on the right-hand side of (\ref{eq1:dec0}) is a supersolution for equation (\ref{eq1:sousol}) on $B_{C_{\alpha}}$ for some constant $C_{\alpha}$ depending only on $\alpha$. We then use the maximum principle on $B_{\min(C_{\alpha},\frac{1}{2})}$ to conclude that (\ref{eq1:dec0}) holds. The boundary condition is fulfilled thanks to (\ref{eq1:decexpo}).\\ 

\emph{Step 4.} We now prove (\ref{eq1:decexpo'}) by an interpolation argument. Using (\ref{eq1:eqEL2}),  (\ref{eq1:bornesup}), the definition of $V_n$, (\ref{eq1:decexpo}) and (\ref{eq1:dec0}) we have
\begin{equation}\label{eq1:decexpo''}
|\Delta g_{1,n} (x)|\leq C\om^{5/4+\ep} e^{-s_3 \left(\frac{\vert x \vert-R_n}{h_n}\right)^2}  
\end{equation}
for some $s_3$. Using Gagliardo-Nirenberg's inequality (see for example \cite{BBH1,Nir})
\[
 \Vert \nabla g_{1,n} \Vert_{L^{\infty}} \leq C \Vert \Delta g_{1,n} \Vert_{L^{\infty}} ^{1/2} \Vert g_{1,n} \Vert_{L^{\infty}}^{1/2}
\]
on $B_R$ for any $0 \leq R<R_n$ and on $B_R ^c$ for any $R>R_n$, we get from (\ref{eq1:decexpo}) and (\ref{eq1:decexpo''})
\begin{equation*}
|\nabla g_{1,n} (x)|\leq C\om^{3/4+\ep} e^{-s_2 \left(\frac{\vert x \vert-R_n}{h_n}\right)^2}
\end{equation*}
for some $s_2$, which implies (\ref{eq1:decexpo'}).

\end{proof}

We recall the definition of the blow-up function $\xi_{1,n}$:
\begin{equation}\label{eq1:blowup}
 g_{1,n}(r)=c_{1,n}\xi_{1,n}\left( h_n^{-1} \left( r-R_n \right)\right) 
\end{equation}
where $h_n$ is defined in equation (\ref{eq0:h}). \\
Note that this function is only defined on $]-\frac{R_n}{h_n},+\infty[$, but as $-\frac{R_n}{h_n}\rightarrow -\infty$ and $g_{1,n}$ decreases exponentially fast in this limit, we may (remembering Remark 1) abuse notation and consider it as a function defined on $\R$.\\ 
The next lemma states a few estimates that will prove useful to expand $\lambda_{1,n}$. They are consequences of Proposition \ref{theo:decexpo1d}:

\begin{lemma}[Estimates for $\xi_{1,n}$]\label{theo:estimemoments}\mbox{}\\
Let $n$ be in $\NN_{1/2}$ and $C_{k,\ep}$ a generic constant depending only on $k$ and $\ep$. We have, for $\om$ large enough and every $\ep>0$
 \begin{enumerate}
\item $\left|c_{1,n}^2 h_n \int_{\R} x^k \xi_{1,n}(x)^2 \left ( R_n +h_n x \right)dx \right|\leq C_{k,\ep} \om^{\ep}$.\\
\item $c_{1,n}^2 h_n = \frac{1}{2\pi R_n} + O(\om^{-1/2+\ep})$   \\
\item $\left| h_n\int_{\R} x^k\xi_{1,n}(x) ^2 dx \right|\leq C_{k,\ep} \om^{\ep -1/2}$.\\
\item $\left| h_n \int_{\R} x^k \xi_{1,n}'^2(x) dx \right| \leq C_{k,\ep} \om^{\ep -1/2}$.\\
 \end{enumerate}
\end{lemma}

\begin{proof}
We start with the proof of 1. Using (\ref{eq1:decexpo}) we have for every $\ep>0$ 
\begin{multline*}
\left|c_{1,n}^2 h_n \int_{\R} x^k \xi_{1,n}(x)^2 \left( R_n +h_n x \right)dx\right| = \left| \int_{\R^+} \left(\frac{r-R_n}{h_n}\right)^k g_{1,n}^2 (r)rdr\right|\\
\leq C_{\ep}\int_{\R^+} \left|\frac{r-R_n}{h_n}\right|^k \om^{1/2+\ep} e^{-2 s_1 \left( \frac{r-R_n}{h_n} \right)^2} rdr.
\end{multline*}
A change of variables $r=R_n + h_nx$ then yields
\begin{equation}\label{eq1:mom2bis}
\left| \int_{\R^+} \left(\frac{r-R_n}{h_n}\right)^k g_{1,n}^2 (r)rdr\right|\leq C_{\ep} \om^{\ep} \int_{\R} x^k e^{-2 s_1 x^2} dx\leq C_{k,\ep} \om^{\ep}.  
\end{equation}
We turn to the proof of 2. We have
\begin{equation}\label{eq1:c1n1}
c_{1,n} ^2 h_n = c_{1,n} ^2 h_n \int_{\R} \xi_{1,n}^2 = \int_{\R^+} g_{1,n} ^2 (r)dr
\end{equation}
but
\begin{eqnarray*}
\int_{\R^+} g_{1,n} ^2 (r)dr &=& \frac{1}{R_n}\int_{\R^+}g_{1,n}^2(r)rdr+\frac{1}{R_n}\int_{\R^+} \left(R_n -r\right)g_{1,n}^2(r)dr \\
&=& \frac{1}{2\pi R_n}+\frac{1}{R_n}\int_{S_{n}} \left(R_n -r\right)g_{1,n}^2(r)dr  \\ &&+ \frac{1}{R_n}\int_{S_{n}^c} \left(R_n -r\right)g_{1,n}^2(r)dr.
\end{eqnarray*}
Using (\ref{eq1:decexpo}) again one gets
\begin{equation}\label{eq1:mom3}
\left|\int_{S_{n}^c} \left(R_n -r\right)g_{1,n}^2(r)dr \right|\leq C \om^{-1/2+\ep}. 
\end{equation}
Now
\begin{equation}\label{eq1:mom4}
-h_n\int_{S_{n}}g_{1,n}^2(r)dr \leq \int_{S_{n}} \left(R_n -r\right)g_{1,n}^2(r)dr \leq h_n \int_{S_{n}}g_{1,n}^2(r)dr
\end{equation}
and gathering (\ref{eq1:mom3}) and (\ref{eq1:mom4}) we get 
\[
 \left| \int_{\R^+} g_{1,n}^2(r)dr - \frac{1}{2\pi R_n}\right|\leq C \om^{-1/2+\ep} 
\]
for some constant $C$. Using (\ref{eq1:c1n1}) we have proved 2.\\
The proof of 3 uses the same kind of computations:
\begin{eqnarray*}
 \left| h_n\int_{\R}x^k\xi_{1,n}^2(x)dx\right|&=&\left| \frac{1}{c_{1,n} ^2} \int_{\R^+} \left( \frac{r-R_n}{h_n}\right)^k g_{1,n}^2(r)dr\right| \\
&\leq & \frac{C_{k,\ep}}{c_{1,n}^2}\om^{\ep}\\
&\leq&   C_{k,\ep}\om^{\ep} h_n.
\end{eqnarray*}
We have used (\ref{eq1:decexpo}) for the first inequality and the point 2 for the last one. We prove 4 exactly like 3, using (\ref{eq1:decexpo'}) instead of (\ref{eq1:decexpo}).

\end{proof}

We are now able to present the 
\begin{proof}[Proof of Theorem \ref{theo:1D1}] We split the proof in four steps.\\
\emph{Step 1.} We use the blow-up (\ref{eq1:blowup}) and expand $V_n(r)$ around $R_n$ in the expression of $\lambda_{1,n}$. Using the fact that
\begin{equation}\label{eq1:derivees}
\forall j\in \N, \quad \forall n\in \NN_{1/2} \quad  V^{(j)}(R_n)=O(\om^2)
\end{equation}
and $h_n=O(\om^{-1/2})$ we get 
\begin{multline}\label{eq1:develo1}
\lambda_{1,n}=F_n(g_{1,n})= 2\pi\frac{c_{1,n} ^2}{h_n}\int_{\R} \xi_{1,n}'(x)^2 \left(R_n + h_nx\right)dx
\\ + 2\pi c_{1,n}^2 \sum_{k=0}^N \frac{V_n^{(k)}(R_n)}{k!}\int_{\R} h_n^{k+1}x^k \xi_{1,n}(x)^2 \left(R_n + h_n x\right)dx \\ + O(\om^{2-N/2}) \int_{\R} \left| x^N \xi_{1,n}(x)^2 \left(R_n + h_nx\right)\right| dx
\end{multline}
where 
\begin{equation}\label{eq1:develo2}
c_{1,n}^2 = \frac{1}{2\pi h_n \int \xi_{1,n}(x)^2 \left(R_n + h_nx\right)dx } =\frac{1}{2\pi h_n R_n + 2\pi h_n^2\int_{\R} x\xi_{1,n}(x)^2dx}
\end{equation} 
Then we use the estimates of Lemma \ref{theo:estimemoments} and (\ref{eq1:derivees}) again to bound the coefficients of $\lambda_{1,n}$ in its expansion in powers of $h_n$ (\ref{eq1:develo1}). Recalling the definition of $h_n$ (\ref{eq0:h}) we obtain
\begin{equation}\label{eq1:develo3}
\lambda_{1,n}= V_n(R_n) +\sqrt{\frac{V_n '' (R_n)}{2}} \int _{\R} \left(\xi_{1,n} '(x) ^2 +x^2 \xi_{1,n} (x)^2 \right)dx+ O(\om^{1/2+\ep}).
\end{equation}
for every $\ep>0$. Combining this estimate with (\ref{eq1:bornesup}) and dividing by $\sqrt{\frac{V_n '' (R_n)}{2}}=O(\om)$ we find
\begin{equation}\label{eq1:energieOH}
\int _{\R} \left(\xi_{1} '(x) ^2 +x^2 \xi_{1} (x)^2 \right) + O(\om^{-1/2+\ep}) \geq \int _{\R} \left(\xi_{1,n} '(x) ^2 +x^2 \xi_{1,n} (x)^2 \right) \geq \int _{\R} \left(\xi_{1} '(x) ^2 +x^2 \xi_{1} (x)^2 \right),
\end{equation}
but $\xi_1$ is the unique normalized minimizer of the energy associated to (\ref{eq0:OH}), so we deduce 
\begin{equation}\label{eq1:xipremieraprox}
 \xi_{1,n} = \xi_1 + O_{HO}(\om^{-1/4+\ep})
\end{equation}
for every $\ep>0$. The next steps consist in improving (\ref{eq1:xipremieraprox}), using first a sharper energy expansion then an equation satisfied by $\xi_{1,n}$.

\bigskip

\emph{Step 2.} We progressively improve the remainder term in (\ref{eq1:develo3}) by a bootstrap argument. Going back to (\ref{eq1:develo1}) and using again the results of Lemma \ref{theo:estimemoments} we may be more explicit:
\begin{multline}\label{eq1:develo4}
\lambda_{1,n}\geq V_n(R_n) +\sqrt{\frac{V_n '' (R_n)}{2}} \int _{\R} \left(\xi_{1,n} '(x) ^2 +x^2 \xi_{1,n} (x)^2 \right)dx\\+ C\om^{1/2+\ep}\left(\int_{\R} x\xi_{1,n}(x)^2dx+\int_{\R} x^3\xi_{1,n}(x)^2dx + \int_{\R} x\xi_{1,n}'(x)^2dx\right)-C\om^{\ep} 
\end{multline}
for any $\ep >0$. If, for some $\beta$ we have 
\begin{equation}\label{eq1:xipremieraproxbis}
 \xi_{1,n} = \xi_1 + O_{HO}(\om^{\beta}),
\end{equation}
then it is easy to show that for every $\ep >0$
\begin{equation}\label{eq1:estimoment1} 
\left| \int_{\R} x^3\xi_{1,n}(x)^2dx-\int_{\R} x^3\xi_{1}(x)^2dx \right|\leq C \om ^{\beta+\ep}, 
\end{equation}
using (\ref{eq1:xipremieraproxbis}), (\ref{eq1:decexpo}) and the fact that $\xi_1$ is a Gaussian. Now, $\xi_1$ being an even function we deduce
\begin{equation}\label{eq1:estimoment2} 
\left| \int_{\R} x^3\xi_{1,n}(x)^2dx\right|\leq C \om ^{\beta+\ep}, 
\end{equation}
and arguing likewise to bound $\left| \int_{\R} x\xi_{1,n}(x)^2dx\right|$ and $\left| \int_{\R} x\xi_{1,n}'(x)^2dx\right|$ we improve (\ref{eq1:develo4}) into
\begin{equation}\label{eq1:develo5}
\lambda_{1,n}\geq V_n(R_n) +\sqrt{\frac{V_n '' (R_n)}{2}} \int _{\R} \left(\xi_{1,n} '(x) ^2 +x^2 \xi_{1,n} (x)^2 \right)dx+ C\om^{1/2+\ep+\beta}
\end{equation}
and the same argument as that used with (\ref{eq1:energieOH}) to get (\ref{eq1:xipremieraprox}) yields
\begin{equation}\label{eq1:xipremieraproxter}
 \xi_{1,n} = \xi_1 + O_{HO}(\om^{\frac{-1/2+\ep+\beta}{2}}),
\end{equation}
for every $\ep>0$. The fixed point of the function
\[
 \beta \mapsto -\frac{1}{4}+\frac{\beta}{2}+\frac{\ep}{2}
\]
being $\beta = -\frac{1}{2}+\ep $, we deduce by induction that for every $\ep >0$
\begin{equation}\label{eq1:xideuxaprox}
 \xi_{1,n} = \xi_1 + O_{HO}(\om^{-1/2+\ep}),
\end{equation}
using (\ref{eq1:xipremieraprox}) as a starting point for the induction.

\bigskip

\emph{Step 3.} We are going to improve the expansion (\ref{eq1:xideuxaprox}) using the equation satisfied by $\xi_{1,n}$. We claim that 
\begin{equation}\label{eq1:xitroisaprox}
 \xi_{1,n} = \xi_1 + h_n P_n \xi_1 + O_{HO}(\om^{-1+\ep})
\end{equation}
where $P_n$ is an odd polynomial of degree 3. Let us write the equation for $\xi_{1,n}$. We write (\ref{eq1:eqEL2}) in radial coordinates:
\begin{equation}\label{eq1:eqELbis}
-g_{1,n} ''(r) - \frac{g_{1,n}'(r)}{r}+ V_n(r) g_{1,n}(r) = \lambda_{1,n} g_{1,n}(r).
\end{equation}
Then, making the change of variables $r=R_n+h_n x$, expanding $V_n$ and $\frac{1}{r}$ and multiplying by $h_n^2 c_{1,n} ^{-1}$ we get
\begin{multline}\label{eq1:eqELxi0}
-\xi_{1,n}'' - \frac{h_n}{R_n}\xi_{1,n}'+h_n ^2 V_n(R_n)\xi_{1,n}+\frac{V_n ''(R_n)}{2}h_n^4x^2\xi_{1,n}+\frac{V_n '''(R_n)}{6}h_n^5x^3\xi_{1,n}
\\ = \lambda_{1,n}h_n ^2 \xi_{1,n}+O_{L^2}(\om^{-1+\ep}).
\end{multline}
We have used the estimates of Lemma \ref{theo:estimemoments} again to evaluate the remainder. Taking $\beta = -1/2 +\ep$ in (\ref{eq1:develo5}) we have
\[
\lambda_{1,n}=V_n(R_n)+\sqrt{\frac{V_n '' (R_n)}{2}} +O(\om^{\ep}),
\]
so, using the definition of $h_n$ (\ref{eq0:h}), (\ref{eq1:eqELxi0}) reduces to
\begin{equation}\label{eq1:eqELxi}
-\xi_{1,n}'' - \frac{h_n}{R_n}\xi_{1,n}'+ x^2\xi_{1,n}+\frac{V_n '''(R_n)}{6}h_n^5x^3\xi_{1,n}= \xi_{1,n}+O_{L^2}(\om^{-1+\ep}).
\end{equation}
Note that $\frac{V_n ''(R_n)}{2}h_n^4=1$ by definition. Now, if we write $\xi_{1,n}$ as
\begin{equation}\label{eq1:xitroisaproxbis}
\xi_{1,n} = \xi_1 + h_n\varphi_n
\end{equation}
and insert this in (\ref{eq1:eqELxi}), using the Euler-Lagrange equation for $\xi_1$ and the estimate
\[
\varphi_n = O_{HO}(\om^{\ep}) \quad \forall \ep >0 , 
\]
deduced from (\ref{eq1:xideuxaprox}) and (\ref{eq0:h2}), we get an equation satisfied by $\varphi_n$, namely
\begin{equation}\label{eq1:eqELphi}
-\varphi_n '' + x^2\varphi_n - \varphi_n = \frac{1}{R_n}\xi_{1}'-\frac{V_n '''(R_n)}{6}h_n^4x^3\xi_{1}+O_{L^2}(\om^{-1/2+\ep}).
\end{equation}
On the other hand, the mass constraint for $\xi_{1,n}$ allows us to deduce from (\ref{eq1:xitroisaproxbis})
\begin{equation}\label{eq1:orthogonalite}
\left| \int  \xi_1 \varphi_n \right| \leq C\om^{-1/2+\ep} 
\end{equation} 
Using an expansion on the basis of normalized eigenfunctions of the harmonic oscillator (\ref{eq0:OH}), one can see that the problem 
\begin{equation}\label{probPn}
\begin{cases} 
  -u '' + x^2 u - u = \frac{1}{R_n}\xi_{1}'-\frac{V_n'''(R_n)}{6}h_n^4x^3\xi_{1}\\ 
 \int  \xi_1 u  =0
\end{cases}
\end{equation}
has a unique solution in $H^1(\R)\cap L^2(\R,x^2 dx)$, which is of the form $u=P_n\xi_1$ where $P_n$ is an odd polynomial of degree 3. Substracting (\ref{eq1:eqELphi}) from the first line of (\ref{probPn}) and using (\ref{eq1:orthogonalite}) it is easy to show that
\begin{equation}\label{eq1:phi}
 \varphi_n = P_n\xi_1 + O_{HO}(\om ^{-1/2+\ep})
\end{equation}
thus proving (\ref{eq1:xitroisaprox}). The first line of (\ref{DLxi2}) follows by inspection of the right-hand side of (\ref{eq1:eqELphi}). \\

\bigskip

\emph{Step 4.} We complete the proof of (\ref{eq0:DLfonction}). The argument is similar to that of Step 3. We can compute from (\ref{eq1:xitroisaprox}) that
\begin{equation}\label{lambdandl}
 \lambda_{1,n}=V_n(R_n)+\sqrt{\frac{V_n '' (R_n)}{2}} +K_n,
\end{equation}
where $K_n = O(1)$. We write
\begin{equation}\label{eq1:xiquatreaproxbis}
\xi_{1,n}=  \xi_1 + h_n P_n\xi_1 + h_n^2 \psi_n
\end{equation}
and arguing as in Step 3 we get an equation for $\psi_n$:
\begin{multline}\label{eq1:eqELpsi}
-\psi_n '' + x^2\psi_n - \psi_n = K_n \xi_1 -\frac{x}{R_n^{2}}\xi_{1}'-\frac{V_n^{ (4)}(R_n)}{4!}h_n^4 x^4\xi_{1}\\+\frac{\left(P_n \xi_1\right)'}{R_n}-\frac{V_n ^{(3)}(R_n)}{4!}h_n^4x^3 P_n\xi_{1}+O_{L^2}(\om^{-1/2+\ep}) 
\end{multline}
and the condition
\begin{equation}\label{eq1:orthogonalite2}
\int  \xi_1 \psi_n = -\frac{1}{2}\int P_n ^2 \xi_1 ^2+O(\om^{-1/2}).  
\end{equation} 
Note that, multiplying (\ref{eq1:eqELpsi}) by $\xi_1$, integrating and using (\ref{eq1:orthogonalite2}) we get 
\begin{equation}\label{Knprime}
K_n = K'_n +O(\om^{-1/2+\ep})
\end{equation}
where 
\begin{equation}\label{Knprime2}
K'_n:=  \int \xi_1 \left(\frac{x}{R_n^{2}}\xi_{1}'+\frac{V_n^{ (4)}(R_n)}{4!}h_n^4 x^4\xi_{1}-\frac{\left(P_n \xi_1\right)'}{R_n}+\frac{V_n ^{(3)}(R_n)}{4!}h_n^4x^3 P_n\xi_{1} \right).
\end{equation}
Solving (\ref{eq1:eqELpsi}) we find
\begin{equation}\label{eq1:psi}
\psi_n = Q_n\xi_1 + O_{HO}(\om ^{-1/2+\ep}) 
\end{equation}
for every $\ep>0$, where $Q_n$ is a polynomial and $Q_n \xi_1$ is the unique solution of the problem 
\begin{equation}\label{probQn}
\begin{cases} 
  -u '' + x^2 u - u = K'_n \xi_1 -\frac{x}{R_n^{2}}\xi_{1}'-\frac{V_n^{ (4)}(R_n)}{4!}h_n^4 x^4\xi_{1}+\frac{\left(P_n \xi_1\right)'}{R_n}-\frac{V_n ^{(3)}(R_n)}{4!}h_n^4x^3 P_n\xi_{1}\\ 
 \int  \xi_1 u  = -\frac{1}{2}\int P_n ^2 \xi_1 ^2.
\end{cases}
\end{equation} This concludes the proof of (\ref{eq0:DLfonction}). We deduce (\ref{eq0:DLenergie1}) by straightforward calculations. 
\end{proof}

We present the
\begin{proof}[Proof of Corollary \ref{cor:lambdan/n^*}]
 
At this stage, we know that 
\[
\lambda_{1,n}= \Cc(R_n) + K_n\\ 
\]
where the function $\Cc$ is given by
\begin{equation}\label{eq1:Fcout}
\Cc (R_n) = V_n(R_n) + \sqrt{\frac{V_n '' (R_n)}{2}}
\end{equation}
and $K_n=O(1)$. Note that $K_n - K_m = O(\om^{-1/2})$ for any $m,n\in \NN_{1/2}$, as a consequence of (\ref{DLxi2}), which allows to focus on $\Cc$. It is a function of $R_n$ only, remembering the one-to-one relation (\ref{eq0:rn*}) between $R_n$ and $n$. Recalling (\ref{eq:Vn}) we have explicitly:
\begin{multline}\label{Fcout3}
\Cc (R_n) = D_{\Om}\frac{\om ^2}{2}R_n ^4 +\left(1-D_{\Om}\right)\om ^2 R_n ^2 + D_{\Om}\om ^2 -2\om ^2 \sqrt{D_{\Om}}R_n ^2 \sqrt{R_n ^2 +\frac{1-D_{\Om}}{D_{\Om}}}  \\  + \sqrt{12D_{\Om}\om ^2 R_n ^2 + 12\om ^2 \left(1-D_{\Om}\right)}.
\end{multline}
Studying the variations of $\Cc(r)$ with respect to $r\in \R^+$ we can see that $\Cc$ is minimized at a unique $R \in \R^{+}$ depending continuously on $D_{\Om}$ and $\om$. We have no explicit expression for $R$ but one can see that $R=1+O(\om ^{-1})$. Expanding $\Cc$ around $R$ we have for any $n\in \NN_{1/2}$
\begin{equation}\label{eq1:Fcout2}
\Cc(R_n) = \Cc(R) + \frac{\Cc '' (R)}{2} \left(R_n - R\right)^2\left(1+O(\om^{-1/2})\right) 
\end{equation}
because $\Cc ^{(k)} (R) = O(\om^2)$ for $k\geq 2$ and $R_n-R = O(\om^{-1/2})$ when $n\in \NN_{1/2}$. On the other hand, using (\ref{eq0:rn*}) we get that for any $n\in \NN_{1/2}$, $\om |R_n-R|\propto |N-n|$  where $N=D_{\Om}^{1/2}\om \left(R^6+\frac{1-D_{\Om}}{D_{\Om}}R^4\right)^{1/2}$ so
\begin{equation}\label{eq1:Fcout3}
\Cc(R_n) - \Cc(R) \propto  \frac{\Cc '' (R)}{2 \om^2} \left(n - N \right)^2\left(1+O(\om^{-1/2})\right) .
\end{equation}
We want to define $n^*$ as the minimizer in $\NN_{1/2}$ of $ \left(n - N \right)^2$. Only two cases can occur : either $N$ is exactly half an integer, then $[N]$ and $[N]+1$ both minimize $ \left(n - N \right)^2$ in $\NN_{1/2}$, or there is exactly one minimizer. \\
In both cases we pick one minimizer $n^*$ and (\ref{eq1:dependanceN}) follows by inspection of (\ref{eq1:Fcout3}). An important point is that $|n^*-N|\leq 1/2$.

%
%
\end{proof}
 
We now turn to the proof of Proposition \ref{theo:1D2}. We begin with an upper bound on $\lambda_{2,n}$:
\begin{lemma}[Upper bound for $\lambda_{2,n}$]\label{theo:bornesup1d2}\mbox{}\\
Let $\lambda_{2,n}$ and $g_{2,n}$ be defined by (\ref{eq1:lambda2n}).\\ 
We have for every $n\in\NN_{1/2}$, as $\om \rightarrow \infty$ and $D_{\Om}\rightarrow D$
\begin{equation}\label{eq1:bornesup2}
\lambda_{2,n} = F_n(g_{2,n})\leq V_n(R_n) + 3 \sqrt{\frac{V_n '' (R_n)}{2}}+O(\om^{1/2}).
\end{equation}
\end{lemma}

\begin{proof}
We begin with a test function constructed as that in the proof of Lemma \ref{theo:bornesup1d}:
\begin{equation}\label{eq1:ftest21}
 G_n ^{test}(r)=c\xi _2 \left( h_n^{-1}(r-R_n) \right)
\end{equation}
where $c$ is chosen so that $2\pi\Vert G_n ^{test} \Vert _{L^2}^2 = 1 $.
Then we define 
\begin{equation}\label{eq1:ftest22}
\tilde{G}_n ^{test}= G_n ^{test} - \left(\int_{\R^+} g_{1,n}(r)G_n ^{test}(r)rdr\right) g_{1,n} 
\end{equation}
and, as $\int_{\R^+}\tilde{G}_n ^{test}g_{1,n}rdr=0$ we have
\begin{equation}\label{lam21}
 F_n(g_{2,n})\leq \frac{1}{\Vert \tilde{G}_n ^{test} \Vert_{L^{2}(\R^+,rdr)} ^2}F_n(\tilde{G}_n ^{test}).
\end{equation}
Using the Euler-Lagrange equation (\ref{eq1:eqEL2}) for $g_{1,n}$ one proves that
\begin{equation}\label{lam22}
 F_n(\tilde{G}_n ^{test})=F_n(G_n ^{test})- \left(\int_{\R^+} g_{1,n}(r)G_n ^{test}(r)rdr\right) ^2 \lambda_{1,n} \leq F_n(G_n ^{test}).
\end{equation}
Computing $F_n(G_n ^{test})$ as (\ref{eq0:energietest'}), we obtain
\begin{multline}\label{lam23}
 F_n(G_n ^{test}) = V_n(R_n)+\sqrt{\frac{V_n '' (R_n)}{2}} \int_{\R}\left(\xi_2' (x) ^2 + x^2 \xi_2 (x) ^2 \right)dx + O(\om ^{1/2}) 
\\  = V_n(R_n) + 3 \sqrt{\frac{V_n '' (R_n)}{2}}+O(\om^{1/2}).
\end{multline}
because $\xi_2$ is the second normalized eigenfunction of the harmonic oscillator (\ref{eq0:OH}), associated to the eigenvalue $3$. On the other hand
\begin{multline}\label{lam24}
 \Vert \tilde{G}_n ^{test} \Vert_{L^{2}(\R^+,rdr)} ^2 = 1-\left(\int_{\R^+} g_{1,n}(r)G_n ^{test}(r)rdr\right) ^2 
\\\geq 1 -C \left( \int_{\R} \xi_{1,n} (x) \xi_2(x) (R_n +h_n x )dx \right) ^2 \geq 1-C \om^{-1} 
\end{multline}
using (\ref{eq0:DLfonction}) and $\int_{\R} \xi_1 \xi_2 =0$. Gathering equations (\ref{lam21}) to (\ref{lam24}) we get the result, remembering that $V_n(R_n) + 3 \sqrt{\frac{V_n '' (R_n)}{2}}=O(\om)$.

\end{proof}

\begin{proof}[Proof of Proposition \ref{theo:1D2}]
Starting from the upper bound (\ref{eq1:bornesup2}), we can prove the equivalents of Proposition \ref{theo:decexpo1d} and Lemma \ref{theo:estimemoments} for $g_{2,n}$. We omit the detailed calculations since they are easy modifications of those we used for $g_{1,n}$. Just note that $g_{2,n}$ is not positive so we have to rely on the inequality 
\begin{equation}\label{sousolg2n}
 -\Delta |g_{2,n}| ^2 + 2 V_n |g_{2,n}|^2 \leq 2 \lambda_{2,n} |g_{2,n}|^2
\end{equation}
to prove the equivalents of (\ref{eq1:decexpo}) and (\ref{eq1:dec0}). (\ref{sousolg2n}) is a consequence of (\ref{EElgjn}).\\ 
Then we can argue as in Step 1 of the proof of Theorem \ref{theo:1D1} and obtain
\begin{equation}\label{eq1:develo23}
\lambda_{2,n}= V_n(R_n) +\sqrt{\frac{V_n '' (R_n)}{2}} \int _{\R} \left(\xi_{2,n} '(x) ^2 +x^2 \xi_{2,n} (x)^2 \right)dx+ O(\om^{1/2+\ep}).
\end{equation}
But
\begin{equation*}
\int_{\R^+} g_{1,n}(r)g_{2,n}(r) rdr  = 0 
\end{equation*}
and this implies, using (\ref{eq0:DLfonction}), 
\begin{equation}\label{ortho}
 \int_{\R} \xi_1(x) \xi_{2,n} (x)dx = O(\om^{-1/2})
\end{equation}
We obtain (\ref{eq0:DLenergie2}) from (\ref{eq1:develo23}) and (\ref{ortho}). The second inequality holds because $V_n '' (R_n) = O(\om ^2)$. 
\end{proof}

\section{Lower bound on the interaction energy of $u_{\om}$, proof of Theorem \ref{theo:densite/energie}}

In this section we show how the estimates of Section 2 on the one-dimensional problems allow one to bound the quartic term in the energy (\ref{eq0:modele3}) from below. We then deduce the energy and density asymptotics of Theorem \ref{theo:densite/energie}.\\

We begin with a property of exponential decay for $u_{\om}$:

\begin{proposition}[Exponential decay for $u_{\om}$]\label{theo:decexpou}\mbox{}\\
There is a constant $\sigma$ so that 
\begin{equation} \label{decexpou}
|u_{\om}(x)|\leq C \om^{1/2} e^{-\sigma \om ||x|-1|^2}\mbox{ for any } x \in \R^2.
\end{equation}
\end{proposition}

\begin{proof}
 
We recall the Euler-Lagrange equation for $u_{\om}$:
\begin{equation}\label{eq2:EELu}
-\left(\nabla -i\om x^{\perp} \right)^2 u_{\om} +D_{\Om}\frac{\om^2}{2} \left(\vert x\vert^2 - 1 \right)^2 u_{\om} + 2 G \vert u_{\om}\vert ^2 u_{\om}=\mu _{\om} u_{\om} 
\end{equation} 
where $\mu_{\om}$ is the Lagrange multiplier associated with the mass constraint. From this we deduce an equation satisfied by 
$U:=\vert u_{\om}\vert^2$,
\begin{equation}\label{eq2:EELU}
-\Delta U +4 G U^2 +2\left| \nabla u_{\om} -i\om x^{\perp}u_{\om} \right|^2+2 U \left(D_{\Om}\frac{\om^2}{2} \left(\vert x\vert^2 - 1 \right)^2-\mu_{\om} \right)=0.
\end{equation}
As in the proof of Proposition \ref{theo:decexpo1d}, we need first a global $L^{\infty}$ bound. We claim that
\begin{equation}\label{eq2:borneLinf}
\Vert U \Vert_{L^{\infty}} \leq C \om, 
\end{equation}
which is proved using a method due to Alberto Farina \cite{Far}. We introduce
\[
 \tilde{U}:= U-\frac{\mu_{\om}}{2G}.
\]
We use Kato's inequality and (\ref{eq2:EELU})
\begin{eqnarray*}
\Delta (\tilde{U} ^+) &\geq& \one_{\tilde{U}\geq 0} \Delta \tilde{U}\\ &=& \one_{\tilde{U}\geq 0} \left(4G U^2+2\left| \nabla u_{\om} -i\om x^{\perp}u_{\om} \right|^2 + 2 U \left(D_{\Om}\frac{\om^2}{2} \left(\vert x\vert^2 - 1 \right)^2-\mu_{\om} \right) \right)\\
&\geq& \one_{\tilde{U}\geq 0} \left(4G U^2 -2\mu_{\om} U \right)\\
&=& \one_{\tilde{U}\geq 0} 4G \tilde{U}^2 \geq 4G (\tilde{U}^+ )^2.
\end{eqnarray*}
We have $-\Delta (\tilde{U}^+) + 4G(\tilde{U}^+)^2\leq 0$ which implies $\tilde{U}^+ \equiv 0$ (see \cite{Bre}). We multiply the Euler-Lagrange equation (\ref{eq0:EELu}) by $u_{\om}^*$ and integrate over $\R^2$ to get
\begin{equation}\label{potchim0}
\mu_{\om}= I_{\om} + G \int_{\R^2} |u_{\om}|^4
\end{equation}
and thus
\begin{equation}\label{potchim}
\mu_{\om} \leq 2 I_{\om} \leq C\om 
\end{equation}
because of the upper bound (\ref{eq0:energietest'}). Using (\ref{potchim}) and the fact that $G$ is larger than some constant, the claim is proved.\\
The rest of the proof is an application of the maximum principle. We define
\begin{equation}\label{Anneau}
\A_{\om,\delta}=\left\lbrace x\in \R^2,\: \left| \vert x \vert-1 \right|^2 \leq \delta \om ^{-1} \right\rbrace. 
\end{equation}
We have
\begin{equation}\label{eq2:sousol}
 -\Delta U + 2 U \left(D_{\Om}\frac{\om^2}{2} \left(\vert x\vert^2 - 1 \right)^2-\mu_{\om} \right) \leq 0
\end{equation}
and $\left(D_{\Om}\frac{\om^2}{2} \left(\vert x\vert^2 - 1 \right)^2-\mu_{\om} \right)>0$ on $\A_{\om,\delta}^c$ if $\delta$ is large enough. We use two comparison functions:
\[
 F_1 (x):= C_1\om e^{-\sigma_1 \om \left(|x|^2 - 1\right)^2}
\]
which is a supersolution to (\ref{eq2:sousol}) on $\A_{\om,\delta}^c  \cap \left\lbrace |x|\leq 1 \right\rbrace$ if $\sigma_1$ is chosen small enough and
\[
 F_2 (x):= C_2\om e^{-\sigma_2 \om \left(|x| - 1\right)^2}
\]
which is a supersolution to (\ref{eq2:sousol}) on $\A_{\om,\delta}^c  \cap \left\lbrace |x|\geq 1 \right\rbrace$ if $\sigma_2$ is chosen small enough. The boundary conditions for the application of the comparison principle are fulfilled thanks to (\ref{eq2:borneLinf}) by taking $C_1$ and $C_2$ large enough. There remains to note that 
\[
 F_1(x)\leq C_1\om e^{-\sigma_1 \om \left(|x| - 1\right)^2}
\]
to show that (\ref{decexpou}) holds on $\A_{\om,\delta}^c$ with $\sigma = \frac{1}{2}\min (\sigma_1,\sigma_2)$ and $C=\max (\sqrt{C_1},\sqrt{C_2})$.\\
On the other hand (\ref{eq2:borneLinf}) implies that (\ref{decexpou}) is true (for a large enough constant $C$) on $\A_{\om,\delta}$ because $e^{-\sigma \om \left(|x| - 1\right)^2}$ is bounded below there, which concludes the proof.
\end{proof}

Let us recall that we have the Fourier expansion 
\[
 u_{\om}(r,\theta)= \sum_{n\in \Z} f_n(r) e^{in\theta}.
\]
We denote $\left< .,.\right>$ the scalar product in $L^2(\R^{+},rdr)$:
\[
\left\langle u,v \right\rangle = \int_{\R^+} u(r) v^*(r) rdr
\]
and introduce
\begin{equation}\label{eq2:utilde}
\tilde{u}_{\om}(r,\theta):= \sum_{n\in \NN_{1/2}} \left< f_n,g_{1,n} \right> g_{1,n}(r)e^{in\theta}. 
\end{equation}
As a consequence of our analysis in Section 2, the modes corresponding to $n\notin \NN_{1/2}$ carry very little mass, whereas for $n\in \NN_{1/2}$ we have $f_n \propto g_{1,n}$ in the $L^2$ sense. This can be summed up in the

\begin{proposition}[First estimate on $\Vert u_{\om} -\tilde{u}_{\om}\Vert_{L^2(\R^2)}$]\label{theo:confine1}\mbox{}\\
Let $u_{\om}$ and $\tilde{u}_{\om}$ be defined respectively by (\ref{eq0:Energiemin}) and  (\ref{eq2:utilde}). Under the assumptions of Theorem \ref{theo:densite/energie} we have
\begin{equation}\label{eq2:confine1}
\Vert u_{\om} -\tilde{u}_{\om}\Vert_{L^2(\R^2)} \leq C G^{1/2} \om^{-1/4}=o(1). 
\end{equation}
 \end{proposition}

\begin{proof}
We have to prove
\begin{equation}\label{eq2:confine1modes3/2}
\sum _{n\in \NN_{1/2}^c} \int_{\R^+} |f_n|^2 rdr \leq C G\om^{-1/2} 
\end{equation}
and
\begin{equation}\label{eq2:confine1modesautres}
\sum _{n\in \NN_{1/2}} \int_{\R^+} |f_n- \langle f_n,g_{1,n} \rangle g_{1,n}|^2 rdr \leq C G\om^{-1/2}. 
\end{equation}
For $n\in \NN_{1/2}$ we can write $f_n$ as 
\begin{equation}\label{fndecompose}
f_n =  \langle f_n,g_{1,n} \rangle g_{1,n} + \left( \int_{\R^+} |f_n- \langle f_n,g_{1,n} \rangle g_{1,n}|^2 rdr \right)^{1/2}  g^{\perp}_n
\end{equation}
where 
\begin{equation*}
\int_{\R^+} g_{1,n} (r) g^{\perp}_n (r) rdr =0 
\end{equation*}
and 
\begin{equation*}
 2\pi \int_{\R^+} |g^{\perp}_n (r)| ^2 rdr = 1,
\end{equation*}
thus we have
\begin{equation}\label{decomposeFn}
F_n(f_n) \geq  \lambda_{1,n} \left|\langle f_n,g_{1,n} \rangle\right|^2  + \lambda_{2,n}\int_{\R^+}\left| f_n -\langle f_n,g_{1,n} \rangle g_{1,n} \right|^2rdr
\end{equation}
for $n\in \NN_{1/2}$. Using (\ref{eq0:decouple}) and (\ref{decomposeFn}), ignoring the interaction energy, we bound $F_{\om}(u_{\om})$ from below in the following way
\begin{multline*}
F_{\om}(u_{\om})\geq \sum_{n\in \NN_{1/2}} \lambda_{1,n} \left|\langle f_n,g_{1,n} \rangle\right|^2  +\sum_{n\in \NN_{1/2}} \lambda_{2,n}\int_{\R^+}\left| f_n -\langle f_n,g_{1,n} \rangle g_{1,n} \right|^2rdr 
\\+ \sum_{n\in \NN_{1/2}^c} \lambda_{1,n} 2\pi \int_{\R^+} \left|f_n\right|^2 rdr. 
\end{multline*} 
But the energy of $g_{1,n^*}e^{in^* \theta}$ is an upper bound for $F_{\om}(u_{\om})$ and $\lambda_{1,n} \geq \lambda_{1,n^*}$ for any $n$. Combined with the results of Proposition \ref{theo:1D2} and (\ref{eq0:premiersmodesborneinf}) we get 
\begin{eqnarray}\label{eq2:preuveconfine1}
\lambda_{1,n^*}+ 2\pi G \int_{\R^+} g_{1,n^*} ^4(r)rdr &\geq& \sum_{n\in \NN_{1/2}} \left( \left|\langle f_n,g_{1,n} \rangle \right|^2 \lambda_{1,n^*} \right.\nonumber 
\\ &+&\left.  \left(\lambda_{1,n^*} +  \sqrt{\frac{V''_n(R_n)}{2}} \right)  \int_{\R^+}\left| f_n -\langle f_n,g_{1,n} \rangle g_{1,n} \right|^2rdr \right) \nonumber 
\\&+&\sum_{n\in \NN_{1/2}^c} c\: \om 2\pi \int_{\R^+} |f_n|^2rdr
\end{eqnarray}
with $c>\sqrt{6}$. 
On the other hand, we know from the $L^2$ normalization of $u_{\om}$ that
\begin{eqnarray}\label{eq2:reconstruire1}
1=\sum_{n\in\Z} 2\pi \int |f_n|^2 rdr &=&  \sum_{n\in \NN_{1/2}} \left|\langle f_n,g_{1,n} \rangle\right|^2\nonumber \\ &+& \sum_{n\in\NN_{1/2}} \int \left| f_n -\langle f_n,g_{1,n} \rangle g_{1,n} \right|^2rdr \nonumber\\ &+&\sum_{n\in \NN_{1/2}^c} 2\pi\int_{\R^+} |f_n|^2 rdr.
\end{eqnarray}
Combining (\ref{eq2:preuveconfine1}) and (\ref{eq2:reconstruire1}), using that $\lambda_{1,n*}-\sqrt{6}\om < 0$ for $\om$ large enough we get 
\begin{eqnarray*}
2\pi G \int_{\R^+} g_{1,n^*} ^4(r)rdr \geq \sum_{n\in \NN_{1/2}}  \int_{\R^+}\left| f_n -\langle f_n,g_{1,n} \rangle g_{1,n} \right|^2rdr \sqrt{\frac{V''_n(R_n)}{2}} +\sum_{n\in \NN_{1/2}^c} C \: \om 2\pi \int_{\R^+} |f_n|^2rdr
\end{eqnarray*}
and this implies the desired result because $\sqrt{\frac{V''_n(R_n)}{2}} > C \om$ for $n\in \NN_{1/2}$ and 
\[
2\pi \int_{\R^+} g_{1,n^*} ^4(r)rdr \leq C \om^{1/2} 
\]
which is a consequence of (\ref{eq0:DLfonction}). 
\end{proof}

We now prove a simple lemma which states that the total interaction energy of $u_{\om}$ is an upper bound to the pairwise interactions of the modes $f_n$. We state the lemma for a general $\Phi \in L^4 (\R ^2)$:

\begin{lemma}[Generic lower bound on $\int _{\R^2} |\Phi|^4$]\label{theo:interactions}\mbox{}\\
Let $\Phi$ be any $L^4 (\R ^2)$ function with Fourier decomposition
\[
 \Phi (r,\theta)= \sum_{n\in \Z} \phi_n (r)e^{in\theta}.
\]
We have
 \begin{equation}\label{eq2:interactions}
\int _{\R^2} |\Phi|^4 \geq 2\pi \sum_{p,q} \int_{\R^+} |\phi_p|^2|\phi_q|^2 rdr.
\end{equation}

\end{lemma}

\begin{proof}
We expand $|\Phi|^4$:
\begin{equation}
\int _{\R^2} |\Phi|^4 = \sum_{n,p,q,r}\int_{\R^+} \int_0 ^{2\pi}\phi_n \phi_p \overline{\phi_q}\overline{\phi_r}e^{i\left(n+p-q-r\right)\theta}d\theta rdr 
\end{equation}
and the only integrals with respect to $\theta$ that are not zero are those for which $n+p-q-r=0$ so
\begin{eqnarray*}
\int _{\R^2} |\Phi|^4 &=& 2\pi \sum_{n,p,q} \int_{\R^+} \phi_{n+q}\phi_p \overline{\phi_q}\overline{\phi_{n+p}}rdr\\
&=& 2\pi \sum_n \int_{\R^+} \left|\sum_p \phi_p \overline{\phi_{n+p}}\right|^2rdr\\
&\geq& 2\pi \int_{\R^+} \left|\sum_p |\phi_p|^2 \right|^2 rdr\\
&=&2\pi \sum_{p,q} \int_{\R^+} |\phi_p|^2 |\phi_q|^2 rdr
\end{eqnarray*}

\end{proof}

We are now able to prove the lower bound on $\int_{\R^2} \vert u_{\om} \vert^4$, combining the results of Proposition \ref{theo:confine1} and Lemma \ref{theo:interactions}. A very important fact in the computation is that for $n\in \NN_{1/2}$, $f_n \propto g_{1,n}$, and that the appropriate blow-up of $g_{1,n}$ converges to a gaussian (equation (\ref{eq0:DLfonction})).
 
\begin{proposition}[Lower bound on the interaction energy]\label{theo:borneinfint}\mbox{}\\
Let $u_{\om}$ be a solution to (\ref{eq0:Energiemin}) with Fourier expansion
\[
 u_{\om}=\sum_{n\in \Z} f_n (r)e^{in\theta}.
\]
Set
\[
\tilde{u}_{\om}(r,\theta):= \sum_{n\in \NN_{1/2}} \left< f_n,g_{1,n} \right> g_{1,n}(r)e^{in\theta}.
\]
where $g_{1,n}$ is defined by equation (\ref{eq0:lambda1n}).
We have
\begin{multline}\label{eq2:borneinfint}
\int _{\R^2} |u_{\om}|^4 \geq 2\pi \int_{\R^+} |g_{1,n^*}|^4(r)rdr \\-C \om ^{-1/2} \sum_{n\in \NN_{1/2}}|n-n^*| ^2  \left|\langle f_n,g_{1,n} \rangle\right|^2 \\ +O(\om^{1/2}\Vert u_{\om} -\tilde{u}_{\om}\Vert_{L^2})+O(1).
 \end{multline}
\end{proposition}
 
\begin{proof}
We begin by using Lemma \ref{theo:interactions}:
\[
 \int_{\R^2} |u_{\om}|^4 \geq 2\pi \sum_{p,q} \int_{\R^+} |f_p|^2 |f_q|^2rdr.
\]
A consequence of Proposition \ref{theo:confine1} is that it is enough to consider only the pairwise interactions of the modes labeled by $n\in \NN_{1/2}$ (see equation (\ref{eq2:reconstruire}) below), so that we use 
\[
\int_{\R^2} |u_{\om}|^4 \geq 2\pi \sum_{p,q\in \NN_{1/2}} \int_{\R^+} |f_p|^2 |f_q|^2rdr
\]
and evaluate the right-hand side.\\
We note that, using (\ref{decexpou}) 
\begin{eqnarray}\label{reducdomaine}
\sum_{p,q\in \NN_{1/2}} \int_{r\geq 3/2} |f_p|^2 |f_q|^2rdr &\leq& C \om ^ 4 \sum_{p,q\in \NN_{1/2}} \int_{r\geq 3/2} e^{- C \om ^2 |r-1|^2} rdr \nonumber
\\ &\leq & C \sum_{p,q\in \NN_{1/2}} e^{-C\om} \int_{r\geq 3/2} e^{- C \om ^2 |r-1|^2} rdr \nonumber
\\ &\leq & C (\sharp (\NN_{1/2}) ) ^2 e^{-C\om} \leq e^{-C \om}.
\end{eqnarray}
We have used that $\sharp (\NN_{1/2}) \leq C \om ^{1/2}$. A similar estimate holds true for $\sum_{p,q\in \NN_{1/2}} \int_{r\leq 1/2} |f_p|^2 |f_q|^2 rdr$ so that it is sufficient to consider integration domains where $ 1/2 \leq r\leq 3/2$, which we implicitly do in the rest of the proof.\\
It is convenient to introduce the following quantities:
\begin{equation}\label{eq0:alphan}
\alpha_n = R_n - R_{n^*} 
\end{equation}
and the rescaled functions $\zeta_{p}$ defined for $p\in \NN_{1/2}$ by
\begin{equation}\label{eq2:blowupfn}
 f_{p}(r)=d_{p}\zeta_{p}\left( h_p^{-1} (r-R_p)\right)
\end{equation}
where $h_p$ is defined by equation (\ref{eq0:h}),
\begin{equation}\label{eq2:dp1}
d_{p}^2 = \frac{\int |f_p|^2 rdr}{ h_p \int \left|\zeta_{p}(x)\right|^2 \left(R_p + h_px\right)dx }
\end{equation}
and $\Vert \zeta_{p} \Vert_{L^2}=1$. In this proof we denote $\beta_n$ the $n$-th term of a generic sequence satisfying   
\begin{equation}\label{betan}
\sum_{n\in \Z} \beta_n ^2 \leq C
\end{equation}
where $C$ is indenpendent of $\om$ and $D_{\Om}$. The actual value of the quantity $\beta_n$ may thus change from line to line.\\ 
Using (\ref{eq2:confine1}), we have
\begin{eqnarray*}
 \sum_{p \in \NN_{1/2}} \int \left| d_p \zeta_p - c_{1,p} \left\langle f_p,g_{1,p}\right\rangle \xi_{1,p} \right| ^2 (R_p + h_p x)dx  &=& \sum_{p \in \NN_{1/2}} h_p ^{-1} \int_{\R^+} \left| f_p -  \left\langle f_p,g_{1,p}\right\rangle g_{1,p} \right| ^2 rdr 
\\&\leq& C \om^{1/2} \Vert u_{\om} - \tilde{u}_{\om}\Vert_{L^2} ^2
\end{eqnarray*}
thus
\begin{equation}\label{eq2:zeta/xi1}
\forall p\in \NN_{1/2} \quad d_p \zeta_p = c_{1,p} \left\langle f_p,g_{1,p}\right\rangle \xi_{1,p} + O_{L^2((R_p+h_p x )dx )}(\beta_p \om^{1/4} \Vert u_{\om}-\tilde{u}_{\om}\Vert_{L^2})
\end{equation}
where $\sum_p \beta_p ^2$ is bounded uniformly with respect to $\om$.\\
On the other hand, using the bounds $ \sum_n F_n (f_n) \leq F_{\om}(u_{\om})\leq C \om$ we have
\begin{equation}\label{borneH1zeta}
 \Vert d_p \zeta_p \Vert _{H^1((R_p+h_px )dx)} \leq \beta_p \om^{1/4}.
\end{equation}
Also
\begin{equation}\label{borneH1xi}
 \Vert c_{1,p} \xi_{1,p} \Vert _{H^1((R_p+h_px )dx)} \leq \om^{1/4}
\end{equation}
follows from (\ref{eq0:DLfonction}).\\
A change of variables $r=R_{n^*} +h_{n^*}x$ yields ($\alpha_n$ is defined by (\ref{eq0:alphan}))
\begin{equation}\label{eq2:calculsint1}
\sum_{p,q \in \NN_{1/2}}2\pi \int |f_p|^2 |f_q|^2 rdr=\sum_{p,q}  h_{n^*}d_p ^2 d_q ^2 \left( \int\left|\zeta_{p}\right|^2\left(\frac{h_{n^*}}{h_p}x+\frac{\alpha_p}{h_p}\right)\left|\zeta_{q}\right|^2\left(\frac{h_{n^*}}{h_q}x+\frac{\alpha_q}{h_ q} \right) \left(R_{n^*}+h_{n^*}x  \right)dx \right).
\end{equation}
Then we have, using (\ref{eq2:zeta/xi1}), (\ref{borneH1zeta}), (\ref{borneH1xi}) and Sobolev imbeddings 
\begin{multline}\label{eq2:calculsint2}
\sum_{p,q} h_{n^*}d_p ^2 d_q ^2 \int\left|\zeta_{p}\right|^2\left(\frac{h_{n^*}}{h_p}x+\frac{\alpha_p}{h_p}\right)\left|\zeta_{q}\right|^2\left(\frac{h_{n^*}}{h_q}x+\frac{\alpha_q}{h_ q} \right) \left(R_{n^*}+h_{n^*}x  \right)dx= 
\\ \sum_{p,q} \left|\left\langle f_p,g_{1,p}\right\rangle \right|^2 \left|\left\langle f_q,g_{1,q}\right\rangle\right|^2
h_{n^*}c_{1,p} ^2 c_{1,q} ^2 \int \xi_{1,p}^2\left(\frac{h_{n^*}}{h_p}x+\frac{\alpha_p}{h_p}\right)\xi_{1,q}^2\left(\frac{h_{n^*}}{h_q}x+\frac{\alpha_q}{h_ q} \right) \left(R_{n^*}+h_{n^*} x  \right)dx 
\\ + O(\om^{1/2}\Vert u_{\om}-\tilde{u}_{\om}\Vert_{L^2}).
\end{multline}
Remark that it is clear from the definition of $c_{1,n}$ (\ref{eq1:cin}) and (\ref{eq0:DLfonction}) that
\begin{equation}\label{estimc1n}
c_{1,n} = \frac{1}{2 \pi h_n R_n} + O(\om^{-1/2})
\end{equation}
because $\xi_1$ is an even function.
Using (\ref{eq0:DLfonction}) and (\ref{estimc1n}) we obtain
\begin{multline}\label{eq2:calculsint22}
\sum_{p,q} h_{n^*}d_p ^2 d_q ^2 \int\left|\zeta_{p}\right|^2\left(\frac{h_{n^*}}{h_p}x+\frac{\alpha_p}{h_p}\right)\left|\zeta_{q}\right|^2\left(\frac{h_{n^*}}{h_q}x+\frac{\alpha_q}{h_ q} \right) \left(R_{n^*}+h_{n^*}x  \right)dx= 
\\ \sum_{p,q} \left|\left\langle f_p,g_{1,p}\right\rangle \right|^2  \left|\left\langle f_q,g_{1,q}\right\rangle\right|^2
\frac{h_{n^*}}{h_p R_p h_q R_q} \int \xi_{1}^2\left(\frac{h_{n^*}}{h_p}x+\frac{\alpha_p}{h_p}\right)\xi_{1}^2\left(\frac{h_{n^*}}{h_q}x+\frac{\alpha_q}{h_ q} \right) \left(R_{n^*}+h_{n^*} x  \right)dx 
\\ + O(\om^{1/2}\Vert u_{\om}-\tilde{u}_{\om}\Vert_{L^2}) + O(1).
\end{multline}
But for any $p\in \NN_{1/2}$ we have from (\ref{hp-hq}) that $\frac{h_{n^*}}{h_p}=1+O(\om^{-1}|n^* - p|) = 1 + O(\om^{-1/2})$ because $n^*= \om +O(1)$, so 
\begin{multline}\label{eq2:calculsint3}
\int\xi_{1}^2\left(\frac{h_{n^*}}{h_p}x+\frac{\alpha_p}{h_p}\right)\xi_{1}^2\left(\frac{h_{n^*}}{h_q}x+\frac{\alpha_q}{h_ q} \right) \left(R_{n^*}+h_n ^* x  \right)dx \\ = \int \xi_{1}^2\left(x+\frac{\alpha_p}{h_p} \right)\xi_{1}^2\left( x+\frac{\alpha_q}{h_ q}\right) \left(R_{n^*}+h_{n^*}x  \right)dx +O (\om^{-1/2})
\end{multline}
and using that $\xi_1$ is a Gaussian we have
\begin{multline}\label{eq2:calculsint4}
\int \xi_{1}^2\left(x+\frac{\alpha_p}{h_p} \right)\xi_{1}^2\left( x+\frac{\alpha_q}{h_ q}\right) \left(R_{n^*}+h_{n^*}x  \right)dx=e^{-(\frac{\alpha_p}{h_p}-\frac{\alpha_q}{h_q})^{2}} \int \xi_1 ^4 \left(x+\frac{\alpha_p}{2h_ p }+\frac{\alpha_q}{2h_ q }\right)\left(R_{n^*}+h_{n^*}x  \right)dx
\\ \geq  \left( 1-C\om\left( \alpha_p-\alpha_q\right)^2\right)\int \xi_1 ^4(x) \left(R_{n^*}+h_{n^*}\left(x -\frac{\alpha_p}{2h_ p }-\frac{\alpha_q}{2h_ q }\right) \right)dx
\\ = \left( 1-C\om\left( \alpha_p-\alpha_q\right)^2\right)\int \xi_1 ^4(x) \left(R_{n^*}+h_{n^*}x  \right)dx +O(\om^{-1/2})
\end{multline}
where we have used that $h_n=O(\om^{-1/2})$ and $|\alpha_n| =|R_n -R_{n^*}|\propto \om^{-1}|n-n^*|=O(\om^{-1/2})$ (see (\ref{Rp-Rq})).\\
On the other hand we have
\begin{equation}\label{eq2:calculsint50}
\forall p,q \in \NN_{1/2}, \quad \frac{1}{h_p h_q} = \frac{1}{h_{n ^*} ^2} + O(\om ^{1/2}) 
\end{equation}
using that $\frac{h_{n^*}}{h_p}=1+O(\om^{-1}|n^* - p|) = 1 + O(\om^{-1/2})$ and
\begin{equation}\label{eq2:calculsint5}
\forall p,q \in \NN_{1/2}, \quad \frac{1}{R_p R_q} = \frac{1}{R_{n^*}^2}+O(\om^{-1/2})
\end{equation}
using that $|R_p - R_{n^*}|\propto \om^{-1}|p-n^*| = O(\om ^{-1/2})$.
Gathering equations (\ref{eq2:calculsint1}) to (\ref{eq2:calculsint5}) we have 
\begin{multline}\label{eq2:intpresque}
\sum_{p,q}2\pi \int |f_p|^2 |f_q|^2 rdr\geq \left(\sum_{p,q\in \NN_{1/2}} \left|\left\langle f_p,g_{1,p}\right\rangle \right|^2 \left|\left\langle f_q,g_{1,q}\right\rangle\right|^2 \right) \frac{1}{2\pi h_{n^*}R_{n^*}^2}
\\ \times  \left(\left( 1-C\om \left( \alpha_p-\alpha_q\right)^2)\right)\int \xi_1 ^4(x) \left(R_{n^*}+h_{n^*}x  \right)dx \right) 
\\  + O(\om^{1/2} \Vert u_{\om}-\tilde{u}_{\om}\Vert_{L^2})+O (1) .
\end{multline}
But (\ref{eq0:DLfonction}) gives
\begin{equation}\label{interactiong1n}
2\pi \int_{\R} g_{1,n^*}(r) ^4 rdr=\frac{1}{2\pi h_{n^*} R_{n^*}^2}\int_{\R} \xi_1(x)^4 \left(R_{n^*}+h_{n^*}x.  \right)dx +O(\om^{-1/2})
\end{equation}
We also have, as a consequence of (\ref{eq2:confine1}) and (\ref{eq0:masseu}),
\begin{multline}\label{eq2:reconstruire}
\sum_{p,q\in \NN_{1/2}} \left|\left\langle f_p,g_{1,p}\right\rangle \right|^2  \left|\left\langle f_q,g_{1,q}\right\rangle\right|^2=\left(\sum_{p\in \NN_{1/2}}  \left|\left\langle f_p,g_{1,p}\right\rangle \right|^2 \right)^2
\\ =\left(1-\Vert u_{\om}-\tilde{u}_{\om}\Vert_{L^2}\right)^{2} \geq 1 - C \Vert u_{\om}-\tilde{u}_{\om}\Vert_{L^2}
\end{multline} 
and
\begin{eqnarray}\label{eq2:reconstruire2}
\sum_{p,q\in \NN_{1/2}} \left(\alpha_p -\alpha_q\right)^2   \left|\left\langle f_p,g_{1,p}\right\rangle \right|^2  \left|\left\langle f_q,g_{1,q}\right\rangle\right|^2 &\leq& 2 \sum_{p,q\in \NN_{1/2}} \left(\alpha_p ^2 +\alpha _q  ^2\right) \left|\left\langle f_p,g_{1,p}\right\rangle \right|^2 \left|\left\langle f_q,g_{1,q}\right\rangle\right|^2 \nonumber
\\ &\leq& 4 \left(\sum_{p \in \NN_{1/2}} \alpha_p ^2 \left|\left\langle f_p,g_{1,p}\right\rangle \right|^2 \right)\left(\sum_{q \in \NN_{1/2}} \left|\left\langle f_q,g_{1,q}\right\rangle \right|^2 \right)  \nonumber
\\ &\leq& 4\sum_{p\in \NN_{1/2}} \alpha_p ^2  \left|\left\langle f_p,g_{1,p}\right\rangle\right| ^2
\end{eqnarray}
because $\sum_{q \in \NN_{1/2}} \left|\left\langle f_q,g_{1,q}\right\rangle \right|^2 \leq 1$.
Combining (\ref{eq2:intpresque}), (\ref{interactiong1n}), (\ref{eq2:reconstruire}) and (\ref{eq2:reconstruire2}) we obtain
\begin{equation*}
\int _{\R^2} |u_{\om}|^4 \geq 2\pi \int_{\R^+} |g_{1,n^*}|^4(r)rdr -C \sum_{p\in \NN_{1/2}}  \om ^{3/2} \alpha_p ^2  \left|\left\langle f_p,g_{1,p}\right\rangle \right|^2 +O(\om^{1/2}\Vert u_{\om} -\tilde{u}_{\om}\Vert_{L^2})+O(1).
 \end{equation*}
We conclude the proof by recalling that $|\alpha_n| = |R_n-R_{n^*}|\propto \om^{-1} |n-n^*|$ using (\ref{eq0:rn*}), so that 
\[
\om ^{3/2} \alpha_n ^2 \geq C \om^{-1/2}|n-n^*|^2.                                                                    \]
\end{proof}

Note that the $n$-th term of the second line of (\ref{eq2:borneinfint}) can be interpreted as an upper bound to the gain in interaction energy per unit mass carried by the mode $f_n$. On the other hand Corollary \ref{cor:lambdan/n^*} provides a lower bound to the loss in kinetic and potential energy per unit mass carried by the mode $f_n$. The comparison of this two quantities is the crucial argument of the 

\begin{proof}[Proof of Theorem \ref{theo:densite/energie}]
 
We proceed in three steps

\emph{Step 1. } We first claim that the following improvement of (\ref{eq2:confine1}) holds
\begin{equation}\label{eq3:confine1}
\Vert u_{\om}- \tilde{u}_{\om}\Vert _{L^2} \leq C G \om^{-1/2}. 
\end{equation}
Combining the results of Theorem \ref{theo:1D1} and Proposition \ref{theo:borneinfint} we get  
\begin{eqnarray}\label{eq3:preuveconfine1}
F_{\om}(u_{\om}) &\geq& 2\pi G \int_{\R^+} |g_{1,n^*}|^4(r)rdr \nonumber
\\ &+& \sum_{n\in \NN_{1/2},} \left|\langle f_n,g_{1,n} \rangle\right|^2 \left( \lambda_{1,n}-C G \om^{-1/2}|n-n^*|^2 \right) \nonumber 
\\&+& \sum_{n\in \NN_{1/2}} \lambda_{2,n}\int_{\R^+}\left| f_n -\langle f_n,g_{1,n} \rangle g_{1,n} \right|^2rdr \nonumber 
\\&+& \sum_{n\in \NN_{1/2}^c}\lambda_{1,n} 2\pi \int_{\R^+} |f_n|^2 rdr+O(G\om^{1/2}\Vert u_{\om} -\tilde{u}_{\om}\Vert_{L^2})+O(G).
\end{eqnarray}
We now use $g_{1,n^*} e^{in^* \theta}$ as a test function and proceed as in the proof of Proposition \ref{theo:confine1}. With the result of Corollary \ref{cor:lambdan/n^*} we get
\begin{multline}\label{eq3:preuveconfine2}
C(G \om^{1/2}\Vert u_{\om} -\tilde{u}_{\om}\Vert_{L^2}+G+G\om^{-1/2} )\\ \geq \sum_{n\in \NN_{1/2},n\neq n^*, n^*+1}  \left|\langle f_n,g_{1,n} \rangle\right|^2\left(C |n-n^*|^2 - G  \om ^{-1/2}  |n-n^*|^2 \right)  
\\+ \sum_{n\in \NN_{1/2}} C \om \int_{\R^+}\left| f_n -\langle f_n,g_{1,n} \rangle g_{1,n} \right|^2rdr
\\+\sum_{n\in \NN_{1/2}^c}C \om  \int_{\R^+} |f_n|^2 rdr.
\end{multline}
Now, for $n\in \NN_{1/2}$, one has 
\[
C |n-n^*|^2 - G  \om ^{-1/2} |n-n^*|^2 \geq 0
\]
because $G\ll \om^{1/2}$, so equation (\ref{eq3:preuveconfine2}) yields
\begin{eqnarray}
C\left(G \om^{1/2}\Vert u_{\om} -\tilde{u}_{\om}\Vert_{L^2} +G + G\om^{-1/2} \right)&\geq& C\om \sum_{n\in \NN_{1/2}}\int_{\R^+}\left| f_n -\langle f_n,g_{1,n} \rangle g_{1,n} \right|^2rdr \nonumber
\\&+&C\om \sum_{n\in \NN_{1/2}^c}  \int_{\R^+} |f_n|^2rdr \nonumber
\\ &\geq& C\om \Vert u_{\om} -\tilde{u}_{\om}\Vert_{L^2}^2 \nonumber
\end{eqnarray}
which implies (\ref{eq3:confine1}). \\
Remark also that we can replace $ \lambda_{2,n}\int_{\R^+}\left| f_n -\langle f_n,g_{1,n} \rangle g_{1,n} \right|^2rdr$ by $\sum_{j\geq 2} \lambda_{j,n} |\left\langle f_n,g_{j,n}\right\rangle|^2$ in (\ref{eq3:preuveconfine1}). Indeed, writing
\begin{equation}\label{projectionfn}
f_n = \sum_{j \geq 1} \left\langle f_n , g_{j,n} \right\rangle g_{j,n}
\end{equation}
one has
\begin{equation}\label{expansionFn}
 F_n (f_n) = \sum_{j\geq 1 } \left|\left\langle f_n , g_{j,n} \right\rangle\right|^2 \lambda_{j,n}.
\end{equation}
This yields 
\begin{equation}\label{souslecoude}
 C(G^2 + G + G\om^{-1/2}) \geq \sum_{n\in \NN_{1/2}} \sum_{j\geq 2} \lambda_{j,n} | \left\langle f_n,g_{j,n}\right\rangle|^2 
\end{equation}
which will be useful in Section 4.\\

\emph{Step 2.} 
We prove (\ref{eq0:resultenergie}). The upper bound is obtained by taking $g_{1,n^*} e^{i n^* \theta}$ as a test function. For the lower bouind we first prove the following essential estimate:
\begin{equation}\label{estimefonda}
 \sum_{n\in \NN_{1/2}}  \Vert f_n \Vert_{L^2 (\R^+,rdr )} ^2 |n-n^*|^2 \leq C G^2 .  
\end{equation} 
We recall that we assume $G^2\geq g^2$ for some constant $g$ so it suffices to prove
\begin{equation}\label{estimefonda'} 
\sum_{n\in \NN_{1/2},n\neq n^*, n^*+1} \Vert f_n \Vert_{L^2 (\R^+,rdr )} ^2 |n-n^*|^2  \leq C G^2.
\end{equation}
Coming back to (\ref{eq3:preuveconfine2}) and using (\ref{eq3:confine1}) we get
\begin{equation}\label{preuveconfine3}
C\left(G^2 + G + G \om^{-1/2} \right) \geq \sum_{n\in \NN_{1/2}, n\neq n^*, n^*+1} \left|\langle f_n,g_{1,n} \rangle\right|^2 \left(|n-n^*|^2 - G  \om ^{-1/2} |n-n^*|^2 \right).    
\end{equation}
But $G\om^{-1/2} \ll 1$, so if $n\neq n^*, n^* +1$ we have, for some constant $C'<C$,  
\[
C |n-n^*|^2 - G  \om ^{-1/2} |n-n^*|^2 \geq C'|n-n^*|^2.
\]
This allows to deduce from (\ref{preuveconfine3}) that 
\begin{equation}\label{estimefonda''} 
\sum_{n\in \NN_{1/2},n\neq n^*, n^*+1} \left|\langle f_n,g_{1,n} \rangle\right|^2 |n-n^*|^2  \leq C G^2.
\end{equation}
On the other hand we already have as a consequence of (\ref{eq3:confine1}) that 
\begin{equation}\label{estimfondatilde}
 \sum_{n\in \NN_{1/2}} |n-n^*|^2  \int_{\R^+}\left|f_n-\left\langle f_n, g_{1,n} \right\rangle\right|^2 rdr \leq C \om \sum_{n\in \NN_{1/2}} \int_{\R^+}\left|f_n-\left\langle f_n, g_{1,n} \right\rangle\right|^2 rdr \leq C G^2
\end{equation}
so (\ref{estimefonda''}) implies (\ref{estimefonda'}) ans thus (\ref{estimefonda}).


Now, according to Proposition \ref{theo:borneinfint} and (\ref{eq3:confine1}) we have 
\begin{equation}\label{lbound1}
F_{\om}(u_{\om})\geq \lambda_{1,n^*} + 2\pi G \int_{\R^+ } g_{1,n^*} ^4 (r)rdr - C G  \om ^{-1/2}  \sum_{n\in \NN_{1/2}}|n-n^*|^2  \Vert f_n \Vert^2 _{L^2(\R^+,rdr)} - C G^2.
\end{equation}
Using (\ref{estimefonda}) and $G\ll \om ^{1/2}$ we get
\[
 F_{\om}(u_{\om})\geq \lambda_{1,n^*} + 2\pi G \int_{\R^+ } g_{1,n^*} ^4 (r)rdr -C G^2 
\]
which was the missing lower bound to prove (\ref{eq0:resultenergie}). Recall that $G\ll \om ^{1/2}$ and $\int_{\R^+} g_{1,n^*} ^4 (r)rdr\propto \om^{1/2}$.\\


\emph{Step 3.} We prove (\ref{eq0:resultdensite}). We have 
\begin{equation}\label{preuvedensite1}
\int_{\R^2} \left(|u_{\om}|^2 -g_{1,n^*}^2\right)^2 = \int_{\R^2} |u_{\om}|^4+g_{1,n^*}^4-2g_{1,n^*}^2|u_{\om}|^2 \leq 2 \int_{\R^2} \left(g_{1,n^*}^4 -g_{1,n^*}^2|u_{\om}|^2\right).
\end{equation}
Indeed
\[
\lambda_{1,n^*} + G \int_{\R^2} |g_{1,n^*}|^4 \geq F_{\om}(u_{\om}) \geq \lambda_{1,n^*} + G \int_{\R^2} |u_{\om}|^4 
\]
so 
\begin{equation}\label{bornesupint}
 \int_{\R^2} |u_{\om}|^4 \leq \int_{\R^2} |g_{1,n^*}|^4.
\end{equation}
Now, using the same calculations as those in the proof of Proposition \ref{theo:borneinfint} and (\ref{eq3:confine1}) we can prove
\begin{multline}\label{varianteborneinf}
\int_{\R^2} g_{1,n^*}^2|u_{\om}|^2 \geq \sum_{n \in \NN_{1/2}} 2\pi \int_{\R^+} |f_n(r)|^2g_{1,n^*}^2(r)rdr \\ 
\geq \int_{\R^2} g_{1,n^*}^4 - C  \sum_{n\in \NN_{1/2}}  \om ^{-1/2} |n-n^*|^2 \Vert f_n \Vert^2 _{L^2(\R^+,rdr)} - C G
\end{multline}
Then plugging (\ref{varianteborneinf}) and (\ref{estimefonda}) into (\ref{preuvedensite1}) we get 
\[
 \int_{\R^2} \left(|u_{\om}|^2 -g_{1,n^*}^2\right)^2 \leq C  \sum_{n\in \NN_{1/2}}  \om ^{-1/2}  |n-n^*|^2 \Vert f_n \Vert^2 _{L^2(\R^+,rdr)} + C G \leq C (G^2 \om^{-1/2}+G) \leq C G 
\]
and the proof is complete.
\end{proof}

\section{The giant vortex state for a fixed coupling constant}

In this section we present the proofs of Theorems \ref{theo:densite/energie2} and \ref{theo:vortex}. They require improvements of the method presented in Sections 2 and 3. In particular we will need an energy expansion with a remainder term going to $0$ as $\om$ goes to infinity. We know from (\ref{eq0:resultenergie}) that 
\[
 F_{\om} (u_{\om}) = \lambda_{1,n^*} + 2\pi G \int_{\R^+} g_{1,n^*} ^4 rdr+O(1)
\]
in the regime where $G$ is a fixed constant. But, for any $n$, by the definition (\ref{eq4:gamman}) of $\gamma_n$ 
\[
 \lambda_{1,n} + 2\pi G \int_{\R^+} g_{1,n} ^4 rdr \geq \gamma_{n}
\]
and one can realize from (\ref{eq0:DLfonction}) and (\ref{eq4:DLfonction}) below that the difference between these two quantities is of order $1$. We thus need to analyze in some details the nonlinear problem defining $\gamma_{n}$, and use the results to provide better upper and lower bounds to $F_{\om} (u_{\om})$. We first present the analysis of the nonlinear one dimensional problems in Subsection 4.1 below.

\subsection{Nonlinear one dimensional problems}

We recall the definition of the one-dimensional energies
\begin{equation}\label{eq4:1Dnonlineaire}
 E_n(f)=2\pi \int _ {\R ^+} \left( |f' (r)|^2 +V_n(r)|f(r)|^2 +G |f(r)|^4\right)rdr
\end{equation}
with
\begin{equation}\label{eq4:gamman}
 \gamma_{n}= E_n(\Psi_{n})= \inf \left\lbrace E_n(f), \: f\in H^1(\R^+,rdr)\cap L^2(\R^+,rV_n(r)dr), 2\pi \int_{\R ^+} f^2(r)rdr =1  \right\rbrace. 
\end{equation}
Note that $\Psi_n$ is uniquely defined up to a multiplicative constant of modulus $1$, that we fix by requiring that 
\begin{equation}\label{Psinpositive}
 \Psi_n \geq 0.
\end{equation}
We also define the rescaled functions $\rho_{n}$ by
\begin{equation}\label{eq4:rescaling}
 \rho _{n}(x)= c_{n}^{-1} \Psi_{n}\left( R_n + h_n x\right)
\end{equation}
where $c_{n}$ is chosen so that $\Vert \rho_{n}\Vert_{L^2}=1$:
\begin{equation}\label{eq4:cn}
c_{n}^2  =\frac{1}{2\pi h_n R_n + 2\pi h_n^2\int_{-\frac{R_n}{h_n}}^{+\infty} x\rho_{n}(x)^2 dx}.
\end{equation}

We have
\begin{theorem}[Asymptotics for the ground states of the one-dimensional nonlinear problems] \label{theo:1D1NL} \mbox{}\\
Let $\gamma_{n}$ be defined by (\ref{eq4:gamman}).\\ 
Suppose that $n\in \NN _{1/2}$ as defined in (\ref{eq0:N3/2}) and let $R_n$  be defined in (\ref{eq0:rn*}). Let $\xi_1$ be the normalized ground state of the harmonic oscillator (\ref{eq0:OH}). \\  
We have as $\om \rightarrow \infty$ and $D_{\Om}\rightarrow D$
\begin{equation}\label{eq4:DLenergie}
\gamma_n = V_n(R_n)+\sqrt{\frac{V''_n (R_n)}{2}}+\frac{G}{2\pi h_n R_n}\int_{\R} \xi_1 (x) ^4 dx + J_n 
\end{equation}
where $J_n = O(1)$.\\
Let  $\rho_{n}$ be defined by equation (\ref{eq4:rescaling}). We have 
\begin{equation}\label{eq4:DLfonction}
\rho_{n}=\xi_1 + h_n \tau_n + h_n^2 \upsilon_n +O_{HO}(\om^{-3/2})\\
\end{equation}
where $h_n$ is defined in (\ref{eq0:h}), $\tau_n$ and $\upsilon_n$ are solutions to linear second order ODEs. They are bounded in harmonic oscillator norm, uniformly with respect to $n$. Also, for any $p,q \in \NN_{1/2}$
\begin{eqnarray}\label{DLzeta2}
\tau _p & =& \tau_q + O_{HO}(\om^{-1}|p-q|)\\
\upsilon _p &=& \upsilon_q + O_{HO}(\om^{-1}|p-q|).\nonumber
\end{eqnarray}
\end{theorem} 

From this theorem we deduce as a corollary 
\begin{corollaire} [Variations of $\gamma_{n}$ with respect to $n$] \label{cor:lambdan/n^*NL}\mbox{}\\
Let $n^*$ and $n ^* +1$ be defined as in Corollary \ref{cor:lambdan/n^*}. Under the assumptions of Theorem \ref{theo:1D1NL} and for $\om$ large enough, for any $n\neq n ^*,\: n^* +1$ one has
\begin{equation}\label{eq4:dependanceN}
\gamma_{n}=\gamma_{n^*} +   C  \left(n-n^*\right)^2\left(1+ O(\om^{-1/2})\right). 
\end{equation}
Moreover
\begin{equation}\label{gamman+1}
 \gamma_{n^*+1} \geq \gamma_{n^*} - C\om^{-1/2}.
\end{equation}
\end{corollaire}

We also state the equivalent of Corollary \ref{theo:Lpg1n}:
\begin{corollaire}[$L^p$ bounds for $\Psi_n$]\label{theo:LpPsin}\mbox{}\\
For any $n\in \NN_{1/2}$ and any $1 \leq p \leq +\infty$ there is a constant $C_p$ depending only on $p$ such that
\begin{equation}\label{eq4:LpPsin}
\Vert \Psi_{n} \Vert_{L^p (\R^+,rdr)} \leq C_p \om^{1/4-\frac{1}{2p}}
\end{equation}
with the convention that $\frac{1}{p}= 0$ if $p=+\infty$.
\end{corollaire}

The proofs are exactly similar to those of the corresponding results in Section 2, so we only give their main steps.\\

\begin{proof}[Proof of Theorem \ref{theo:1D1NL}]
We begin with an upper bound on $\gamma_n$ using a test function of the form
\[
 \Psi_n ^{\mbox{test}} = c_n^{\mbox{test}} \xi_1 \left(\frac{r-R_n}{h_n}\right).
\]
We obtain 
\begin{equation}\label{eq4:bornesup}
\gamma_n \leq V_n(R_n)+\sqrt{\frac{V''_n (R_n)}{2}}+\frac{G}{2\pi h_n R_n}\int_{\R} \xi_1 (x) ^4 dx + O(1)
\end{equation}
We then prove an exponential decay property for $\Psi_n$ and its first derivative exactly similar to that of Proposition \ref{theo:decexpo1d}. Step 1 of the proof of this Proposition needs no modification to apply to $\Psi_n$.\\
We use the maximum principle on the equation for $\Psi_n$ which reads
\begin{equation}\label{eq4:equationPsi}
-\Delta \Psi_n + V_n(r)\Psi_n +2G  \Psi_n ^3 = \left(\gamma_n + 2\pi G \int_{\R^+} |\Psi_n (r)|^4 rdr \right)  \Psi_n.
\end{equation}
Using (\ref{eq4:bornesup}) we get easily
\[
 -\Delta \Psi_n  + \left(V_n(r)-V_n(R_n) - \sqrt{V_n '' (R_n)}\right)\Psi_n \leq 0
\]
and we proceed as in the proof of Proposition \ref{theo:decexpo1d}. We then prove the equivalent of Lemma \ref{theo:estimemoments}, which allows to get an expansion of $\gamma_n$. We use this expansion as in Step 1 of the proof of Theorem \ref{theo:1D1} and get
\[
 \rho_n = \xi_1 + O_{HO} (\om ^{-1/4+\ep}) 
\]
for any $\ep>0$, and follow the method of Step 2 of the same proof to improve the remainder term and get
\begin{equation}
\rho_n = \xi_1 + O_{HO} (\om ^{-1/2+\ep}).
\end{equation}
We deduce an equation for $\rho_n$ starting from (\ref{eq4:equationPsi}):
\begin{multline}\label{eq4:equationrho}
-\rho_n''-\frac{h_n}{R_n} \rho '_n + \frac{V_n ''(R_n)}{2}h_n ^4 x^2 \rho_n +2Gc_n ^2 h_n ^2 \rho_n ^3 +\frac{V_n ^{(3)}(R_n)}{6}h_n ^5 x^3 \rho_n \\ = h_n^2 \left(\gamma_n + 2\pi G \int_{\R^+} |\Psi_n (r)|^4 rdr \right) \rho_n +O_{L^2}(h_n ^2).
\end{multline}
We then deduce from this equation that 
\begin{equation} 
\rho_{n}=\xi_1 + h_n \tau_n + h_n^2 \upsilon_n +O_{HO}(\om^{-3/2}).
\end{equation}
We have denoted $\tau_n$ the solution of the problem
\[
\begin{cases} 
  -u '' + x^2 u - u = \frac{1}{R_n}\xi_{1}'-\frac{V_n^{(3)}(R_n)}{6}h_n^4 x^3 \xi_{1}+ \frac{G}{\pi R_n} \left(\int \xi_1 ^4 \right) \xi_1 - \frac{G}{\pi R_n}\xi_1 ^3\\ 
 \int  \xi_1 u  =0
\end{cases}
\]
and $\upsilon_n$ the solution of the problem
\[
\begin{cases} 
  -u '' + x^2 u - u = J'_n \xi_1 -\frac{x}{R_n ^2}\xi_1' -\frac{V_n ^{(4)}}{4!}h_n ^4 x^4 \xi_1 + \frac{\tau_n '}{R_n}-\frac{V_n^{(3)}(R_n)}{6}h_n ^4 x^3 \tau _n \\ \quad \quad - \frac{3G}{\pi R_n} \tau_n \xi_1 ^3 + \frac{2G}{\pi R_n} \left(\int \xi_1 ^4 \right) \tau_n\\ 
 \int  \xi_1 u  = -\frac{1}{2}\int \tau_n ^2
\end{cases}
\]
where $J'_n=O(1)$ is chosen so that the right-hand side of the above equation is a function orthogonal to $\xi_1$, which implies that the system indeed has a unique solution. In particular one has $J'_n= J_n + O(\om^{-1/2})$.\\
This follows exactly Steps 3 and 4 of the proof of Theorem \ref{theo:1D1}. Remark that the coefficients of the ODEs satisfied by $\tau_n$ and $\upsilon_n$ depend continuously on $R_n$ seen as a continuous variable for the range of $n$ we consider. The esimates (\ref{DLzeta2}) are a consequence of this fact. We deduce (\ref{eq4:DLenergie}) from (\ref{eq4:DLfonction}). 
\end{proof}
 
\begin{proof}[Proof of Corollary \ref{cor:lambdan/n^*NL}]
At this stage, we know that 
\[
\gamma_{n}= \Cc(R_n) + \DD (R_n) + J_n\\ 
\]
where the function $\Cc$ is defined in equation (\ref{eq1:Fcout}), $\DD$ is given by
\begin{equation}\label{eq4:Fcout}
\DD (R_n) = \frac{G}{2\pi h_n R_n}\int_{\R} \xi_1 (x) ^4 dx
\end{equation}
and $J_n=O(1)$. We have $J_n - J_m = O(\om^{-1/2})$ for any $m,n\in \NN_{1/2}$, as a consequence of (\ref{DLzeta2}), and from the explicit expression of $\DD$ one can compute that 
\[
| \DD ' (R_n) | \leq C\om^{1/2}
\]
for any $n\in\NN_{1/2}$. We thus have, using Corollary \ref{cor:lambdan/n^*}
\begin{eqnarray*}
 \gamma_{n} &=& \gamma_{n^*} + C|n-n^*|^2 +O(\om^{1/2})|R_n-R_{n^*}| + O(\om^{-1/2})\\
&=& \gamma_{n^*} + C|n-n^*|^2 + O(\om^{-1/2})|n-n^*| + O(\om^{-1/2})\\
&=& \gamma_{n^*} + C|n-n^*|^2 (1+O(\om^{-1/2}))
\end{eqnarray*}
for any $n\neq n^*,\: n^* +1$.\\
Equation (\ref{gamman+1}) also follows from this discussion. 
\end{proof}

\subsection{Some improved estimates}

In this subsection we aim at obtaining a better $L^{\infty}$ bound for $u_{\om}$, which will be crucial in our analysis. Indeed, the first bound that we obtained (see Proposition \ref{theo:decexpou}) is far from being optimal as one can realize by taking the $L^2$ norm of both sides of (\ref{decexpou}). The right-hand side has a $L^2$ norm proportional to $\om^{1/4}$ whereas $\Vert u_{\om} \Vert_{L^2 (\R^2)} = 1$. A result of the analysis below is the improved bound
\begin{equation}\label{uLinfmieux}
 \Vert u_{\om} \Vert_{L^{\infty} (\R^2)} \leq C \om^{1/4}.
\end{equation}
This will follow from (\ref{decexpoutilde}) and (\ref{estimeLinf}) below and implies (this is a simple modification of the proof of Proposition \ref{theo:decexpou}) that (\ref{decexpou})
can be improved as
\begin{equation} \label{decexpoumieux}
|u_{\om}(x)|\leq C \om^{1/4} e^{-\sigma \om ||x|-1|^2}\mbox{ for any } x \in \R^2.
\end{equation}  
The right-hand side of this inequality is now bounded in $L^2$.\\
We begin with a estimates for $\tilde{u}_{\om}$: 

\begin{proposition}[Estimates for $\tilde{u}_{\om}$]\label{theo:decexpo}\mbox{}\\
There exist a constant $\tilde{\sigma}$ so that the following holds on $\R^2$
\begin{equation}\label{decexpoutilde}
|\tilde{u}_{\om}(x)|\leq C \om^{1/4} e^{-\tilde{\sigma}\om ||x|-1|^2}.
\end{equation}
Moreover, for any $0 \leq p\leq +\infty$
\begin{equation}\label{Lputilde}
 \Vert \tilde{u}_{\om} \Vert_{L^p(\R ^2)} \leq C \om ^{1/4-1/2p } 
\end{equation}
with the convention that $\frac{1}{p} = 0$ if $p=+\infty$.
\end{proposition}

\begin{proof}
We recall that 
\[
 \tilde{u}_{\om}(r,\theta):= \sum_{n\in \NN_{1/2}}  \left< f_n,g_{1,n} \right> g_{1,n}(r)e^{in\theta}
\]
so, using (\ref{eq0:Lpg1n}) we get 
\begin{eqnarray}\label{Linftilde}
 |\tilde{u}_{\om}(x)|&\leq& C \om^{1/4} \sum_{n\in \NN_{1/2}} \left| \left< f_n,g_{1,n} \right> \right| \nonumber
\\&\leq& C \om^{1/4} \left(\sum_{n\in \NN_{1/2}} \left| \left< f_n,g_{1,n} \right> \right| ^2 |n-n^*|^2 \right)^{1/2} \left(\sum_{n\in \NN_{1/2}}|n-n^*|^{-2}\right)^{1/2}\nonumber 
\\&\leq& C \om^{1/4}
\end{eqnarray}
where we have used (\ref{estimefonda}) to pass to the third line. The estimate (\ref{decexpoutilde}) is proved
using the same kind of computations : We remark that if $||x|-1|\geq \tilde{C} \om^{-1/2}$, then $(|x|-R_n) ^2 \geq (1-C\tilde{C})(|x|-1)^2$ for some constant $C$. This is a consequence of the fact that if $n\in \NN_{1/2}$ then $R_n-1=O(\om^{-1/2})$. With a proper choice of $\tilde{C}$ we thus have, using (\ref{eq0:DLfonction}), that for any $ r \in \R $ so that $ |r-1| \geq \tilde{C} \om^{-1/2}$ and any $n\in\NN_{1/2}$
\begin{equation}\label{decexopg1nbis}
|g_{1,n}(r)| \leq C \om^{1/4} e^{-c h_n ^2 |r-R_n|^2}\leq C \om^{1/4} e^{-\tilde{\sigma}\om |r-1|^2}.
\end{equation}
We use this fact and the same trick as in (\ref{Linftilde}) to conclude that (\ref{decexpoutilde}) holds on the domain $\left\lbrace ||x|-1|\geq \tilde{C} \om^{-1/2} \right\rbrace$. On the complement of this domain $e^{-\tilde{\sigma}\om |r-1|^2}$ is bounded below, so (\ref{decexpoutilde}) is a consequence of (\ref{Linftilde}).\\
Finally (\ref{Lputilde}) follows from (\ref{decexpoutilde}).
\end{proof}

We now aim at improving (\ref{eq3:confine1}), giving estimates in stronger norms. As already emphasized the most important result is (\ref{estimeLinf}), but we also state estimates (\ref{estimeH1}) and (\ref{estimeH2}) because they actually imply (\ref{estimeLinf}) (see Step 5 of the proof below).

\begin{proposition}[Estimates in stronger norms]\label{theo:estimsobolev}\mbox{}\\
Recall that $\tilde{u}_\om$ is defined in equation (\ref{eq2:utilde}). The following estimates hold true when $\om\rightarrow + \infty$ and $D_{\Om}\rightarrow D$:
\begin{eqnarray}
\Vert \left| u_{\om} -\tilde{u}_{\om} \right| \Vert_{H^1 (\R^2)} &\leq& C_{\ep} \om ^{\ep} \mbox{ for any } \ep >0 \label{estimeH1} \\
\Vert \left| u_{\om}  -  \tilde{u}_{\om} \right| \Vert_{H^2 (\R^2)} &\leq& C_{\ep} \om ^{1+\ep}\label{estimeH2}\mbox{ for any } \ep > 0 \\  
\Vert u_{\om}  - \tilde{u}_{\om} \Vert_{L^{\infty} (\R^2)} &\leq& C_{\ep} \om ^{\ep} \mbox{ for any } \ep >0. \label{estimeLinf}
\end{eqnarray}
\end{proposition}

The proof uses mainly the Euler-Lagrange equation for $u_{\om}$, combined with elliptic estimates for the Ginzburg-Landau operator. We state the result we are going to use for convenience and refer to \cite{LuPan} or \cite{AB-2D} for a proof.

\begin{lemma}\label{theo:LuPan}
 Let $A\in \left(W^{2,\infty}(\R^2)\right)^2$ be a divergence-free map, $g\in L^2(\R^2)$ and $w$ a solution of
\begin{equation}\label{eq4:GL}
-\left(\nabla -iA\right)^2 w=g \mbox{ in } \R^2. 
\end{equation}
Then for any $R>0$ there exists a constant $C_R>0$ so that
\begin{multline}\label{eq4:GLestim}
 \sum_{j,k=1} ^2 \int_{B_R} \left| (\partial_j-iA_j)(\partial_k-iA_k)w\right|^2 \leq \\ C_R\left(\Vert \Delta A\Vert_{L^{\infty}(B_{2R})}+ \Vert curl A\Vert_{L^{\infty}(B_{2R})}\right) \left(\int_{B_{2R}} |g|^2+|w|^2\right).
\end{multline}
Moreover, $C_R$ remains bounded as $R$ goes to infinity. 
\end{lemma}

\begin{proof}[Proof of Proposition \ref{theo:estimsobolev}]
\emph{Step 1.} We claim that  
\begin{equation}\label{eq3:confineH1mieux}
F_{\om}\left( u_{\om}- \tilde{u}_{\om} \right) \leq C\om^{1/2}
\end{equation}
which implies 
\begin{equation}\label{estimeH1'}
 \Vert u_{\om} -\tilde{u}_{\om} \Vert_{H^1 (\R^2)} \leq C \om ^{1/4}.
\end{equation}
We begin by multiplying the Euler-Lagrange equation (\ref{eq0:EELu}) by $\tilde{u}_{\om} ^*$. Integrating and injecting the result into $F_{\om}\left( u_{\om}- \tilde{u}_{\om} \right)$, we obtain
\begin{multline}\label{eq4:Energiedifference}
F_{\om}\left( u_{\om}- \tilde{u}_{\om} \right) = F_{\om}\left( u_{\om}\right) +F_{\om}\left( \tilde{u}_{\om} \right) - \mu_{\om} \int_{\R^2} \left( u_{\om}\tilde{u}_{\om}^* + u_{\om}^*\tilde{u}_{\om}\right) \\
+ G \int_{\R^2}\left( 2 |u_{\om}|^2|\tilde{u}_{\om}|^2+ u_{\om} |u_{\om}|^2 \tilde{u}_{\om}^*+u_{\om}^* |\tilde{u}_{\om}|^2 \tilde{u}_{\om} - u_{\om} \tilde{u}_{\om}  ^{*3}  - u_{\om}^* \tilde{u}_{\om}^{3} \right).
\end{multline}
We note that a consequence of (\ref{bornesupint}) is
\begin{equation}\label{borneEint}
 \int_{\R^2} |u_{\om}|^4 \leq C G \om^{1/2}.
\end{equation}
The terms in the second line of (\ref{eq4:Energiedifference}) are also bounded by $CG \om^{1/2}$ using H\"{o}lder inequalities, (\ref{Lputilde}) and (\ref{borneEint}). On the other hand, equations (\ref{eq0:resultenergie}), (\ref{estimefonda}) and (\ref{decexpoutilde}) yield
\begin{equation}\label{eq3:Energie1}
F_{\om}(u_{\om})= F_{\om}(\tilde{u}_{\om})+O(G \om^{1/2})=\lambda_{1,n^*} + O(G \om^{1/2}).
\end{equation}
From (\ref{potchim0}) and (\ref{borneEint}) we also obtain
\begin{equation}\label{potchimestim}
\mu_{\om} = \lambda_{1,n^*} + O(G \om^{1/2}). 
\end{equation}
Then, using (\ref{eq3:confine1}) we deduce (\ref{eq3:confineH1mieux}) from (\ref{eq4:Energiedifference}), remembering that in this section we assume that $G$ is a fixed constant.\\

\emph{Step 2.}
Using (\ref{Lputilde}) together with (\ref{estimeH1'}) and a Sobolev imbedding implies  
\begin{equation}\label{Lpu}
 \Vert u_{\om} \Vert_{L^p(\R ^2)} \leq C_{\ep} \om ^{1/4+\ep}
\end{equation}
for any $1 \leq p \leq +\infty$. We then interpolate between $L^2$ (remember that $\Vert u_{\om} \Vert_{L^2(\R ^2)}=1$) and $L^p$, make $p\rightarrow + \infty$ to get 
\begin{equation}\label{L6u}
 \Vert u_{\om} \Vert_{L^6(\R ^2)} \leq C_{\ep} \om ^{1/6+\ep} \mbox{ for any } \ep > 0. 
\end{equation}

\emph{Step 3.}
We now turn to the proof of (\ref{estimeH2}). The Euler-Lagrange equation for $g_{1,n}$ can be written
\begin{equation}\label{eq4:EELg}
-\left(\nabla -i\om x^{\perp} \right)^2 g_{1,n}e^{in\theta} +D_{\Om}\frac{\om^2}{2} \left(\vert x\vert^2 - 1 \right)^2 g_{1,n}e^{in\theta} =\lambda_{1,n} g_{1,n}e^{in\theta}, 
\end{equation}
so that we get for $\tilde{u}_{\om}$
\begin{equation}\label{eq4:EELutilde}
-\left(\nabla -i\om x^{\perp} \right)^2 \tilde{u}_{\om} +D_{\Om}\frac{\om^2}{2} \left(\vert x\vert^2 - 1 \right)^2 \tilde{u}_{\om} =\sum_{n\in \NN_{1/2}} \lambda_{1,n}\left\langle f_n,g_{1,n}\right\rangle g_{1,n}e^{in\theta}.
\end{equation} 
We substract this equation from the equation for $u_{\om}$ to obtain
\begin{multline}\label{eq4:EELu-utilde}
-\left(\nabla -i\om x^{\perp} \right)^2 \left(u_{\om} - \tilde{u}_{\om}\right) =  \mu_{\om} u_{\om} - \sum_{n\in \NN_{1/2}} \lambda_{1,n}\left\langle f_n,g_{1,n}\right\rangle g_{1,n}e^{in\theta}\\- D_{\Om}\frac{\om^2}{2} \left(\vert x\vert^2 - 1 \right)^2 \left(u_{\om} - \tilde{u}_{\om}\right)-2G|u_{\om}|^2 u_{\om}.
\end{multline}
We denote
\begin{equation}\label{RHS}
\delta_{\om}:= \mu_{\om} u_{\om} - \sum_{n\in \NN_{1/2}} \lambda_{1,n}\left\langle f_n,g_{1,n}\right\rangle g_{1,n}e^{in\theta}\\- D_{\Om}\frac{\om^2}{2} \left(\vert x\vert^2 - 1 \right)^2 \left(u_{\om} - \tilde{u}_{\om}\right)-2G|u_{\om}|^2 u_{\om}
\end{equation}
and provide a bound in $L^2(B_{2R})$ to this quantity, for some $R>0$. First, using (\ref{eq3:confine1}) and the exponential decay results of Propositions \ref{theo:decexpou} and \ref{theo:decexpo} we get 
\begin{equation}\label{aplLuPan1}
\int_{B_{2R}} D_{\Om}^2 \om ^4 \left(\vert x\vert^2 - 1 \right)^4 \left|u_{\om} - \tilde{u}_{\om}\right|^2 \leq C_{\ep} \om^{1+\ep}.
\end{equation}
On the other hand
\begin{eqnarray}\label{aplLuPan2}
\int_{B_{2R}}  \left|\mu_{\om} u_{\om} - \sum_{n\in \NN_{1/2}} \lambda_{1,n}\left\langle f_n,g_{1,n}\right\rangle g_{1,n}e^{in\theta}\right| ^2 &\leq& C \mu_{\om} ^2 \left( \sum_{n\in \NN_{1/2}} \int_{B_{2R}}\left|f_n - \left\langle f_n,g_{1,n}\right\rangle g_{1,n} \right|^2 + \sum_{n\in \NN_{1/2}^c} \int_{B_{2R}} | f_n |^2 \right) \nonumber
\\ &+& C \sum_{n\in \NN_{1/2}} \left(\lambda_{1,n}-\mu_{\om}\right)^2 \left| \left< f_n,g_{1,n} \right> \right|^2 \nonumber.
\end{eqnarray}
We know that $|\mu_{\om}|\leq C \om$ and $\Vert u_{\om} - \tilde{u}_{\om} \Vert_{L^2(\R^2)}\leq C\om^{-1/2}$ , whereas comparing (\ref{eq0:resultenergie}) and (\ref{eq1:dependanceN}) we get $|\lambda_{1,n}-\mu_{\om}|\leq C \max(\om ^{1/2},|n-n^*|^2)$. Using (\ref{estimefonda}) and recalling that for $n\in \NN_{1/2},\quad |n-n^*| \leq C \om^{1/2}$ we thus have
\begin{multline}\label{aplLuPan3}
\int_{B_{2R}}  \left|\mu_{\om} u_{\om} - \sum_{n\in \NN_{1/2}} \lambda_{1,n}\left\langle f_n,g_{1,n}\right\rangle g_{1,n}e^{in\theta}\right| ^2 \leq C\om ^2 \Vert u_{\om} - \tilde{u}_{\om} \Vert_{L^2(\R^2)} ^2 + C \om \sum_{n\in \NN_{1/2}} \left| \left< f_n,g_{1,n} \right> \right| ^2 \\ + C\om \sum_{n\in \NN_{1/2}}|n-n^*| ^2\left| \left< f_n,g_{1,n} \right> \right| ^2
 \leq C \om. 
\end{multline}
Gathering (\ref{L6u}), (\ref{aplLuPan1}) and (\ref{aplLuPan3}) we have for any $\ep > 0$
\[
 \int_{B_{2R}} |\delta_{\om}|^2 \leq C_{\ep} \om^{1+\ep}. 
\]
Applying Lemma \ref{theo:LuPan} with $A=\om x^{\perp}$ we find for any $R$
\begin{equation}
\sum_{j,k=1} ^2 \int_{B_R} \left| (\partial_j-iA_j)(\partial_k-iA_k) \left( u_{\om}-\tilde{u}_{\om}\right)\right|^2  \leq C_{\ep} \om^{2+\ep}
\end{equation}
with $C_{\ep}$ independent of $R$. Using the diamagnetic inequality (see \cite{LiLo}) twice, we get
\begin{equation}
 \Vert \:\nabla |u_{\om}-\tilde{u}_{\om}|\:\Vert _{H^1(B_R)}\leq C_{\ep}\om^{1+\ep} 
\end{equation}
with $C_{\ep}$ independent of $R$. This concludes the proof of (\ref{estimeH2}).\\

\bigskip
\emph{Step 4.} Let us now prove (\ref{estimeH1}). We multiply (\ref{eq4:EELu-utilde}) by $\left(u_{\om}-\tilde{u}_{\om}\right)^*$, integrate over $\R^2$ and use the diamagnetic inequality to obtain
\begin{equation}\label{H1equation}
 \Vert \left| u_{\om} -\tilde{u}_{\om} \right|\Vert_{H^1 (\R^2)} ^2 \leq 2G \int_{\R^2} |u_{\om}|^3 \left| u_{\om} -\tilde{u}_{\om} \right| + \int_{\R^2} \left|\mu_{\om} u_{\om} - \sum_{n\in \NN_{1/2}} \lambda_{1,n}\left\langle f_n,g_{1,n}\right\rangle g_{1,n}e^{in\theta}\right| \left| u_{\om} -\tilde{u}_{\om} \right|.
\end{equation}
For both terms on the right-hand side of (\ref{H1equation}) we use the Cauchy-Schwarz inequality and (\ref{eq3:confine1}), combined with (\ref{L6u}) for the first one and (\ref{aplLuPan3}) for the second. This yields (\ref{estimeH1}).

\bigskip
\emph{Step 5.}
We interpolate between $H^2$ and $H^{1+\eta}$, make $\eta \rightarrow 0$ and use a Sobolev imbedding to deduce (\ref{estimeLinf}) from (\ref{estimeH1}) and (\ref{estimeH2}).

\end{proof}

\subsection{A refined lower bound on the interaction energy}

The last essential ingredient of our analysis is a refinement of Proposition \ref{theo:borneinfint} that we present in Proposition \ref{theo:refinedlowerbound} below. Its proof uses the same ideas as that of Proposition \ref{theo:borneinfint} and requires to first improve the asymptotics for the modes $f_n$, $n\in \NN_{1/2}$. We know from (\ref{eq3:confine1}) (G is now a fixed constant) that for $n\in \NN_{1/2}$
\begin{equation}\label{fnDLpourri}
 f_n = \left\langle f_n, g_{1,n}\right\rangle g_{1,n} + O_{L^2} (\om^{-1/2}).
\end{equation}
We now have to be more precise and evaluate the term proportional to $\om^{-1/2}$ in the equation above. This is what we do in Proposition \ref{theo:modes}. The key observation is that (\ref{eq0:resultdensite}) and (\ref{uLinfmieux}) allow to linearize the Euler-Lagrange equation for $u_{\om}$, at the price of a relatively small remainder term. Thus, inserting the Fourier expansion of $u_{\om}$ we obtain equations for each mode $f_n$ that allow to improve (\ref{fnDLpourri}). It then turns out that the term proportional to $\om^{-1/2}$ in (\ref{fnDLpourri}) has to be the solution of an elliptic problem that we describe in Lemma \ref{theo:propgamman}.\\  
In this section we will often consider $g_{1,n}$ and $f_n$ (as well as any function defined on $\R^+$) as radial functions defined on $\R^2$.\\
We denote by $\Pi_{1,n}$ and $\Pi_{1,n}^{\perp}$ the othogonal projectors on the space spanned by $g_{1,n}$ and its orthogonal in $H^1 (\R^+,rdr)\cap L^2(\R^+, V_n (r)rdr)$, respectively.\\

In the following Lemma we introduce some functions $\Gamma_n$ that will be useful in the rest of the paper because they appear naturally when writing expansions for the modes $f_n$ (see Proposition \ref{theo:modes} below). We also state the main properties of these functions that we will need in our analysis. In particular the estimate (\ref{gamman-gammam}) will be essential in the proof of Proposition \ref{theo:refinedlowerbound}.

\begin{lemma}[Properties of $\Gamma_n$]\label{theo:propgamman}\mbox{}\\
Let $n\in \NN_{1/2}$. The problem
\begin{equation}\label{defigammabis}
 \begin{cases} 
 -\Delta \Gamma_n +V_n \Gamma_n -\lambda_{1,n} \Gamma_n = \Pi_{1,n} ^{\perp} \left(-2G \om ^{1/2} g_{1,n^*} ^2 g_{1,n} \right) \\ 
 \int_{\R^2} g_{1,n} \Gamma_n  = 0.
 \end{cases}
\end{equation}
has a unique radial solution in $H^1(\R^+,rdr)\cap L^2(\R^+,rV_n (r) dr)$. It satisfies the bounds
\begin{equation}\label{bornesGamma1}
\Vert \Gamma_n \Vert_{ L^{2} (\R^2)} \leq C, \quad \quad  \Vert \Gamma_n \Vert_{ H^{1} (\R^2)} \leq C \om^{1/2} 
\end{equation}
 and, for any $\ep >0$, 
\begin{equation}\label{bornesGamma2}
\Vert \Gamma_n \Vert_{ L^{4} (\R^2)} \leq C_{\ep} \om^{1/8 + \ep}, \quad  \Vert \Gamma_n \Vert_{ L^{6} (\R^2)} \leq C_{\ep} \om^{1/6+\ep}, \quad \Vert \Gamma_n \Vert_{ L^{\infty} (\R^2)} \leq C_{\ep} \om^{1/4 + \ep}.
\end{equation}
Moreover, for any $n,m \in \NN_{1/2}$ and any $\ep>0$ 
\begin{equation}\label{gamman-gammam}
\Vert \Gamma_n -\Gamma_m \Vert_{ L^{2} (\R^2)} \leq C_{\ep} \om^{-1/2+\ep} |n-m|. 
\end{equation}
\end{lemma}

\begin{proof}
It is easy to show that problem (\ref{defigammabis}) has a unique radial solution in the energy space, for example by expanding both sides of the equation on the basis $(g_{j,n})_{ j=1 ... +\infty}$.\\
Using the expansion
\begin{equation}\label{expansionGamman}
\Gamma_n = \sum_{j\geq 1} \left\langle \Gamma_n, g_{j,n}\right\rangle g_{j,n} 
\end{equation} 
combined with (\ref{EElgjn}) the first line in (\ref{defigammabis}) can be written
\begin{equation}\label{eqgammanbis}
 \sum_{j\geq 2} \left( \lambda_{j,n}-\lambda_{1,n} \right) \left\langle \Gamma_n, g_{j,n}\right\rangle g_{j,n} = \Pi_{1,n} ^{\perp} \left(-2G \om ^{1/2} g_{1,n^*} ^2 g_{1,n} \right).
\end{equation}
Thus, taking the $L^2$ norm of both sides of (\ref{eqgammanbis}) we obtain 
\begin{equation}\label{normeGamman}
 \sum_{j\geq 2} \left(\lambda_{j,n}-\lambda_{1,n}\right) ^{2} \left| \left\langle \Gamma_n, g_{j,n}\right\rangle\right|^2 \leq C\om \Vert g_{1,n^*}^2 g_{1,n} \Vert_{L^2(\R^2)}^2 \leq C \om^2
 \end{equation}
because $g_{1,n^* }^2 \leq C\om^{1/2}$ and $g_{j,n}$ is normalized in $L^2(\R^2)$ for any $j$. Now, using (\ref{eq0:DLenergie1}) and (\ref{eq0:DLenergie2}), one can find a constant $0 <\alpha < 1$ such that
\begin{equation}\label{lambdaj-lambda1}
 \lambda_{j,n}-\lambda_{1,n} = \alpha \lambda_{j,n} + (1-\alpha )\lambda_{j,n} - \lambda_{1,n} \geq \alpha \lambda_{j,n} + (1-\alpha ) \lambda_{2,n} - \lambda_{1,n} \geq \alpha \lambda_{j,n} \geq \alpha \lambda_{1,n} \geq  C \om,
\end{equation}
for any $j\geq 2$. Thus, (\ref{normeGamman}) implies 
\begin{equation}\label{normeGamman2}
\sum_{j\geq 2}  \left|\left\langle \Gamma_n, g_{j,n}\right\rangle\right|^2 \leq  C. 
\end{equation}
This yields the $L^2$ estimate in (\ref{bornesGamma1}) because of the expansion (\ref{expansionGamman}) (recall that by definition $\left\langle \Gamma_n, g_{1,n}\right\rangle = 0$). We also obtain from (\ref{normeGamman}) and (\ref{lambdaj-lambda1})
\begin{equation*}
F_n (\Gamma_n) = \sum_{j \geq 2} \lambda_{j,n} \left|\left\langle \Gamma_n, g_{j,n}\right\rangle\right|^2 \leq C \sum_{j\geq 2}  \left(\lambda_{j,n}-\lambda_{1,n}\right) \left|\left\langle \Gamma_n, g_{j,n}\right\rangle\right|^2 \leq  C  \om,  
\end{equation*}
thus the $H^1$ bound in (\ref{bornesGamma2}) holds. We have used again $\left\langle \Gamma_n, g_{1,n}\right\rangle = 0$.\\ 
We now claim that $\Gamma_n$ satisfies the pointwise estimates
\begin{equation}\label{decexpogamma}
 |\Gamma_{n}(r)|\leq C_{\ep}\om^{1/4+\ep} e^{-\hat{\sigma} \left( \frac{r-R_n}{h_n} \right)^2} 
\end{equation}
for some constant $\hat{\sigma}$, and for any $\alpha < 1$ and $r \leq C_{\alpha}$
\begin{equation}\label{dec0gamma}
 |\Gamma_{n}(r)| \leq C e^{-C \om } r^{\frac{n}{\alpha}}.
\end{equation}
The proofs are very similar to that of Proposition \ref{theo:decexpo1d}, so we give only the main ideas. First we remark that $\Gamma_n$ is radial, so for any $\eta >0$
\[
  C \om ^{1/2} \geq \Vert \Gamma_n (x) \Vert_{H^1 (\R^2,dx)} = \sqrt{2\pi} \Vert \Gamma_n (r) \Vert_{H^1 (\R^+,rdr)} \geq \sqrt{2\pi \eta} \Vert \Gamma_n (r)\Vert_{H^1 ([\eta,+\infty [,dr)}
\]
and thus the bound
\[
 \Vert \Gamma_n (r) \Vert_{L^{\infty}([\eta,+\infty [)} \leq C_{\ep} \eta^{-1/2} \om^{1/4+\ep}
\]
follows by interpolation and Sobolev imbeddings as in Step 1 of the proof of Proposition \ref{theo:decexpo1d}. We deduce from (\ref{defigammabis}) that $|\Gamma_n|^2$ satisfies
\begin{equation}\label{eqgammancarre}
 -\Delta |\Gamma_n|^2 + 2 |\Gamma_n|^2 \left( V_n - \lambda_{1,n} \right) \leq \Gamma_n \Pi_{1,n} ^{\perp} \left(-2G \om ^{1/2} g_{1,n^*} ^2 g_{1,n} \right).
\end{equation}
Using $\Pi_{1,n} ^{\perp} (f) = f - \Pi_{1,n} (f)$ and $g_{1,n^* }^2 \leq C\om^{1/2}$ we have
\[
\left| \Pi_{1,n} ^{\perp} \left(-2G \om ^{1/2} g_{1,n^*} ^2 g_{1,n} \right)\right| \leq C \left| \om^{1/2} g_{1,n^*} ^2 g_{1,n} \right| + C \left| \int_{\R ^2 } \om ^{1/2} g_{1,n^*} ^2 g_{1,n}^2  \right| |g_{1,n}| \leq C\om g_{1,n}.
\]
Defining the domain
\[
 I_n = \left\lbrace x \in \R ^2,\quad |\Gamma_n| \geq g_{1,n} \right\rbrace
\]
we thus deduce from (\ref{eqgammancarre}) that 
\[
 -\Delta |\Gamma_n|^2 + 2 |\Gamma_n|^2 \left( V_n -C\om \right) \leq 0
\]
on $I_n$. We prove that (\ref{decexpogamma}) and (\ref{dec0gamma}) hold in $I_n$ by using the maximum principle as in the proof of Proposition \ref{theo:decexpo1d}. On the other hand, (\ref{decexpogamma}) and (\ref{dec0gamma}) hold on $I_n ^c$ by definition, recalling (\ref{eq1:decexpo}) and (\ref{eq1:dec0}). The $L^p$ bounds (\ref{bornesGamma2}) follow from (\ref{decexpogamma}).\\
We now turn to the proof of (\ref{gamman-gammam}). We begin by noting that because of (\ref{eq0:DLfonction}), a change of variables $r=R_n+h_n x$ and a Taylor expansion yield
\begin{equation}\label{gn-gm}
 \Vert g_{1,n} - g_{1,m} \Vert_{L^2 (\R ^2)} \leq C \om^{-1/2}|n-m|.
\end{equation}
Thus
\[
 \left| \int_{\R^2} \Gamma_m g_{1,n} \right| = \left| \int_{\R^2} \Gamma_m \left(g_{1,m}-g_{1,n}\right) \right| \leq C \om^{-1/2} |n-m|
\]
because $\int_{\R^2} \Gamma_m g_{1,m} =0 $ and $\Gamma_m$ is bounded in $L^2$. This implies
\begin{equation}\label{orthogammanm}
\left|\int_{\R^2} \left(\Gamma_n - \Gamma_m\right) g_{1,n} \right| \leq C \om^{-1/2}|n-m|.
\end{equation}
On the other hand, denoting 
\begin{equation}\label{Snm}
S_{n,m}:= \Gamma_n-\Gamma_m 
\end{equation}
we get from (\ref{defigammabis}) the equation
\begin{multline}\label{eqgamman-m}
-\Delta S_{n,m} +V_n S_{n,m} -\lambda_{1,n} S_{n,m} = \Pi_{1,m} ^{\perp} \left(2G \om ^{1/2} g_{1,n^*} ^2 g_{1,m} \right) - \Pi_{1,n} ^{\perp} \left(2G \om ^{1/2} g_{1,n^*} ^2 g_{1,n} \right) 
\\ +\left(V_m - V_n\right)\Gamma_m +\left(\lambda_{1,m} - \lambda_{1,n} \right)\Gamma_m
\\ = \Pi_{1,m} ^{\perp} \left(2G \om ^{1/2} g_{1,n^*} ^2 \left(g_{1,m}-g_{1,n}\right) \right)+\left(\Pi_{1,m} ^{\perp}-\Pi_{1,n}^{\perp} \right)\left( 2G \om ^{1/2} g_{1,n^*} ^2 g_{1,n}  \right) 
\\ +\left(V_m - V_n\right)\Gamma_m +\left(\lambda_{1,m} - \lambda_{1,n} \right)\Gamma_m.
\end{multline}
We estimate the $L^2$ norm of the right-hand side of this equation. First remark that, using the definition (\ref{eq:Vn}) of the potential $V_n$, we have for $n,m\in \NN_{1/2}$ and $|r-1|\leq C_{\ep} \om^{-1/2+\ep}$
\begin{multline}\label{Vn-Vm}
|V_m(r)-V_n(r)| = \left|\frac{m^2 -n^2 }{r}-2\om (m-n)\right| = |n-m|\left|m+n-2\om +(m+n)O(\om^{-1/2+\ep})\right|
\\ \leq C_{\ep} \om^{1/2+\ep}|m-n|.
\end{multline}
We then write 
\[
 \int_{\R^2} \left(V_m - V_n\right)^2 \Gamma_m ^2 = \int_{|r-1|\leq C\om ^{-1/2+\ep}} \left(V_m - V_n\right)^2 \Gamma_m ^2 + \int_{|r-1| > C\om ^{-1/2+\ep}} \left(V_m - V_n\right)^2 \Gamma_m ^2.
\]
We estimate the first term using (\ref{Vn-Vm}) and the fact that $\Gamma_m$ is bounded in $L^2$. The second one is exponentially small because of (\ref{decexpogamma}) and (\ref{dec0gamma}) (recall that $R_n = 1 + O(\om^{-1/2})$ and $h_n \propto \om^{-1/2}$). The result is
\begin{equation}\label{estimn1}
\Vert \left(V_m - V_n\right) \Gamma_m \Vert_{L^{2}(\R^2)} \leq C_{\ep} \om^{1/2+\ep} |n-m|. 
\end{equation}
Also, using (\ref{eq1:dependanceN}), we have for $n,m \in \NN_{1/2}$
\begin{equation}\label{lambdan-m}
\left|\lambda_{1,m} - \lambda_{1,n} \right| \leq C \om^{1/2}|n-m|. 
\end{equation}
On the other hand
\begin{equation}\label{estimn2}
\left\Vert \Pi_{1,m} ^{\perp} \left(2G \om ^{1/2} g_{1,n^*} ^2 \left(g_{1,m}-g_{1,n}\right) \right) \right\Vert_{L^{2}(\R^2)} \leq \left\Vert 2G \om ^{1/2} g_{1,n^*} ^2 \left(g_{1,m}-g_{1,n} \right)\right\Vert_{L^{2}(\R^2)} \leq C\om^{1/2} |n-m|
\end{equation}
because of (\ref{gn-gm}). We now use the fact that $\Pi_{1,m} ^{\perp}-\Pi_{1,n}^{\perp} = \Pi_{1,n} -\Pi_{1,m}$ to obtain
\[
 \left(\Pi_{1,m} ^{\perp}-\Pi_{1,n}^{\perp} \right)\left( 2G \om ^{1/2} g_{1,n^*} ^2 g_{1,n}  \right)  = \left( \int_{\R^2} 2G g_{1,n^* } ^2 g_{1,n} ^2\right) \left( g_{1,n}-g_{1,m} \right) + \left( \int_{\R^2} 2G g_{1,n^* } ^2 g_{1,n} \left( g_{1,n}-g_{1,m} \right) \right) g_{1,m}
\]
and thus, using (\ref{gn-gm}) again, plus the Cauchy-Schwarz inequality for the second term
\begin{equation}\label{estimn3}
\left\Vert \left(\Pi_{1,m} ^{\perp}-\Pi_{1,n}^{\perp} \right)\left( 2G \om ^{1/2} g_{1,n^*} ^2 g_{1,n}  \right) \right\Vert_{L^{2}(\R^2)}  \leq C \om^{1/2} |n-m|.   
\end{equation}
Arguing exactly as when we obtained (\ref{normeGamman}) from (\ref{defigammabis}), the combination of (\ref{eqgamman-m}) and the estimates (\ref{estimn1}) to (\ref{estimn3}) yield
\begin{equation}\label{normeSnm}
\sum_{j\geq 2} \left(\lambda_{j,n} - \lambda_{1,n}\right) ^2 \left|\left\langle S_{n,m}, g_{j,n}\right\rangle\right|^2 \leq C_{\ep} \om ^{1+\ep} |n-m|^2.
\end{equation}
Using (\ref{lambdaj-lambda1}) we thus have
\[
 \sum_{j\geq 2} \left|\left\langle S_{n,m}, g_{j,n}\right\rangle\right|^2\leq C_{\ep} \om ^{-1+\ep} |n-m|^2
\]
and there only remains to recall (\ref{orthogammanm}) to conclude the proof because $(g_{j,n})_{j=1...+\infty}$ is a Hilbert basis for $H^1 (\R^+, rdr)\cap L^2 (\R^+ , V_n rdr)$. 
\end{proof}

We now show that the second order in (\ref{fnDLpourri}) is given, after an appropriate normalization, by the function $\Gamma_n$ that we have introduced in Lemma \ref{theo:propgamman}.\\
Let us first introduce some notation:
\begin{eqnarray}\label{defiN1/4}
\mathcal{N}_{1/4}&=& \left\lbrace n\in \Z, |n-\om|\leq \om^{1/4} \right\rbrace \\
\mathcal{N}_{1/4,+}&=&\left\lbrace n\in \NN_{1/4}, \left\langle f_n, g_{1,n}\right\rangle \neq 0 \right\rbrace \\
\mathcal{N}_{1/2,+}&=&\left\lbrace n\in \NN_{1/2}, \left\langle f_n, g_{1,n}\right\rangle \neq 0 \right\rbrace  
\end{eqnarray}
In the rest of the paper we will always denote $\beta_n$ the $n$-th term of a generic sequence satisfying   
\begin{equation}\label{betanbis}
\sum_{n\in \Z}  \beta_n^2 \leq C.
\end{equation}
The actual value of the quantity $\beta_n$ may thus change from line to line. Note that $\beta_n$ can depend on $\om$ and $D_{\Om}$ as long as (\ref{betanbis}) is satisfied.

\begin{proposition}[Refined asymptotics for the mode $f_n$]\label{theo:modes}\mbox{}\\
Let $n\in \NN_{1/2,+}$. We have
\begin{equation}\label{fnRafin}
\left\Vert f_n - \left\langle f_n, g_{1,n}\right\rangle \left( g_{1,n} + \om^{-1/2} \Gamma_n \right) \right\Vert_{L^2(\R^+ , rdr)} \leq \beta_n \max\left(\om^{-3/4}, |n-n^*|^2 \om^{-3/2} \right) 
\end{equation}
and
\begin{equation}\label{fnRafinenergie}
F_n \left( f_n - \left\langle f_n, g_{1,n}\right\rangle \left( g_{1,n} + \om^{-1/2} \Gamma_n \right) \right) \leq \beta_n ^2 \max\left(\om^{-1/2},|n-n^*|^4 \om^{-2} \right)
\end{equation}
where $\Gamma_n$ is defined in Lemma \ref{theo:propgamman}.
\end{proposition}

\begin{proof}[Proof of Proposition \ref{theo:modes}]
We need an equation satisfied by $f_n$. First recall that (\ref{decexpoutilde}) and (\ref{estimeLinf}) imply that 
\begin{equation}\label{Linfraf}
\Vert u_{\om} \Vert_{L^{\infty}(\R^2)}\leq C \om^{1/4}. 
\end{equation}
Then, using (\ref{eq0:resultdensite}) ($G$ is now a fixed constant) we get
\[
\int_{\R^2} \left(|u_{\om}|^2u_{\om} - g_{1,n^*}^2 u_{\om} \right) ^2 \leq C \om ^{1/2} \int_{\R^2} \left(|u_{\om}|^2 - g_{1,n^*}^2 \right) ^2 \leq C \om ^{1/2}.
\]
This allows to rewrite the Euler-Lagrange equation for $u_{\om}$ as follows:
\begin{equation}\label{ELLulin}
-\left(\nabla -i\om x^{\perp} \right)^2 u_{\om} +D_{\Om}\frac{\om^2}{2} \left(\vert x\vert^2 - 1 \right)^2 u_{\om} + 2G  g_{1,n^*}  ^2 u_{\om}=\mu _{\om} u_{\om} + L, 
\end{equation} 
where 
\begin{equation}\label{L}
 L = 2G\left( g_{1,n^*}^2 u_{\om} - |u_{\om}|^2u_{\om} \right)
\end{equation}
satisfies $\Vert L \Vert_{L^2}\leq C \om^{1/4}$. Using the Fourier expansion (\ref{eq0:Fourier}) of $u_{\om}$ we can write (\ref{ELLulin}) as 
\begin{equation}\label{ELLulin2}
\sum_n e^{in \theta} \left(-\Delta f_n + V_n(r) f_n +2G  g_{1,n^*}  ^2 f_n \right) = \sum_n e^{in\theta }\mu_{\om} f_n + L 
\end{equation} 
We multiply this equation by $e^{-in\theta}$ and integrate over $\theta$ to get 
\begin{equation}\label{equationfn}
-\Delta f_n  + V_n(r) f_n + 2G  g_{1,n^*} ^2 f_n = \mu_{\om} f_n + L_n 
\end{equation}
where 
\begin{equation}\label{Ln}
 L_n (r) = \int_{0} ^{2\pi} e^{-in\theta} L (r,\theta) d\theta,
\end{equation}
so that 
\begin{equation}\label{normeLn}
\sum_{n} \Vert L_n \Vert_{L^2(\R^+,rdr)} ^2 \leq C \Vert L \Vert_{L^2(\R^2)} ^2 \leq C \om^{1/2}.
\end{equation}
For $n\in \NN_{1/2,+}$ we write $f_n$ as 
\begin{equation}\label{develofn}
 f_n = \left\langle f_n, g_{1,n}\right\rangle \left( g_{1,n} + \om^{-1/2} \kappa_n \right)
\end{equation}
where
\begin{equation}\label{kappaortho}
\left\langle\kappa_n,g_{1,n} \right\rangle = 0 .
\end{equation}
Remembering (\ref{eq3:confine1}) we have  
\begin{equation}\label{kappanorm}
\sum_{n\in \NN_{1/2,+}} \left| \left\langle f_n, g_{1,n}\right\rangle \right| ^2 \Vert \kappa_n \Vert_{L^2(\R^+,rdr)} ^2 \leq C. 
\end{equation}
We claim that 
\begin{equation}\label{estimkappagamma}
\Vert \left\langle f_n, g_{1,n}\right\rangle \left(\kappa_n - \Gamma_n \right) \Vert_{L^2 (\R^2) } \leq \beta_n \max \left(\om^{-1/4},|n-n^*|^2 \om^{-1} \right)
\end{equation}
and
\begin{equation}\label{energiekappagamma}
F_n\left( \left\langle f_n, g_{1,n}\right\rangle \left(\kappa_n - \Gamma_n \right) \right) \leq \beta_n \max \left( \om^{1/2}, |n-n^*|^4\om^{-1} \right)
\end{equation}
with $\beta_n$ satisfying (\ref{betanbis}). Let us denote
\begin{equation}\label{Dnm}
 D_n :=\left\langle f_n, g_{1,n}\right\rangle \left(\kappa_n - \Gamma_n \right)
\end{equation}
Inserting (\ref{develofn}) into (\ref{equationfn}), using the fact that $g_{1,n}$ is a normalized eigenfunction of $-\Delta + V_n$ for the eigenvalue $\lambda_{1,n}$, we obtain an equation for $\left\langle f_n, g_{1,n}\right\rangle \kappa_n$:
\begin{multline}\label{eqKappan}
 -\Delta \left\langle f_n, g_{1,n}\right\rangle \kappa_n + V_n \left\langle f_n, g_{1,n}\right\rangle \kappa_n - \lambda_{1,n}\left\langle f_n, g_{1,n}\right\rangle \kappa_n = \left(\mu_{\om}-\lambda_{1,n} \right)\left\langle f_n, g_{1,n}\right\rangle \kappa_n -2 G g_{1,n^*} ^2 \left\langle f_n, g_{1,n}\right\rangle \kappa_n \\ 
-2 G g_{1,n^*} ^2 \left\langle f_n, g_{1,n}\right\rangle \om^{1/2} g_{1,n} + \left(\mu_{\om} - \lambda_{1,n} \right) \om^{1/2} \left\langle f_n, g_{1,n}\right\rangle g_{1,n}  + \om^{1/2} L_n.
\end{multline}
We multiply this equation by $g_{1,n}$ and integrate to get the useful relation
\begin{multline}\label{orthoutile}
\om^{1/2}\left\langle f_n, g_{1,n}\right\rangle \left( 2G \int_{\R ^2} g_{1,n^*} ^2 g_{1,n}^2 - \left(\mu_{\om} - \lambda_{1,n}\right)\right) =  \om ^{1/2} \int_{\R^2 } L_n g_{1,n} - 2 G \left\langle f_n, g_{1,n}\right\rangle \int_{\R^2} g_{1,n^*} ^2 \kappa_n g_{1,n} 
\\= O(\beta_n \om ^{3/4}). 
\end{multline}
We have used Cauchy-Schwarz, the bound $g_{1,n^*}^2 \leq C \om^{1/2}$, (\ref{normeLn}) and (\ref{kappanorm}).
Substracting the equation for $\left\langle f_n, g_{1,n}\right\rangle \Gamma_n$ from (\ref{eqKappan}), we obtain
\begin{multline}\label{eqDn}
 -\Delta D_n +V_n D_n - \lambda_{1,n} D_n = \left(\mu_{\om} - \lambda_{1,n}\right)\left\langle f_n, g_{1,n}\right\rangle \kappa_n -2G \left\langle f_n, g_{1,n}\right\rangle g_{1,n^*} ^2 \kappa_n \\
+ \left\langle f_n, g_{1,n}\right\rangle \Pi_{1,n} \left( \left(\mu_{\om} - \lambda_{1,n} \right) \om^{1/2}  g_{1,n} -2 G g_{1,n^*} ^2 \om^{1/2} g_{1,n} \right)+ \om^{1/2} L_n .
\end{multline}
Using (\ref{eq0:resultenergie}), (\ref{eq1:dependanceN}), (\ref{potchim0}) and (\ref{borneEint}) we obtain $|\mu_{\om} - \lambda_{1,n}|\leq C \max(\om^{1/2},|n-n^*|^2)$ and thus 
\begin{equation}\label{estimDn1}
\left\Vert \left(\mu_{\om} - \lambda_{1,n}\right)\left\langle f_n, g_{1,n}\right\rangle \kappa_n \right\Vert_{L^2 (\R^2)} \leq  \beta_n \max\left(\om^{1/2}, |n-n^*|^2 \right)
\end{equation}
for any $n\in \NN_{1/2,+}$. We have used (\ref{kappanorm}) and $n^* = \om +O(1)$. Also, using $g_{1,n^*}^2\leq C\om^{1/2}$ and (\ref{kappanorm}),
\begin{equation}\label{estimDn2}
\left\Vert 2G \left\langle f_n, g_{1,n}\right\rangle g_{1,n^*} ^2 \kappa_n \right\Vert_{L^2 (\R^2)} \leq \om^{1/2} \beta_n.
\end{equation}
On the other hand (\ref{orthoutile}) implies
\begin{eqnarray}\label{estimDn3}
 \left\Vert \left\langle f_n, g_{1,n}\right\rangle \Pi_{1,n} \left( \left(\mu_{\om} - \lambda_{1,n} \right) \om^{1/2}  g_{1,n} -2 G g_{1,n^*} ^2 \om^{1/2} g_{1,n} \right) \right\Vert_{L^2 (\R^2)} &=& \nonumber
\\ \left| \om^{1/2}\left\langle f_n, g_{1,n}\right\rangle \left( \left(\mu_{\om} - \lambda_{1,n}\right) -2G \int_{\R ^2} g_{1,n^*} ^2 g_{1,n}^2 \right) \right| \left\Vert g_{1,n} \right\Vert_{L^2 (\R^2)}
&\leq& \beta_n \om^{3/4}
\end{eqnarray}
We compute the $L^2$ norm of both sides of (\ref{eqDn}) after having expanded the left-hand side on the basis $(g_{j,n})_{j=1...+\infty}$ as in (\ref{eqgammanbis}). Gathering equations (\ref{estimDn1}) to (\ref{estimDn3}), recalling (\ref{normeLn}) we obtain 
\begin{equation}\label{normeDn}
 \sum_{j\geq 2} \left(\lambda_{j,n}-\lambda_{1,n}\right)^{2} \left| \left\langle D_n, g_{j,n}\right\rangle\right|^2 \leq  C  \beta_n ^2 \max \left( \om^{3/2}, |n-n^*|^4 \right)
 \end{equation}
where $\lambda_{j,n}$ and $g_{j,n}$ are the eigenvalues and eigenfunctions of the restriction of $-\Delta + V_n$ to the space of radial functions. Recalling (\ref{lambdaj-lambda1}), (\ref{normeDn}) implies 
\begin{equation*}
\sum_{j\geq 2}  \left|\left\langle D_n, g_{j,n}\right\rangle\right|^2 \leq  C  \beta_n ^2  \om ^{-2} \max \left( \om^{3/2}, |n-n^*|^4 \right)  
\end{equation*}
which is (\ref{estimkappagamma}) because $\left\langle D_n, g_{1,n}\right\rangle = 0$ (see the second equation in (\ref{defigammabis}) and (\ref{kappaortho})). We also obtain from (\ref{normeDn}) and (\ref{lambdaj-lambda1})
\begin{equation*}
F_n (D_n) = \sum_{j \geq 2} \lambda_{j,n} \left|\left\langle D_n, g_{j,n}\right\rangle\right|^2 \leq C \sum_{j\geq 2}  \left(\lambda_{j,n}-\lambda_{1,n}\right) \left|\left\langle D_n, g_{j,n}\right\rangle\right|^2 \leq  C  \beta_n ^2  \om ^{-1} \max \left( \om^{3/2}, |n-n^*|^4 \right),  
\end{equation*}
thus (\ref{energiekappagamma}) holds. We have used again $\left\langle D_n, g_{1,n}\right\rangle = 0$. 
\end{proof}

We are now able to present the main result of this section, a refinement of the lower bound on the interaction energy. Its proof uses the same as ideas as that of Proposition \ref{theo:borneinfint}, but the next order in the lower bound is now computed. This is made possible because of the refined asymptotics of Proposition \ref{theo:modes}.\\
Note however that the right-hand side of (\ref{fnRafin}) increases with $|n-n^*|$ so the refinement is not equally efficient for all modes. On the other hand (\ref{estimefonda}) states that the larger $|n-n^*|$ is, the smaller is the mass of the corresponding mode, so the smaller is its contribution to the interaction energy. It turns out that the best lower bound available with our results can be computed by using the refined results of Proposition \ref{theo:modes} only when $n\in \NN_{1/4,+}$.

\begin{proposition}[Refined lower bound on the interaction energy]\label{theo:refinedlowerbound}\mbox{}\\
We have, for any $\ep > 0$,
\begin{multline}\label{eq4:rlbound}
\int_{\R^2}|u_{\om}|^4 \geq \sum_{n\in \NN_{1/4,+}} \left|\left\langle f_n,g_{1,n} \right\rangle\right|^2  \int_{\R^+} \left(g_{1,n} +\om^{-1/2} \Gamma_n \right)^4 rdr 
\\ + \sum_{n\in \NN_{1/2}\setminus \NN_{1/4,+}}  \left|\left\langle f_n,g_{1,n} \right\rangle\right|^2  \int_{\R^+} g_{1,n}^4 rdr -C_{\ep}\om^{-1/4+\ep}. 
\end{multline}
 
\end{proposition}

\begin{proof}
We begin by using Lemma \ref{theo:interactions}:
\begin{equation}\label{debut}
\int_{\R^2} |u_{\om}|^4 \geq 2\pi \sum_{p,q\in \NN_{1/2}} \int_{\R^+} |f_p|^2 |f_q|^2rdr
\end{equation}
and recall that we can restrict the integration domains to $1/2 \leq r \leq 3/2 $ at the price of an exponentially small remainder (see (\ref{reducdomaine})). We implicitly do this restriction in the rest of the proof.\\
We continue to denote $\beta_n$ the $n$-th term of a generic sequence satisfying (\ref{betanbis}) and note that for any $n\in \NN_{1/2}$ we have 
\[
 \Vert f _n \Vert_{L^2([1/2,3/2],dr)} \leq C \Vert f _n \Vert_{L^2(\R^+, rdr)}\leq \beta_n
\]
because of the $L^2$ normalization of $u_{\om}$. Moreover
\[
 \Vert f_n \Vert_{H^1([1/2,3/2],dr)} \leq C \Vert f _n \Vert_{H^1(\R^+, rdr)}\leq \beta_n \om^{1/2}
\]
follows from the bounds $\sum_n F_n (f_n) \leq F_{\om} (u_{\om}) \leq C \om$. Interpolating between these two estimates we get
\[
 \Vert f_n \Vert_{H^{1/2}([1/2,3/2],dr)} \leq \beta_n \om^{1/4}. 
\]
Then we obtain, for any $\ep >0$
\begin{equation}\label{fnL6}
 \Vert f _n \Vert_{L^6([1/2,3/2],dr)} \leq C_{\ep} \beta_n \om^{1/6+\ep}
\end{equation}
by the Sobolev imbedding of $H^{1/2}$ in $L^p$ for any $p$ and interpolation between $L^2$ and $L^{p}$ with $p\rightarrow +\infty$.\\
If $n\notin \NN_{1/4}$ this estimate is improved:
\begin{equation}\label{fnL6mieux}
  \Vert f _n \Vert_{L^6([1/2,3/2],dr)} \leq C_{\ep} \beta_n \om^{-1/12+\ep}.
\end{equation}
Indeed, for $n\notin \NN_{1/4}$ we have 
\begin{equation}\label{L6mieuxpreuve1}
 \Vert f _n \Vert_{L^2([1/2,3/2],dr)} \leq C \Vert f _n \Vert_{L^2(\R^+, rdr)}\leq \beta_n \om^{-1/4}
\end{equation}
using (\ref{estimefonda}). Moreover the expansion
\begin{equation}\label{L6mieuxpreuve2}
 F_n(f_n) = \sum_{j\geq 1} \lambda_{j,n} \left| \left\langle f_n, g_{j,n} \right\rangle\right|^2
\end{equation}
combined with (\ref{souslecoude}), $\left| \left\langle f_n, g_{1,n} \right\rangle\right|^2 \leq \beta_{n} ^2 \om^{-1/2}$ and $\lambda_{1,n} = O(\om)$ yields
\begin{equation}\label{L6mieuxpreuve3}
 \Vert f_n \Vert_{H^1([1/2,3/2],dr)} \leq C \Vert f _n \Vert_{H^1(\R^+, rdr)}\leq C \sqrt{F_n (f_n)} \leq \beta_n \om^{1/4}.
\end{equation}
Interpolating between (\ref{L6mieuxpreuve1}) and (\ref{L6mieuxpreuve3}) we obtain (\ref{fnL6mieux}). We note that there also holds, for any $n \in \NN_{1/2}$,
\begin{equation}\label{gnL6}
\Vert \left\langle f_n,g_{1,n} \right\rangle g_{1,n} \Vert_{L^6([1/2,3/2],dr)} \leq C \beta_n \om^{1/6}
\end{equation}
and for any $n\notin \NN_{1/4,+}$
\begin{equation}\label{gnL6mieux}
\Vert \left\langle f_n,g_{1,n} \right\rangle g_{1,n} \Vert_{L^6([1/2,3/2],dr)} \leq C \beta_n \om^{-1/12}.
\end{equation}
These last two estimates are basic consequences of the properties of $g_{1,n}$, see Corollary \ref{theo:Lpg1n}.\\
We now bound from below the $p$-th term of the sum (\ref{debut}), distinguishing two cases:\\
\textbf{Case 1: $p\notin \NN_{1/4,+}$}\\
We have, for any $q\in \NN_{1/2}$
\begin{multline}\label{cas11}
\left| \int \left(|f_p|^2 |f_q|^2 - \left|\left\langle f_p,g_{1,p} \right\rangle\right|^2 \left|\left\langle f_q,g_{1,q} \right\rangle\right|^2 g_{1,p} ^2 g_{1,q} ^2 \right) rdr \right| \leq 
\\C\Vert f_q - \left\langle f_q,g_{1,q} \right\rangle g_{1,q}\Vert_{L^2}\Vert f_p \Vert_{L^6} ^2 \left(\Vert f_q \Vert_{L^6}  +  \Vert \left\langle f_q,g_{1,q} \right\rangle g_{1,q} \Vert_{L^6} \right)
\\ + C\Vert f_p - \left\langle f_p,g_{1,p} \right\rangle g_{1,p}\Vert_{L^2}\Vert \left\langle f_q,g_{1,q} \right\rangle g_{1,q} \Vert_{L^6} ^2 \left(\Vert f_p \Vert_{L^6}  +  \Vert \left\langle f_p,g_{1,p} \right\rangle g_{1,q} \Vert_{L^6} \right)
\\ \leq C_{\ep}\om^{-1/4+\ep} \beta_p ^2 \beta _q ^2  
\end{multline}
using (\ref{eq3:confine1}) and the estimates (\ref{fnL6}) to (\ref{gnL6mieux}). We then use exactly the same technique as in the proof of Proposition \ref{theo:borneinfint} to obtain that for any $q\in \NN_{1/2}$
\begin{equation}\label{cas12}
 \int g_{1,p} ^2 g_{1,q} ^2 rdr \geq \int g_{1,p} ^4 rdr- C \om^{-1/2}|p-q|^2 -C.
\end{equation}
We conclude from (\ref{cas11}) and (\ref{cas12}) that for any $p\notin_{\NN_{1/4,+}}$
\begin{multline}\label{cas13}
\sum_{q\in \NN_{1/2}} 2\pi \int |f_p|^2 |f_q|^2 rdr \geq \sum_{q\in \NN_{1/2}} \left|\left\langle f_p,g_{1,p} \right\rangle\right|^2 \left|\left\langle f_q,g_{1,q} \right\rangle\right|^2 \left( 2\pi \int g_{1,p} ^4 rdr- C \om^{-1/2}|p-q|^2 -C \right)  
\\ -C_{\ep} \beta_p ^2 \om^{-1/4+\ep}.
\end{multline}
\\

\textbf{Case 2: $p\in \NN_{1/4,+}$}
We again distinguish two cases.\\

\textit{Case 2.1: $q\notin \NN_{1/4,+}$}
We argue as in Case 1. Using estimates (\ref{fnL6}) to (\ref{gnL6mieux}), 
\begin{equation}\label{cas21}
\left| \int \left(|f_p|^2 |f_q|^2 - \left|\left\langle f_p,g_{1,p} \right\rangle\right|^2 \left|\left\langle f_q,g_{1,q} \right\rangle\right|^2 g_{1,p} ^2 g_{1,q} ^2 \right) rdr \right| \leq C_{\ep}\om^{-1/4+\ep} \beta_p ^2 \beta _q ^2  
\end{equation}
because $q\notin \NN_{1/4,+}$. Also 
\begin{equation}\label{cas22}
 \int g_{1,p} ^2 g_{1,q} ^2 rdr \geq \int g_{1,p} ^4 rdr - C \om^{-1/2}|p-q|^2 -C.
\end{equation}
and it is clear from (\ref{eq0:DLfonction}) and the estimates in Lemma \ref{theo:propgamman} that for any $\ep > 0$
\begin{equation}\label{cas23}
 \int g_{1,p} ^4 rdr = \int \left(g_{1,p} + \om^{-1/2}\Gamma_p \right) ^4 rdr + O(\om^{\ep}).
\end{equation}
Gathering (\ref{cas21}), (\ref{cas22}) and (\ref{cas23}) we thus have for any $p\in \NN_{1/4,+}$
\begin{multline}\label{cas24}
\sum_{q\notin \NN_{1/4,+}} 2\pi \int |f_p|^2 |f_q|^2 rdr \geq \sum_{q\notin \NN_{1/4,+}} \left|\left\langle f_p,g_{1,p} \right\rangle\right|^2 \left|\left\langle f_q,g_{1,q} \right\rangle\right|^2 \left( 2\pi \int \left(g_{1,p} + \om^{-1/2}\Gamma_p \right) ^4 rdr \right. 
\\ \left. - C \om^{-1/2}|p-q|^2 -C_{\ep} \om^{\ep} \right)   -C_{\ep} \beta_p ^2 \om^{-1/4+\ep}. 
\end{multline}

\textit{Case 2.2: $q\in \NN_{1/4,+}$}
This is the case where the refined asymptotics of Proposition \ref{theo:modes} will be crucial. Note that for $n\in \NN_{1/4,+}$ (\ref{fnRafin}) becomes
 \begin{equation}\label{fnRafinplus}
\left\Vert f_n - \left\langle f_n, g_{1,n}\right\rangle \left( g_{1,n} + \om^{-1/2} \Gamma_n \right) \right\Vert_{L^2(\R^+ , rdr)} \leq \beta_n \om^{-3/4}.
\end{equation}
 We begin with
\begin{multline}\label{cas25}
\left| \int \left(|f_p|^2 |f_q|^2 - \left|\left\langle f_p,g_{1,p} \right\rangle\right|^2 \left|\left\langle f_q,g_{1,q} \right\rangle\right|^2 \left(g_{1,p}+\om^{-1/2}\Gamma_p\right) ^2 \left(g_{1,q}+\om^{-1/2}\Gamma_q\right) ^2 \right) rdr \right| \leq 
\\C\left\Vert f_q - \left\langle f_q,g_{1,q} \right\rangle \left(g_{1,q}+\om^{-1/2}\Gamma_q\right)\right \Vert_{L^2} \Vert f_p \Vert_{L^6} ^2 \left(\Vert f_q \Vert_{L^6}  +  \left\Vert \left\langle f_q,g_{1,q} \right\rangle \left(g_{1,q}+\om^{-1/2}\Gamma_q\right) \right\Vert_{L^6} \right)
\\ + C\left \Vert f_p - \left\langle f_p,g_{1,p} \right\rangle \left(g_{1,p}+\om^{-1/2}\Gamma_p\right)\right \Vert_{L^2} \left \Vert \left\langle f_q,g_{1,q} \right\rangle \left(g_{1,q}+\om^{-1/2}\Gamma_q\right) \right \Vert_{L^6} ^2 
\\ \times \left(\Vert f_p \Vert_{L^6}  +  \left\Vert \left\langle f_p,g_{1,p} \right\rangle \left(g_{1,p}+\om^{-1/2}\Gamma_p\right) \right\Vert_{L^6} \right) \leq C_{\ep}\om^{-1/4+\ep} \beta_p ^2 \beta _q ^2 . 
\end{multline}
We have used (\ref{fnL6}) and (\ref{gnL6}) along with (\ref{fnRafinplus}) and the estimates of Lemma \ref{theo:propgamman}.\\
Now we use the bounds $\Vert g_{1,p}\Vert_{L^4} \leq C\om^{1/8}$ (see Corollary \ref{theo:Lpg1n}) and $\Vert \Gamma_p \Vert_{L^4} \leq C_{\ep}\om^{1/8+\ep}$ (see Lemma \ref{theo:propgamman}), to obtain 
\begin{multline}\label{cas26}
\int \left(g_{1,p}+\om^{-1/2}\Gamma_p\right) ^2 \left(g_{1,q}+\om^{-1/2}\Gamma_q\right) ^2 rdr \geq \int g_{1,p} ^2 g_{1,q}^2 rdr + 2\om^{-1/2} \int g_{1,p}^2 \Gamma_q g_{1,q} rdr 
\\+ 2\om^{-1/2}\int g_{1,q}^2 \Gamma_p g_{1,p} rdr - C_{\ep} \om^{-1/2+\ep}.
\end{multline}
Using (\ref{eq0:rescaling}) and (\ref{eq0:DLfonction}) we compute 
\begin{equation}\label{eq4:calculsint11}
2\pi \int g_{1,p} ^2 g_{1,q} ^2 rdr \geq h_p c_{1,p} ^2 c_{1,q} ^2 I_{p,q} 
\end{equation}
with a change of variables $r=R_p + h_p x$ where
\[
I_{p,q} = J_{p,q} + 2 K_{p,q} + 2 L_{p,q} + O(\om^{-1}) 
\]
and
\begin{eqnarray}\label{Ipq}
J_{p,q}&:=& \int_{\R^+} \xi_{1} ^2 (x) \xi_1 ^2\left(\frac{h_{p}}{h_q} x+\frac{R_p-R_q}{h_ q} \right) \left(R_{p}+h_{p}x  \right)dx\\
K_{p,q}&:=& \int_{\R^+} \xi_{1} ^2 (x) h_q P_{q}\left(\frac{h_{p}}{h_q}x+\frac{R_p-R_q}{h_ q}  \right)\xi_1 ^2 \left(\frac{h_{p}}{h_q}x +\frac{R_p-R_q}{h_ q} \right) \left(R_{p}+h_{p}x  \right)dx \\
L_{p,q}&:=& \int_{\R^+} h_p \xi_{1} ^2 (x) P_{q} \left(x\right)\xi_1^2 \left(\frac{h_{p}}{h_q}x +\frac{R_p-R_q}{h_ q} \right) \left(R_{p}+h_{p}x  \right)dx.
\end{eqnarray}
We recall that for any $p,q\in \NN_{1/2}$
\[
 |h_p-h_q| \propto \om^{-3/2}|p-q|
\]
and 
\[
 |R_p-R_q|\propto \om^{-1} |p-q|.
\]
Then arguing exactly as in the proof of Proposition \ref{theo:borneinfint} we get 
\[
 J_{p,q} = J_{p,p} +O(\om^{-1} |p - q|^2)
\]
whereas simple Taylor inequalities coupled with the estimates (\ref{DLxi2}) yield
\[
 K_{p,q} = K_{p,p} + O(\om^{-1} |p - q|) \quad L_{p,q} = L_{p,p} + O(\om^{-1} |p - q|).
\]
Moreover, (\ref{eq0:DLfonction}) and (\ref{DLxi2}) imply that for any $p,q\in \NN_{1/2}$
\[
 | c_{1,p} ^2 - c_{1,q}^2 | \leq C \om^{-1} |p-q|. 
\]
All in all,
\begin{equation}\label{cas27}
2\pi \int g_{1,p} ^2 g_{1,q} ^2 rdr \geq 2\pi \int g_{1,p} ^4 - C \om^{-1/2}|p-q|^2 - C\om^{-1/2}|p-q| -C\om^{-1/2}.
\end{equation}
On the other hand 
\begin{eqnarray}\label{cas28}
 \left| \int g_{1,p}^2 \Gamma_q g_{1,q} - \int g_{1,p}^3 \Gamma_p \right| &\leq& \left(\int g_{1,p}^4 \Gamma_q ^2 \right)^{1/2} \left(\int \left( g_{1,q}-g_{1,p}\right)^2 \right)^{1/2} + \left(\int g_{1,p}^6 \right)^{1/2} \left(\int \left(\Gamma_p - \Gamma_q \right) ^2 \right)^{1/2}\nonumber
\\ &\leq& \Vert g_{1,p}\Vert_{L ^6 } ^2 \Vert \Gamma_q \Vert_{L ^ 6} \Vert g_{1,p}-g_{1,q} \Vert_{L ^ 2} + \Vert g_{1,p} \Vert_{L ^ 6} ^3 \Vert \Gamma_p - \Gamma_q \Vert_{L ^ 2}\nonumber
\\ &\leq& C_{\ep} \om^{\ep} |p-q|
\end{eqnarray}
using Corollary \ref{theo:Lpg1n}, (\ref{bornesGamma2}), (\ref{gamman-gammam}) and (\ref{gn-gm}).\\
We conclude from (\ref{cas25}), (\ref{cas26}), (\ref{cas27}) and (\ref{cas28}) that for any $p\in \NN_{1/4,+}$ 
\begin{multline}\label{cas29}
\sum_{q \in \NN_{1/4,+}} 2\pi \int |f_p|^2 |f_q|^2 rdr \geq \sum_{q \in \NN_{1/4,+}} \left|\left\langle f_p,g_{1,p} \right\rangle\right|^2 \left|\left\langle f_q,g_{1,q} \right\rangle\right|^2 \left( 2\pi \int \left(g_{1,p} + \om^{-1/2}\Gamma_p \right) ^4 rdr \right. 
\\ \left. - C \om^{-1/2}|p-q|^2 -C\om^{-1/2}|p-q| -C\om^{-1/2} \right)   -C_{\ep} \beta_p ^2\om^{-1/4+\ep}. 
\end{multline}

\bigskip

To conclude the proof we sum (\ref{cas13}), (\ref{cas24}) and (\ref{cas29}) and use the following estimates that we obtain from (\ref{eq3:confine1}) and (\ref{estimefonda}):
\begin{eqnarray*}
\sum_{p,q\in\NN_{1/2}}\ \left|\left\langle f_p,g_{1,p} \right\rangle\right|^2 \left|\left\langle f_q,g_{1,q} \right\rangle\right|^2  |p-q|^2 &\leq& 
C \left(\sum_{p,q\in\NN_{1/2} } \left|\left\langle f_p,g_{1,p} \right\rangle\right|^2 \left|\left\langle f_q,g_{1,q} \right\rangle\right|^2  |p-n^*|^2 \right. \\ 
&+&  \left. \sum_{p,q\in\NN_{1/2} } \left|\left\langle f_p,g_{1,p} \right\rangle\right|^2 \left|\left\langle f_q,g_{1,q} \right\rangle\right|^2 |q-n^*|^2 \right) \\ &\leq& 
C \sum_{p\in\NN_{1/2} } \Vert f_p \Vert_{L^2(\R ^+,rdr )}^2 |p-n^*|^2 \leq C
\end{eqnarray*}
and 
\begin{eqnarray*}
\sum_{p,q\in \NN_{1/2}} \left|\left\langle f_p,g_{1,p} \right\rangle\right|^2 \left|\left\langle f_q,g_{1,q} \right\rangle\right|^2  |p-q| &\leq& \left(\sum_{p,q\in \NN_{1/2}} \Vert f_p \Vert_{L^2(\R ^+,rdr )}^2 \Vert f_q \Vert_{L^2(\R ^+,rdr )}^2  |p-q|^2 \right)^{1/2} \\
&\times &\left(
\sum_{p,q\in \NN_{1/2}} \Vert f_p \Vert_{L^2(\R ^+,rdr )}^2 \Vert f_q \Vert_{L^2(\R ^+,rdr )}^2 \right)^{1/2} \leq C.
\end{eqnarray*}
Also,
\[
 \sum_{q\notin \NN_{1/4,+}} \left|\left\langle f_q,g_{1,q} \right\rangle\right|^2 \leq C\om^{-1/2}.
\]
Finally
\[
 \sum_{n\in \NN_{1/2}} \left|\left\langle f_p,g_{1,p} \right\rangle\right|^2 = 1 - C\om^{-1}.
\]

 \end{proof}

\subsection{Proofs of Theorems \ref{theo:densite/energie2} and \ref{theo:vortex}}

\begin{proof}[Proof of (\ref{eq0:resultenergie2})]
The upper bound is obtained by taking $\Psi_{n^*} e^{i n^* \theta}$ as a test function in $F_{\om}$. For the lower bound we first write, using Proposition \ref{theo:refinedlowerbound} 
\begin{eqnarray}\label{minoration0}
F_{\om}(u_{\om}) &\geq& \sum_{n\in \Z} F_n (f_n) + \sum_{n\in \NN_{1/4,+}} 2\pi G \left|\left\langle f_n,g_{1,n} \right\rangle\right|^2 \int_{\R^+} \left(g_{1,n}+\om^{-1/2}\Gamma_n \right)^4 rdr \nonumber
\\ &+& \sum_{n\in \NN_{1/2} \setminus \NN_{1/4,+}} 2\pi G \left|\left\langle f_n,g_{1,n} \right\rangle\right|^2 \int_{\R^+} g_{1,n}^4 rdr -C_{\ep} \om^{-1/4+\ep}.
\end{eqnarray}
Now, using Lemma \ref{theo:lemme1}, (\ref{fnRafinenergie}) and the definitions of $\lambda_{1,n}$, $\lambda_{2,n}$ and $E_n$ we obtain from (\ref{minoration0})
\begin{eqnarray}\label{minoration1}
F_{\om} (u_{\om}) & \geq& \sum_{n\in \NN_{1/4,+}} \left|\left\langle f_n,g_{1,n} \right\rangle\right|^2  \left( F_n\left(g_{1,n} +\om^{-1/2}\Gamma_n \right) + 2\pi  G \int_{\R^+} \left(g_{1,n}+\om^{-1/2}\Gamma_n \right)^4 rdr \right) \nonumber
\\ &+& \sum_{n\in \NN_{1/2} \setminus \NN_{1/4,+}}  \left|\left\langle f_n,g_{1,n} \right\rangle\right|^2 \left( F_n (g_{1,n}) + 2\pi G  \int_{\R^+} g_{1,n}^4 rdr \right) + \int_{\R^+} \left(f_n-\left\langle f_n,g_{1,n}\right\rangle g_{1,n}\right)^2rdr \lambda_{2,n} \nonumber
\\ &+& \sum_{n\in \NN_{1/2}^c}  2\pi \Vert f_n \Vert_{L^2(\R^+, rdr)} ^2 \lambda_{1,n}  -C_{\ep} \om^{-1/4+\ep} \nonumber
\\ &\geq& \sum_{n\in \NN_{1/4,+}}  \left|\left\langle f_n,g_{1,n} \right\rangle\right|^2  E_n\left(g_{1,n}+\om^{-1/2}\Gamma_n\right) + \sum_{n\in \NN_{1/2} \setminus \NN_{1/4,+}} \Vert f_n \Vert_{L^2(\R^+, rdr)} ^2 E_n\left(g_{1,n}\right)  \nonumber
\\&+& \sum_{n\in \NN_{1/2}^c}  2\pi \Vert f_n \Vert_{L^2(\R^+, rdr)} ^2 c\om  -C_{\ep} \om^{-1/4+\ep}
\end{eqnarray}
where $c$ is the constant appearing in Lemma \ref{theo:lemme1}. We have used the following :
\[
 \lambda_{2,n} \geq \lambda_{1,n} + C\om \geq E_n(g_{1,n}) +C\om -C' \om^{1/2} \geq E_n(g_{1,n})
\]
which is proved by computing $E_n(g_{1,n})$ with the results of Theorem \ref{theo:1D1} and Proposition \ref{theo:1D2} in mind. The normalization of $g_{1,n}$ and the orthogonality of $g_{1,n}$ and $\Gamma_n$ imply
\[
2\pi \left\Vert g_{1,n} +\om^{-1/2} \Gamma_n  \right\Vert _{L^2(\R^+, rdr )} ^2 \geq 1
\]
so 
\begin{equation}\label{minoration2}
 E_n\left(g_{1,n}+\om^{-1/2}\Gamma_n\right) \geq 2\pi \left\Vert g_{1,n} +\om^{-1/2} \Gamma_n  \right\Vert _{L^2(\R^+, rdr )} ^2 \gamma_n.
\end{equation}
On the other hand, because of (\ref{fnRafin}) we have 
\begin{equation}\label{minoration3}
 2 \pi\Vert f_n \Vert_{L^2(\R^+, rdr)} ^2 = 2\pi \left|\left\langle f_n,g_{1,n} \right\rangle\right|^2\left\Vert g_{1,n} +\om^{-1/2} \Gamma_n  \right\Vert _{L^2(\R^+, rdr )} ^2 + \beta_n ^2 \om^{-3/2}
\end{equation}
for any $n\in\NN_{1/4,+}$. Using that $\gamma_n = O(\om)$ we finally obtain from (\ref{minoration1}), (\ref{minoration2}) and (\ref{minoration3})
\begin{equation}\label{minorationfinale}
F_{\om}(u_{\om})\geq \sum_{n\in\NN_{1/2}} 2\pi \Vert f_n \Vert_{L^2(\R^+, rdr)} ^2 \gamma_n + \sum_{n\in \NN_{1/2}^c}  2\pi \Vert f_n \Vert_{L^2(\R^+, rdr)} ^2 \left(\gamma_{n^*}+ C \om \right)  -C_{\ep} \om^{-1/4+\ep}.
\end{equation}
Note that for $n\notin \NN_{1/2}$ one has $ \lambda_{1,n} \geq c\om \geq \gamma_{n^*} + C\om $ because $c> \sqrt{6}$ and $\gamma_{n^*}\leq \sqrt{6}\om $ which can be computed from (\ref{eq4:DLenergie}). The lower bound in (\ref{eq0:resultenergie2}) follows from (\ref{minorationfinale}) and Corollary \ref{cor:lambdan/n^*NL}, recalling
\[
\sum_{n\in \Z} 2\pi \Vert f_n \Vert_{L^2(\R^+, rdr)} ^2 = 1.
\]
\end{proof}

\begin{proof}[Proof of (\ref{eq0:unseulmode})]
Using
\[
 F_{\om}(u_{\om}) \leq \gamma_{n^*}=\sum_{n\in \Z} \gamma_{n^*} 2\pi \Vert f_n \Vert_{L^2(\R^+, rdr)} ^2 
\]
we obtain from (\ref{minorationfinale})
\[
 C_{\ep}\om^{-1/4+\ep} \geq \sum_{n\in\NN_{1/2}} 2\pi \Vert f_n \Vert_{L^2(\R^+, rdr)} ^2 \left(\gamma_n -\gamma_{n *}\right)+ \sum_{n\in \NN_{1/2}^c}  2\pi \Vert f_n \Vert_{L^2(\R^+, rdr)} ^2  C \om  
\]
with $C>0$, thus, Corollary \ref{cor:lambdan/n^*NL} implies
\[
C_{\ep} \om^{-1/4+\ep} \geq  \sum_{n\in \NN_{1/2},n\neq n^*, n^*+1} 2\pi \Vert f_n \Vert_{L^2(\R^+, rdr)} ^2 |n-n^*|^2 + \sum_{n\in \NN_{1/2}^c}  2\pi \Vert f_n \Vert_{L^2(\R^+, rdr)} ^2  C \om
\]
This immediately yields 
\begin{equation}\label{confinefinal'}
 \sum_{n\in \NN_{1/2},n\neq n^*,n^*+1 } 2\pi \Vert f_n \Vert_{L^2(\R^+, rdr)} ^2 |n-n^*|^2 \leq C_{\ep}\om^{-1/4+\ep}.
\end{equation}
and 
\begin{equation}\label{confinefinalbis}
 \sum_{n\in \NN_{1/2}^c} 2\pi \Vert f_n \Vert_{L^2(\R^+, rdr)} ^2  \leq C_{\ep}\om^{-5/4+\ep}.
\end{equation}
We now introduce 
\begin{equation}\label{upoint}
\dot{u}_{\om} = \alpha \Psi_{n^*} e^{in^*\theta}+ \beta  \Psi_{n^*+1} e^{i(n^*+1)\theta}
\end{equation}
where 
\begin{eqnarray}\label{alphabeta}
&\alpha& : = \sqrt{2\pi}\Vert f_{n^*}\Vert_{L^2 (\R^+,rdr)} \nonumber \\
&\beta&  : = \sqrt{2\pi} \Vert f_{n^*+1}\Vert_{L^2(\R^+,rdr)}. 
\end{eqnarray}
We note that (\ref{eq0:DLfonction}) and (\ref{eq4:DLfonction}) imply
\begin{equation}\label{gn-Psin}
g_{1,n} = \Psi_n + O_{L^2 (\R^2)} (\om ^{-1/2}) 
\end{equation}
so, using equation (\ref{confinefinal'}) and (\ref{confinefinalbis}) combined with (\ref{fnRafin}) we see that after extraction of a subsequence
\[
 \Vert u_{\om}-c\dot{u}_{\om} \Vert_{L^2(\R^2)}\leq C_{\ep} \om^{-1/8+\ep}
\]
with $c$ a constant of modulus $1$, so that we have, recalling (\ref{eq4:LpPsin}) and (\ref{L6u}) 
\begin{multline}\label{uupointL4}
 \int_{\R^2} \left||u_{\om}|^4 - |\dot{u}_{\om}|^4 \right| \leq C\left(\Vert u_{\om}\Vert_{L^6(\R^2)}^3+\Vert \dot{u}_{\om}\Vert_{L^6(\R^2)}^3\right)\Vert u_{\om}-c \dot{u}_{\om} \Vert_{L^2(\R^2)} \\ \leq C_{\ep} \om^{1/2+\ep} \Vert u_{\om}-c \dot{u}_{\om} \Vert_{L^2(\R^2)}\leq C_{\ep}\om^{3/8+\ep}.
\end{multline}
The estimate (\ref{eq4:DLfonction}) yields by a change of variables similar to those we used in the proofs of Propositions \ref{theo:borneinfint} and \ref{theo:refinedlowerbound}
\[
\int_{\R^2} |\Psi_{n^*}|^4 = \int_{\R^2} |\Psi_{n^*+1}|^4 + O(\om^{-1/2})= \int_{\R^2} |\Psi_{n^*}|^2 |\Psi_{n^*+1}|^2 +O(\om^{-1/2})
\]
Thus, computing the interaction energy of $\dot{u}_{\om}$, using (\ref{uupointL4}) and the fact that $\alpha^2 + \beta ^2 = 1+O(\om^{-1/8+\ep})$ for any $\ep>0$ we have
\begin{eqnarray*}
 \gamma_{n^*} &\geq& F_{\om} (u_{\om}) \geq \gamma_{n^*} + 2\alpha^2 \beta^2 \int_{\R^2}|\Psi_{n^*}|^ 4 -C_{\ep}\om^{3/8+\ep}\\
C_{\ep}\om^{3/8+\ep} &\geq&  2\alpha^2 \beta^2 \int_{\R^2}|\Psi_{n^*}|^ 4 \geq C \om ^{1/2} \alpha^2 \beta ^2.
\end{eqnarray*}
But, according to our results (\ref{confinefinal'}) and (\ref{confinefinalbis}), either $\alpha ^2$ or $\beta ^2$ is bounded below. Then, along some subsequence we have either $\alpha ^2 \leq C_{\ep}\om^{-1/8+\ep}$ or $\beta ^2 \leq  C_{\ep}\om^{-1/8+\ep}$. Renaming $n^*$ if necessary, we have proved that along some subsequence 
\begin{equation}\label{confinefinal}
 \sum_{n\in\NN_{1/2}, n \neq n^*} \Vert f_n \Vert_{L^2(\R^+, rdr)} ^2 \leq C_{\ep} \om^{-1/8+\ep} 
\end{equation}
which concludes the proof of (\ref{eq0:unseulmode}), recalling (\ref{fnRafin}) and (\ref{gn-Psin}).

\end{proof}

Finally, we present the 

\begin{proof}[Proof of Theorem \ref{theo:vortex}]
 
%
%

Let us define
 \begin{equation}\label{N7/16}
 \NN_{7/16,+}:=  \left\lbrace n\in \Z ,\quad \left| n-n^* \right| \leq \om^{7/16}, \: \left\langle f_n , g_{1,n}\right\rangle \neq 0 \right\rbrace .
\end{equation}
and
\begin{equation}\label{uchapeau}
\hat{u}_{\om}(r,\theta):= \sum_{n\in \NN_{7/16,+}} \left\langle f_n, g_{1,n}\right\rangle \phi_n e^{in\theta}
\end{equation}
where we have denoted
\begin{equation}\label{phin}
\phi_n :=  g_{1,n} + \om^{-1/2} \Gamma_n
\end{equation}
for short. We claim that 
\begin{equation}\label{u-uchapeauLinf}
 \Vert u_{\om}- \hat{u}_{\om}\Vert_{L^{\infty}(\R^2)} \leq C_{\ep}\om^{-1/32+\ep}.
\end{equation}
Firstly we obtain
\begin{eqnarray}\label{u-uchapeauL2}
 \Vert u_{\om} - \hat{u}_{\om}\Vert_{L^{2}(\R^2)}^2 &\leq& \sum_{n\in \NN_{7/16,+} ^c} 2\pi \Vert f_n \Vert_ {L^2}^2 + \sum_{n\in \NN_{7/16,+}} \Vert f_n - \left\langle f_n, g_{1,n}\right\rangle \phi_n \Vert_{L^2}^2 \nonumber
\\ &\leq& C_{\ep} \om^{-9/8+\ep} + C \om^{-5/4}\leq C_{\ep} \om^{-9/8+\ep}
\end{eqnarray}
from (\ref{fnRafin}), (\ref{confinefinal'}) and (\ref{confinefinalbis}).
We then note that $\phi_n$ satisfies
\begin{equation}\label{eqphin}
-\Delta \phi_n + V_n \phi_n = \lambda_{1,n}\phi_n + \Pi_{1,n} ^{\perp} \left(-2G \om ^{1/2} g_{1,n^*} ^2 g_{1,n} \right).
\end{equation}
Multiplying (\ref{eqphin}) by $\left\langle f_n, g_{1,n}\right\rangle e^{in\theta}$ and summing over $n\in \NN_{7/16,+}$ we deduce that
\begin{multline}\label{equchapeau}
-\left(\nabla -i\om x^{\perp} \right)^2 \hat{u}_{\om} +V(r) \hat{u}_{\om} = \sum_{n\in \NN_{7/16,+}}  \lambda_{1,n} \left\langle f_n, g_{1,n}\right\rangle \phi_n e^{in\theta} 
\\ + \sum_{n\in \NN_{7/16,+}} \left\langle f_n, g_{1,n}\right\rangle \Pi_{1,n} ^{\perp} \left(-2G \om ^{1/2} g_{1,n^*} ^2 g_{1,n} \right) e^{in\theta}.
\end{multline}
Substracting this equation from the Euler-Lagrange equation (\ref{eq0:EELu}) we get an equation for $u_{\om}-\hat{u}_{\om}$:
\begin{multline}\label{EELu-uchapeau}
-\left(\nabla -i\om x^{\perp} \right)^2 \left(u_{\om} - \hat{u}_{\om}\right) + V(r) \left(u_{\om} - \hat{u}_{\om}\right) = \mu_{\om} u_{\om}  - 2 G |u_{\om}|^2 u_{\om} 
\\ - \sum_{n\in \NN_{7/16,+}}  \lambda_{1,n} \left\langle f_n, g_{1,n}\right\rangle \phi_n e^{in\theta} - \sum_{n\in \NN_{7/16,+}} \left\langle f_n, g_{1,n}\right\rangle \Pi_{1,n} ^{\perp} \left(-2G \om ^{1/2} g_{1,n^*} ^2 g_{1,n} \right) e^{in\theta}.
\end{multline}
The same technique as in Steps 3, 4 and 5 of the proof of Proposition \ref{theo:estimsobolev} allows to deduce (\ref{u-uchapeauLinf}) from (\ref{u-uchapeauL2}) and (\ref{EELu-uchapeau}).

\bigskip

\emph{Proof of Item 1.}

We write, using (\ref{u-uchapeauLinf})
\begin{equation}\label{preuveItem11}
\left||u_{\om}|-|\phi_{n^*}|\right|\leq \left|1-\left|\left\langle f_{n^*},g_{1,n^*} \right\rangle\right| \right| |\phi_{n^*}| + \sum_{n\in \NN_{7/16,+},n\neq n^*} \left|\left\langle f_{n},g_{1,n} \right\rangle\right| \left| \phi_n \right| +C_{\ep}\om^{-1/32+\ep},
\end{equation}
A consequence of Lemma \ref{theo:propgamman} (more precisely the $L^{\infty}$ estimate in (\ref{bornesGamma2})), (\ref{eq0:DLfonction}) and (\ref{eq4:DLfonction}) is
\begin{equation}\label{gn-phin-Psin}
\phi_n = g_{1,n} + O_{L^{\infty}} (\om^{-1/4}) = \Psi_n + O_{L^{\infty}} (\om^{-1/4}) = c_n e^{-d \left(\frac{r-R_n}{h_n}\right)^2} + O_{L^{\infty}} (\om^{-1/4}) 
\end{equation}
where $c_n$ is defined by (\ref{eq4:rescaling}) and $d$ is the standard deviation of $\xi_1$. A consequence of (\ref{eq4:DLfonction}) is
\begin{equation}\label{estimcn}
 c_n ^2 = \frac{1}{2\pi h_n R_n} + O(\om^{-1/2})
\end{equation}
because $\xi_1$ is an even function. We recall that, for some constant $h$
\begin{eqnarray}
R_n &=& 1 + o(\om^{-1/2}) \label{Rnegal}\\
h_n &=& h \om^{-1/2} (1+o(1)) \label{hnegal}
\end{eqnarray}
for any $n\in \NN_{7/16}$, so arguing as in the proof of Proposition \ref{theo:decexpo} one deduce from (\ref{preuveItem11}) that 
\begin{multline}\label{Item11}
\left||u_{\om}|-|\Psi_{n^*}|\right|\leq C \left(\left|1-\left|\left\langle f_{n^*},g_{1,n^*} \right\rangle\right| \right| + \sum_{n\in \NN_{7/16,+},n\neq n^*} \left|\left\langle f_{n},g_{1,n} \right\rangle\right| \right) \om^{1/4} e^{-\sigma \om \left(|x|-1\right)^2} 
\\+C_{\ep}\om^{-1/32+\ep}
\end{multline}
with
\begin{equation}\label{defsigma}
\sigma:=\frac{h ^2}{4d}.
\end{equation}

On the other hand (\ref{confinefinal}) and (\ref{fnRafin}) imply
\[
\left|1-\left|\left\langle f_{n^*},g_{1,n^*} \right\rangle\right| \right| \leq C_{\ep}\om^{-1/8+\ep}
\]
and
\begin{eqnarray*}
\sum_{n\in \NN_{7/16,+},n\neq n^*} \left|\left\langle f_{n},g_{1,n} \right\rangle\right| &\leq& C \left(\sum_{n\in \NN_{7/16,+},n\neq n^*}  \Vert f_n \Vert_{L^2(\R^+, rdr)} ^2 |n-n^{*}|^2 \right) ^{1/2} \left(\sum_{n\in \NN_{7/16,+},n\neq n^*} |n-n^*| ^{-2} \right) ^{1/2}
\\ &\leq& C_{\ep}\om^{-1/16+\ep},
\end{eqnarray*}
so we obtain (\ref{eq0:estimunif}) from (\ref{Item11}).

\bigskip

\emph{Proof of Items 2 and 3.} These final results are consequences of (\ref{eq0:estimunif}). Using the facts that 
\[
 \Psi_{n^*} = c_{n^*} e^{-d \left(\frac{r-R_{n^*}}{h_{n^*}}\right)^2} + O_{L^{\infty}} (\om^{-1/4}) 
\]
and $c_{n^*} \propto \om^{1/4}$ one can see from (\ref{eq0:estimunif}) that $u_{\om} (x_{\om}) \rightarrow 0$ uniformly on any region where
\[
 \left| |x_{\om}| - 1 \right| \geq \delta \om^{-1} \ln (\om)
\]
with $\delta > \sigma$. On the other hand, if $x_{\om}=r_{\om}e^{i\theta_{\om}}$ is a point such that
\[
 u_{\om}(x_{\om})\rightarrow 0 
\]
as $\om\rightarrow \infty$, 
(\ref{eq0:estimunif}) implies
\[
\om^{1/4}e^{-dh_{n ^*}^{-2} \left(r-R_{n^*}\right)^2} \leq C
\]
for any constant $C$ as $r\rightarrow r_{\om}$ and thus 
\begin{equation}\label{eq4:horsanneau}
 |r_{\om}- R_{n ^*}|^2 \geq \frac{1}{4d} h_{n ^*}^{-2}\left(\ln (\om) - C \right) \geq \sigma \om^{-1} (\ln (\om)-C)
\end{equation}
which completes the proof because 
\[
R_{n ^*} = 1 + O(\om^{-1}).
\]
\end{proof}

\bigskip

\textbf{Acknowledgments}\\

I express my gratitude to Xavier Blanc and Sylvia Serfaty for their constant support, their careful reading of the manuscript and their many useful suggestions. This work was supported by the R\'{e}gion Ile-de-France through a PhD grant.


\nocite{*}


\begin{thebibliography}{99}
 
\bibitem{Aft}
Aftalion A, \textbf{Vortices in Bose-Einstein Condensates}, Birkh\"{a}user, Boston, (2006)

\bibitem{AAB}
Aftalion A, Alama S, Bronsard L, \textit{Giant vortex and the breakdown of strong pinning  in a
rotating Bose-Einstein condensate},  Archive for Rational Mechanics and Analysis 178, 247-286 (2005)

\bibitem{AB} 
Aftation A, Blanc X, \textit{Vortex lattices in rotating Bose-Einstein condensates}, SIAM Journal of Mathematical Analysis 38, 874 (2006)

\bibitem{AB-2D} 
Aftalion A, Blanc X, \textit{Reduced energy functionals for a three dimensional fast rotating Bose-Einstein condensate}, Annales de l'Institut Henri Poincare (C) Non Linear Analysis 339-355 (2008)

\bibitem{ABN} 
Aftalion A, Blanc X, Nier F, \textit{Lowest Landau Level functionals and Bargmann transform in Bose-Einstein condensates}, Journal of Functional Analysis 241(2), 661-702 (2006)

\bibitem{AD}
Aftalion A, Danaila I, \textit{Giant vortices in combined harmonic and quartic traps}, Phys. Rev. A 69, 033608 (2004)

\bibitem{AH} 
Aftalion A, Hellfer B, \textit{On mathematical models for Bose-Einstein condensates in optical lattices},  Reviews in Mathematical Physics 21(2), 229-278 (2009)   	

\bibitem{BPT}
Bauman P, Phillips D, Tang Q \textit{Stable nucleation for the Ginzburg-Landau system with an applied magnetic field}, Archive for Rational Mechanics and Analysis 142, 1-43 (1998)

\bibitem{Moi}
Blanc X, Rougerie N, \textit{Lowest-Landau-level vortex structure of a Bose-Einstein condensate rotating in a harmonic plus quartic trap}, Physical Review A 77, 053615 (2008)

\bibitem{BBH1}
B\'{e}thuel F, Br\'{e}zis H, H\'{e}lein F, \textit{Asymptotics for the minimization of a Ginzburg-Landau functional}, Calculus of Variations and Partial Differential Equations 1, 123-148 (1993)

\bibitem{BBH}
B\'{e}thuel F, Br\'{e}zis H, H\'{e}lein F, \textbf{Ginzburg-Landau vortices}, Birkh\"{a}user, Boston, (1994)

\bibitem{Exp1}
Bretin V, Stock S, Seurin S, Dalibard J,  \textit{Fast Rotation of a Bose-Einstein Condensate}, Phys. Rev. Lett. 92, 050403 (2004)

\bibitem{Bre} 
Br\'{e}zis H,  \textit{Semilinear Equations in $\R^{N}$}, Applied Mathematics and Optimization 12, 271-282 (1984)

\bibitem{BCPY} Bru JB, Correggi M, Pickl P, Yngvason J, \textit{The TF Limit for Rapidly Rotating Bose Gases in Anharmonic Traps}, Communications in Mathematical Physics 280, 517-544 (2008) 

\bibitem{CRY1}
Corregi M, Rindler-Daller T, Yngvason J,  \textit{Rapidly rotating Bose-Einstein condensates in strongly anharmonic traps}, Journal of Mathematical Physics 48, 042104 (2007)

\bibitem{CRY2}
Corregi M, Rindler-Daller T, Yngvason J,  \textit{Rapidly rotating Bose-Einstein condensates in homogeneous traps}, Journal of Mathematical Physics 48, 102103 (2007)

\bibitem{CY}
Corregi M, Yngvason J,  \textit{Energy and vorticity in fast rotating Bose-Einstein condensates}, Journal of Physics A 41(44), 445002 (2008)

\bibitem{Nous}
Correggi M, Rougerie N, Yngvason J, in preparation.

\bibitem{D}
Danaila I,  \textit{Three-dimensional vortex structure of a fast rotating Bose-Einstein condensate with harmonic-plus-quarting confinement}, Phys. Rev. A 72, 013605 (2005)

\bibitem{Far}
Farina A,  \textit{From Ginzburg-Landau to Gross-Pitaevskii}, Monatsh. Math. 139, 265-269 (2003)

\bibitem{Fetterquart}
Fetter A L,  \textit{Rotating vortex lattice in a Bose-Einstein condensate trapped in combined quadratic and quartic radial potentials}, Physical Review A 64, 063608 (2001)

\bibitem{Fetter}
Fetter A L,  \textit{Rotating trapped Bose-Einstein condensates}, Reviews of Modern Physics 81, 647 (2009) 

\bibitem{FJS}
Fetter A L, Jackson B, Stringari S,  \textit{Rapid rotation of a Bose-Einstein condensate in a harmonic plus quartic trap}, Phys. Rev. A 71, 013605 (2005)

\bibitem{FB}
Fischer U W, Baym G \textit{Vortex states of rapidly rotating dilute Bose-Einstein condensates}, Physical Review Letters 90, 140402 (2003)

\bibitem{FuZa}
Fu H, Zaremba E,  \textit{Transition to the giant vortex state in a harmonic-plus-quartic trap}, Physical Review A 73, 013614 (2006)


\bibitem{HeSjo}
Helffer N, Sj\"{o}strand J,  \textit{Puits multiples en m\'{e}canique semi-classique 6. Cas des puits sous-vari\'{e}t\'{e}s}, Annales de l'Institut Henri Poincar\'{e} (A) 46(4), 353-372 (1987)

\bibitem{Ign-Mil1} Ignat R, Millot V, \textit{The critical velocity for vortex existence in a two dimensional rotating Bose-Einstein condensate}, Journal of Functional Analysis, 260-306 (2006) 

\bibitem{Ign-Mil2} Ignat R, Millot V, \textit{Energy expansion and vortex location for a two dimensional rotating Bose-Einstein condensate}, Review of Mathematical Physics 18, 119-162 (2006)

\bibitem{JaKa}
Jackson A D, Kavoulakis G M,  \textit{Vortices and hysteresis in a rotating Bose-Einstein condensate with anharmonic confinement}, Phys. Rev. A 70, 023601 (2004)

\bibitem{JaKaLu}
Jackson A D, Kavoulakis G M, Lundh E,  \textit{Phase diagram of a rotating Bose-Einstein condensate with anharmonic confinement}, Phys. Rev. A 69, 053619 (2004)

\bibitem{KTU}
Kasamatsu K, Tsubota M, Ueda M,  \textit{Giant hole and circular superflow in a fast rotating Bose-Einstein condensate}, Phys. Rev. A 66, 050606 (2002)

\bibitem{KB}
Kavoulakis G M, Baym G,  \textit{Rapidly rotating Bose-Einstein condensates in anharmonic potentials}, New Journ. Phys. 5, 51.1 (2003)

\bibitem{KiFe}
Kim J K, Fetter A L,  \textit{Dynamics of a rapidly rotating Bose-Einstein condensate in a harmonic plus quartic trap}, Physical Review A 72, 023619 (2005)

\bibitem{LeSe}
Lewin M, Seiringer R, \textit{Strongly correlated phases in rapidly rotating Bose gases}, to appear in 
Journal of Statistical Physics, arXiv:0906.0741 

\bibitem{LiLo} 
Lieb EH, Loss M, \textbf{Analysis}, American Mathematical Society (1997)

\bibitem{LS} 
Lieb EH, Seiringer R, \textit{Derivation of the Gross-Pitaevskii equation for rotating Bose gases}, Communications in Mathematical Physics 103,  (2005)

\bibitem{LSY} 
Lieb EH, Seiringer R, Solovej JP, Yngvason J, \textbf{The Mathematics of the Bose Gas and its Condensation}, Oberwolfach Seminar Series, Vol. 34, Birkh\"{a}user, Boston (2005)

\bibitem{LuPan}
Lu K, Pan X B,  \textit{Eigenvalue problems of Ginzburg-Landau operator in bounded domains}, Journal of Mathematical Physics 40, 2647-2670 (1999)

\bibitem{Lu}
Lundh E,  \textit{Multiply quantized vortices in trapped Bose-Einstein condensates}, Phys. Rev. A 65, 043604 (2002)

\bibitem{Nir}
Nirenberg L, \textit{On elliptic partial differential equations}, Annali Della Scuola Normale Superiore di Pisa 13, p 115-162 (1959) 

\bibitem{San-Ser}
Sandier E, Serfaty S, \textbf{Vortices in the magnetic Ginzburg-Landau model}, Birkh\"{a}user, Boston, (2007)

\bibitem{S}
Seiringer R,  \textit{Gross-Pitaevskii theory of the rotating Bose gas}, Communication in Mathematical Physics 229, 491-509 (2002)

\bibitem{Exp2}
Stock S, Bretin V, Chevy F, Dalibard J,  \textit{Shape oscillation of a rotating Bose-Einstein condensate}, Europhys. Lett. 65, 594 (2004)

\end{thebibliography}
\end{document}